\documentclass{article}
\usepackage[T1]{fontenc}
\usepackage[utf8]{inputenc}
\usepackage{amsmath, amssymb, amsfonts, mathtools, bm}
\usepackage{graphicx}
\usepackage{booktabs}
\usepackage{pdflscape}
\usepackage{subcaption}
\usepackage{float}
\usepackage{microtype}
\usepackage{nicefrac}
\usepackage{enumitem}
\usepackage{xcolor}
\usepackage{float}
\usepackage{url}
\usepackage{listings}
\usepackage[numbers]{natbib}
\usepackage{hyperref}
\usepackage{comment}
\usepackage{caption}
\usepackage{optidef}
\usepackage[
    left=1.2in,
    right=1.2in,
    top=1in,
    bottom=1in
]{geometry}
\definecolor{codegreen}{rgb}{0,0.6,0}
\definecolor{codegray}{rgb}{0.5,0.5,0.5}
\definecolor{codepurple}{rgb}{0.58,0,0.82}
\definecolor{backcolour}{rgb}{0.95,0.95,0.92}
\newtheorem{theorem}{Theorem}
\newtheorem{lemma}{Lemma}

\newtheorem{proposition}{Proposition}
\newtheorem{corollary}{Corollary}
\newenvironment{proof}{\paragraph{Proof.}}{\hfill$\square$\par}
\newcommand{\E}{\mathbb E}
\newcommand{\Prob}{\mathbb P}
\newcommand{\Cov}{\operatorname{Cov}}
\newcommand{\Var}{\operatorname{Var}}
\newcommand{\sign}{\operatorname{sign}}
\newcommand{\supp}{\operatorname{supp}}
\usepackage{algorithm,algorithmic}

\lstdefinestyle{mystyle}{
    backgroundcolor=\color{backcolour},   
    commentstyle=\color{codegreen},
    keywordstyle=\color{magenta},
    numberstyle=\tiny\color{codegray},
    stringstyle=\color{codepurple},
    basicstyle=\ttfamily\footnotesize,
    breakatwhitespace=false,         
    breaklines=true,                 
    captionpos=b,                    
    keepspaces=true,                 
    numbers=left,                    
    numbersep=5pt,                  
    showspaces=false,                
    showstringspaces=false,
    showtabs=false,                  
    tabsize=2
}

\title{Univariate-Guided Interaction Modeling}
\author{
Aymen Echarghaoui\\Department of Statistics, Stanford University. \\ Email: aymen20@stanford.edu
 \and 
 Robert Tibshirani \\Departments of Biomedical Data Science and Statistics, \\Stanford University.\\
 Email: tibs@stanford.edu.}

\begin{document}

\maketitle

\begin{abstract}
We propose a procedure for sparse regression with pairwise interactions, by generalizing the Univariate Guided Sparse Regression (\textsc{UniLasso}) methodology. 
A central contribution is our introduction of \textsc{TripletScan}, which screens
a pair $(j,k)$ using the coefficient of $X_jX_k$ in the local regression of
the response on $1$, $X_j$, $X_k$, and $X_jX_k$. The retained products are
incorporated either jointly with the main effects through \textsc{UniLasso},
yielding \texttt{uniPairs}, or after a first-stage main-effects fit, yielding
\texttt{uniPairs-2stage}. For the \textsc{UniLasso} components of the
procedures, we prove false-positive exclusion and uniform coefficient-error
bounds. In simulations and an HIV drug-resistance application, the proposed procedures
produce substantially smaller fitted models than competing interaction methods
while retaining competitive predictive performance.
\end{abstract}

\section{Introduction}
\label{sec:introduction}
Interactions allow the effect of one predictor to vary with the value of
another and are therefore important in diverse applications.
In genetics, interactions between genes (epistasis) can reveal mechanisms responsible for complex traits. In medicine, the simultaneous presence of two symptoms may enhance (positive interaction) or cancel (negative interaction) diagnostic information. In recommender systems, user-item interactions are important for personalization.

We consider the problem of modeling pairwise interactions between features where the observed response vector \( \mathbf Y \in \mathbb{R}^n \) follows the model
\begin{equation}
Y_i=\beta_0^*+\sum_{j=1}^p \beta_j^* X_{ij}
    +\sum_{1\le j<k\le p}\gamma_{jk}^* X_{ij}X_{ik}+\varepsilon_i,
\label{eq:intro-model}
\end{equation}
with \( \mathbf X=(X_{ij})\in \mathbb{R}^{n \times p} \) denoting the observed design matrix of main effects, and \( \varepsilon \) being random noise. 
The goal is to identify a sparse subset of both main effects and interactions that has strong predictive power with respect to \( Y \), particularly in the high-dimensional regime \( p \gg n \). 
The augmented model contains $p+\binom p2$ candidate
terms. Directly fitting a regularized model to this expanded design can
therefore be expensive and can admit many spurious terms.

For example, the All Pairs Lasso (\texttt{APL})\footnote{This method adds all pairwise products to a linear model, and fits the model using the lasso.} requires \( O(np^2) \) memory and performs multiple passes through coordinate descent on \( O(p^2) \) interaction terms.
\texttt{APL}  becomes computationally expensive as $p$ grows, and tends to favor spurious interaction terms over true main effects. 

One approach to this problem is to impose \emph{hierarchy}: an interaction may
enter only when at least one of its associated main effects enters (weak
hierarchy), or when both do (strong hierarchy). \textsc{HierNet}
\citep{bien2013lasso} incorporates hierarchy through convex constraints on the
main-effect and interaction coefficients. Its symmetric formulation enforces
strong hierarchy, while an asymmetric formulation gives weak hierarchy.
\textsc{Glinternet} \citep{lim2013learning} uses a hierarchical group-lasso
construction that enforces strong hierarchy. A second approach is to screen before fitting the final interaction model.
\textsc{Sprinter} \citep{yu2019reluctant} implements a
reluctance principle: it first regresses out  main effects, then scans all candidate interaction terms for their correlation with the resulting residuals, and finally regresses the residuals on the selected interactions together with the main effects.

Our proposal asks a different local question. For each pair, it measures whether
the product contributes to the response after adjusting for the two constituent
main effects. The resulting score is in the spirit of conditional
sure independence screening \citep{csis}. Unlike the \textsc{Sprinter} screen,
\textsc{TripletScan} does not require a first high-dimensional
main-effects fit.

This work is based on the Univariate Guided Sparse Regression (UniLasso) framework introduced in \citet{chatterjee2024unilasso}, and extends the idea of univariate guidance to both main effects and interaction terms. Our primary objectives are to :

\begin{itemize}
    \item achieve competitive prediction error while producing sparser models than existing methods such as \texttt{Sprinter}, \texttt{HierNet}, and \texttt{Glinternet},
    \item identify pairwise interactions that have strong marginal interaction after adjustment for their two constituent main effects,
    \item design an algorithm whose complexity is sub-quadratic in \( p \), enabling scalability to high-dimensional settings where \( p \gg n \), or which can be parallelized on modern hardware to avoid the quadratic cost.
\end{itemize}

Using \textsc{TripletScan}, we propose two procedures. \texttt{uniPairs}
applies \textsc{UniLasso} jointly to the main effects and the retained
interactions, whereas \texttt{uniPairs-2stage} first estimates the main
effects and then fits the retained interactions to the resulting residuals.
For the \textsc{UniLasso} components of the procedures, we prove
false-positive exclusion and uniform coefficient-error bounds. Our simulation
study and HIV drug-resistance application compare the predictive performance,
sparsity, and interaction-selection behavior of the proposed procedures with
existing methods.

This paper is organized as follows.
Section~\ref{sec:proposed} introduces our Univariate-guided procedures for interaction modeling and describes the high level ideas behind \texttt{uniPairs} and \texttt{uniPairs-2stage}.
Section~\ref{sec:hiv} presents a motivating application to HIV-mutation data.
Section~\ref{sec:details-proposed} gives the full algorithmic details, including the TripletScan screening step and the UniLasso/Lasso fits.
Section~\ref{sec:related-work} summarizes the main ideas from related work on interaction modeling including \texttt{Sprinter}, \texttt{Glinternet}, \texttt{HierNet} and Group Lasso approaches.
Section~\ref{sec:HIVmore} shows more findings on the HIV-mutation data.
Section~\ref{sec:simulation} reports the results of a simulation study comparing \texttt{uniPairs} and \texttt{uniPairs-2stage}, to the existing methods \texttt{Sprinter} and \texttt{Glinternet}.
Section~\ref{sec:unilasso-theory} states our theoretical results about support recovery and \(\ell_{\infty}\) estimation error of the coefficients in the UniLasso step for both methods, extending the UniLasso analysis in \citet{chatterjee2024unilasso}, with proofs deferred to \hyperref[proof:theorem1]{Appendix~A}.
Section~\ref{sec:triplet-scan-theory} explains the statistical motivation behind our screening rule, and its connection with conditional sure independence screening.
We conclude in Section~\ref{discussion} with a brief summary and a discussion of possible future work.
\hyperref[sec:glm-generalization]{Appendix~B} shows how to extend \texttt{uniPairs} and \texttt{uniPairs-2stage} to the Binomial generalized linear model and the Cox proportional hazards model.

\section{Our Proposed Algorithms}
\label{sec:proposed}
\subsection{Setup}
\label{sec:notation}

Let \(\mathbf X=(X_{ij})\in\mathbb R^{n\times p}\) be the observed design
matrix and let \(\mathbf Y=(Y_1,\ldots,Y_n)^\top\in\mathbb R^n\) be the
observed response vector. For \(j\in[p]\), let \(\mathbf X_j\) denote the
\(j\)th column of \(\mathbf X\). For \(j<k\), the corresponding interaction
column is \(\mathbf X_j\odot\mathbf X_k\in\mathbb R^n\). Let
\[
\mathcal{P}=\{(j,k)\in [p]^2\ \,|\,\  j<k\}
\quad \text{and} \quad
\mathbf Z=\big(\mathbf X_j\odot\mathbf X_k\big)_{(j,k)\in\mathcal{P}}
\in\mathbb{R}^{n\times\binom{p}{2}}.
\]

To define the population target of the screening step, let \((X,Y)\) denote
a generic observation, where \(X=(X_1,\ldots,X_p)^\top\in\mathbb R^p\).
For each \((j,k)\in\mathcal P\), write \(Q_{j,k}=X_jX_k\). Let
\(\Pi_{j,k}Q_{j,k}\) be the population linear projection of \(Q_{j,k}\) on
\(\operatorname{span}(1,X_j,X_k)\), and let
\[
R_{j,k}=Q_{j,k}-\Pi_{j,k}Q_{j,k},
\qquad
\tau_{j,k}=\frac{\Cov(Y,R_{j,k})}{\Var(R_{j,k})}
      =\frac{\E(R_{j,k}Y)}{\E(R_{j,k}^2)}.
\]
The second equality holds because the projection space contains the constant
function, and hence \(\E(R_{j,k})=0\).
The local interaction support is
\(S_L=\{(j,k)\in\mathcal P:\tau_{j,k}\ne0\}\). This is
the target of the population version of the triplet regression in
\eqref{eqn:marginalpairwise}. The global support
\(S_I=\operatorname{supp}(\gamma^*)\) is defined by the joint model in \eqref{eq:intro-model}.
They coincide for independent centered Gaussian covariates, but not in
general.

\subsection{The high level idea}
\label{sec:high-level}

The steps of our {\tt uniPairs} procedure are as follows:

\begin{description}
\item{(a)} For each $j,k$, fit a least squares model of $\mathbf Y$
on the triplet $(\mathbf X_j, \mathbf X_k, \mathbf X_j\odot \mathbf X_k)$ and measure the contribution to the fit due to $\mathbf X_j\odot \mathbf X_k$.  Retain the pairs that have contributions greater than some data-adaptive threshold, using either the largest log-gap rule or its rank-normalized version, which can also retain nothing. Note that the number of retained pairs is not an added hyperparameter, but is data-adaptive.
\item{(b)} Apply the uniLasso algorithm with target $\mathbf Y$ to all individual features and the pairs that pass the screen in step (a).
\end{description}
We think of the feature pairs that pass the screen in (a) as displaying {\em  marginal interaction}.
If they are chosen for the final model in step (b), these interactions are more credible and interpretable 
than those pairs with weak marginal interaction.
  
Our {\tt uniPairs-2stage} procedure is very similar, except that we fit main effects as a first stage:

\begin{description}
\item{(a)} Apply UniLasso to the individual features and compute the residual $\mathbf R$.
\item{(b)} For each $j,k$, fit a least squares model of $\mathbf Y$ on
the triplet $(\mathbf X_j, \mathbf X_k, \mathbf X_j\odot \mathbf X_k)$ and measure the contribution to the fit due to $\mathbf X_j\odot \mathbf X_k$.
\item{(c)} Apply the Lasso with target $\mathbf R$ to  the pairs that pass the screen in step (b).

The final model is the sum of the two models obtained from steps (a) and (c).
\end{description}

Notice that the triplet regression in both procedures
uses $\mathbf Y$ as the target. Perhaps surprisingly in the
{\tt uniPairs-2stage} where  the residual
$\mathbf R$ might seem a more appropriate target.
The reason is that we want to find interaction pairs for our multivariate model that also display marginal interactions with $\mathbf Y$.

We call our procedure ``Univariate guided'' because the main effects estimation
in Step (a) uses univariate guidance, and the contribution of the $j,k$
interaction is measured by  the linear covariance between the response and the interaction conditional on its two main effects (Section~\ref{sec:triplet-scan-theory}). 

Before giving details of our proposed method, we show a motivating example.

\subsection{Example: HIV mutation data}
\label{sec:hiv}
As an example, \cite{rhee2003}  studied six nucleoside reverse transcriptase inhibitors that are used to treat HIV-1. The target of these drugs can become resistant through mutation, and they compared a collection of models for predicting the log susceptibility, a measure of drug resistance based on the location of mutations. We chose one of the inhibitors, with a total of $n = 1005$ samples and $p = 211$ mutation sites. 
We retained  features with $\geq 5\%$ ones, leaving $69$ features.
{\tt uniPairs-2stage} and {\tt uniPairs} each chose 24 main effects, and both chose the same four interaction pairs.
The {\em  marginal} interactions for the four chosen pairs
are displayed in Figure \ref{fig:nrti}. All four pairs
show strong interaction effects.

\begin{figure}
\centering
\includegraphics[width=3.9in]{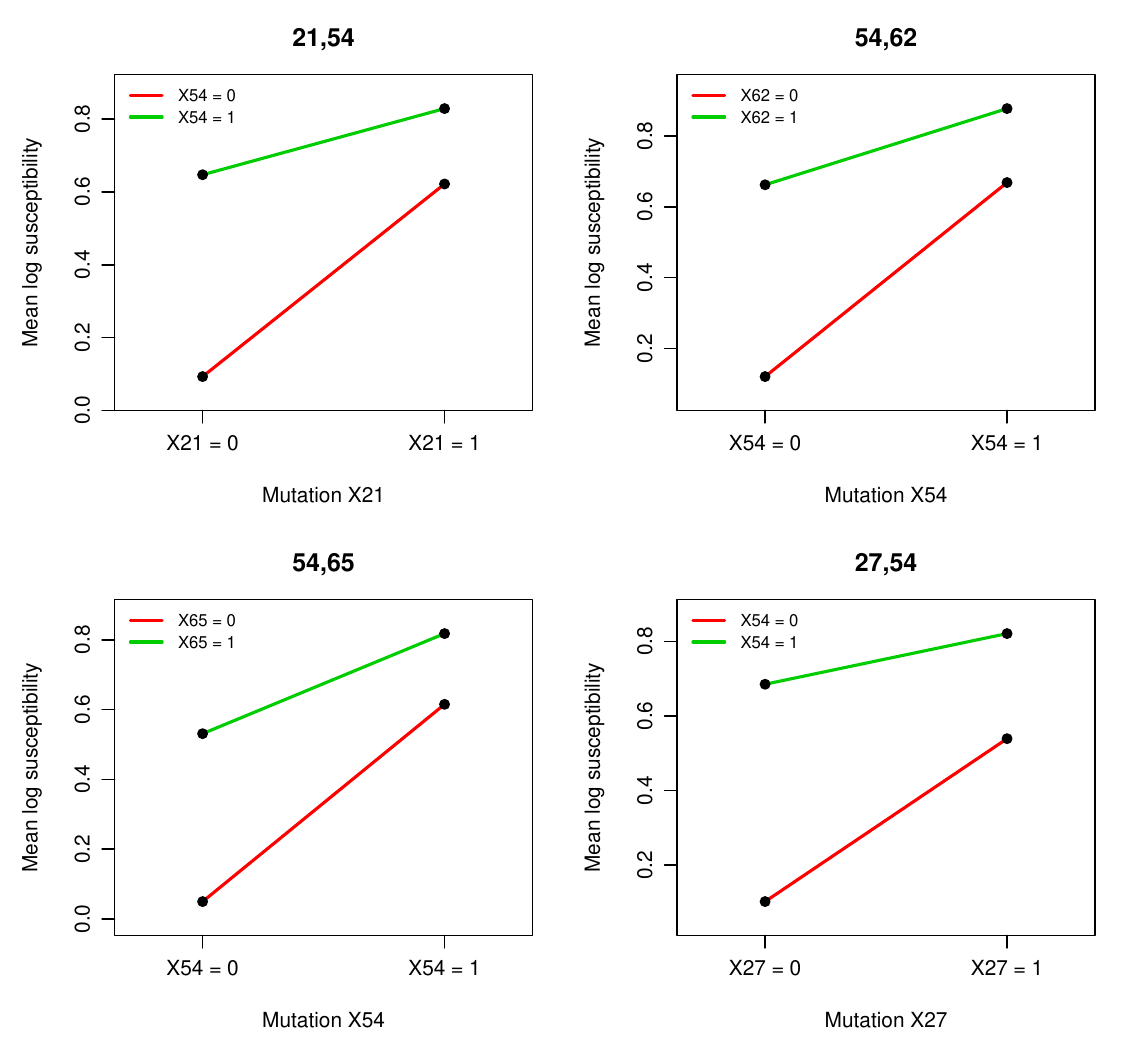}
\caption{\em   Interactions found in the HIV data.
For each of the four selected pairs, the vertical axis shows the mean log susceptibility over the samples having each of the four combinations of the two mutations. The horizontal axis is the first mutation of the pair and the two lines correspond to the second, so non-parallel lines indicate an interaction.}
\label{fig:nrti}
\end{figure}

\section{Details of our proposed algorithms}
\label{sec:details-proposed}

\subsection{Summary}
\label{sec:summary}
\begin{algorithm}[H]
\caption{\em  TripletScan}
\label{alg:tripletscan}
\begin{algorithmic}[1]
\REQUIRE Standardized design matrix $\mathbf X\in\mathbb{R}^{n\times p}$, response $\mathbf Y\in\mathbb{R}^n$, pair index set $\mathcal{P}$ with $M=|\mathcal{P}|$, threshold rule $\in\{\text{raw},\text{normalized}\}$, level $\alpha$ for the normalized rule.
\FOR{each $(j,k)\in\mathcal{P}$}
  \STATE Fit local OLS: $\mathbf Y = \beta_{0,jk}\mathbf 1 + \beta_{j,jk}\mathbf X_j + \beta_{k,jk}\mathbf X_k + \beta_{jk,jk}\mathbf X_j\odot \mathbf X_k + \boldsymbol\varepsilon$.
  \STATE Record the two-sided t-test $p$-value $p_{jk}$ for $\beta_{jk,jk}$.
\ENDFOR
\STATE Sort $\{p_{jk}\}$ increasingly and set the log gaps \(G_r=\log \widehat p_{(r+1)}-\log \widehat p_{(r)}\), \(r=1,\ldots,M-1\).
\IF{rule is raw}
  \STATE \(\widehat r = \arg\max_{1\le r<M}G_r\) \quad (largest log-gap rule).
\ELSE
  \STATE Form \(H_r=rG_r\) and the threshold \(t_{\alpha,M}=-\log(1-(1-\alpha)^{1/(M-1)})\).
  \IF{\(\max_{1\le r<M}H_r\le t_{\alpha,M}\)}
    \STATE \textbf{return} \(\widehat\Gamma=\varnothing\).
  \ENDIF
  \STATE \(\widehat r = \arg\max_{1\le r<M}H_r\) \quad (normalized-gap rule).
\ENDIF
\STATE \(\widehat{\Gamma} = \{(j,k) \in \mathcal{P}: p_{jk}\le p_{(\widehat{r})}\}\).
\ENSURE Selected interactions $\widehat{\Gamma}$.
\end{algorithmic}
\end{algorithm}

Both rules cut the ordered \(p\)-values at a gap, and neither needs a tuning
parameter for the number of retained pairs. They differ in what they compare
the gap against. The raw rule takes the largest gap as it stands, retains at least one pair whenever
\(M\ge2\). The normalized rule weights the \(r\)th gap by its rank and keeps
nothing unless the largest weighted gap clears the threshold \(t_{\alpha,M}\). The
weighting and the threshold are justifed by 
Proposition~\ref{prop:gap-null-law}. We fix \(\alpha\) at \(0.05\) throughout. When the normalized rule
returns the empty set, \texttt{uniPairs} and \texttt{uniPairs-2stage} reduce
to a main-effects fit. Section~\ref{sec:simulation} evaluates both rules.

\begin{algorithm}[H]
\caption{\em  uniPairs-2stage}
\label{alg:hireg}
\begin{algorithmic}[1]
\REQUIRE Design matrix $\mathbf X\in\mathbb{R}^{n\times p}$, response $\mathbf Y\in\mathbb{R}^n$, hierarchy level $h\in\{\text{strong},\text{weak},\text{none}\}$.
\STATE Standardize each column of $\mathbf X$.
\STATE Fit \texttt{UniLasso} on $(\mathbf X,\mathbf Y)$ to obtain main-effects active set $\widehat S_M$ and the out-of-fold (prevalidated) predictions $\widehat{\mathbf Y}^{(1)}_{\mathrm{PV}}$.
\STATE Run \textsc{TripletScan} on $(\mathbf X,\mathbf Y)$.
\STATE Restrict eligible pairs $\mathcal{E}$ based on hierarchy level $h$ and $\widehat S_M$
\STATE Compute residual $\mathbf R = \mathbf Y - \widehat{\mathbf Y}^{(1)}_{\mathrm{PV}}$.
\STATE Fit a \texttt{Lasso} of $\mathbf R$ on the selected interactions $\{\mathbf X_j\odot \mathbf X_k:(j,k)\in\widehat{\Gamma}\cap\mathcal{E}\}$.
\STATE Recover coefficients on the original scale and get active sets $\widehat S_M^{\textit{final}}$ and $\widehat S_I^{\textit{final}}$.
\ENSURE Predictive function
\(
\widehat f(x)
= \widehat\alpha_0
+ \sum_{j\in\widehat S_M^{\textit{final}}}\widehat\alpha_j x_j
+ \sum_{(j,k)\in\widehat S_I^{\textit{final}}}\widehat\alpha_{jk} x_jx_k.
\)
\end{algorithmic}
\end{algorithm}

\begin{algorithm}[H]
\caption{\em uniPairs}
\label{alg:mariner}
\begin{algorithmic}[1]
\REQUIRE Design matrix $\mathbf X\in\mathbb{R}^{n\times p}$, response $\mathbf Y\in\mathbb{R}^n$.
\STATE Standardize each column of $\mathbf X$.
\STATE Run \textsc{TripletScan} on $(\mathbf X,\mathbf Y)$ to obtain interaction set $\widehat{\Gamma}$.
\STATE Form augmented design $\widetilde{\mathbf X} = [\mathbf X, \mathbf X_{\widehat{\Gamma}}]$, where $\mathbf X_{\widehat{\Gamma}} = \{\mathbf X_j\odot \mathbf X_k : (j,k)\in\widehat{\Gamma}\}$.
\STATE Fit \texttt{UniLasso} on $(\widetilde{\mathbf X}, \mathbf Y)$.
\STATE Recover coefficients on the original scale and get active sets $\widehat S_M$ and $\widehat S_I$.
\ENSURE Predictive function
\(
\widehat f(x)
= \widehat\alpha_0
+ \sum_{j\in\widehat S_M}\widehat\alpha_j x_j
+ \sum_{(j,k)\in\widehat S_I}\widehat\alpha_{jk} x_jx_k.
\)
\end{algorithmic}
\end{algorithm}

In practice, we suggest \texttt{uniPairs-2stage} as a default when main-effects are believed to be present and strong/weak hierarchy holds. \texttt{uniPairs} can be seen as a flexible alternative when departures from hierarchy are expected.

\subsection{The uniPairs procedure}
\label{sec:unipairs-details}

For each \(j \in [p]\), let 
\(\mu_j = \frac{1}{n}\sum_{i=1}^n X_{ij}\) and \(\sigma_j^2 = \frac{1}{n-1}\sum_{i=1}^n(X_{ij}-\mu_j)^2\). Let \(\widetilde{\mathbf X}\) have entries
\[
\widetilde X_{ij} = \frac{X_{ij}-\mu_j}{\sigma_j}.
\]
For each \((j,k)\in\mathcal{P}\), fit the OLS model 
\begin{equation}
\mathbf Y=\beta_{0,jk}\mathbf 1+\beta_{j,jk}\widetilde{\mathbf X}_j+\beta_{k,jk}\widetilde{\mathbf X}_k+\beta_{jk,jk}\widetilde{\mathbf X}_{j} \odot \widetilde{\mathbf X}_{k}+\boldsymbol\varepsilon
\label{eqn:marginalpairwise}
\end{equation}
Record
\(
\widehat\beta_{jk}=\widehat\beta_{jk,jk}\) and \[
\widehat p_{jk}=\text{two-sided $t$-test $p$-value for }\widehat\beta_{jk}
\]
As opposed to {\tt uniPairs-2stage}, we can only take the eligible set \(\mathcal{E} = \mathcal{P}\).
Before screening, we remove pairs whose triplet regression is numerically
unstable. This includes rank-deficient designs, designs with a large condition
number or a small minimum Gram eigenvalue, and observations with leverage near
one. Let
\(\mathcal E_{\mathrm{stab}}\subseteq\mathcal E\) denote the remaining pairs.

Sort \(\{\widehat p_{jk}:(j,k)\in\mathcal E_{\mathrm{stab}}\}\) increasingly as
\[
\widehat p_{(1)}\le\cdots\le\widehat p_{(M)},
\qquad M=|\mathcal E_{\mathrm{stab}}|.
\]
Let \(\widehat p^\circ_{(r)}=\max(\widehat p_{(r)},10^{-20})\) and define
\[
G_r=\log\widehat p^\circ_{(r+1)}-\log\widehat p^\circ_{(r)},
\qquad r=1,\ldots,M-1.
\]
The raw rule sets
\(\widehat r_{\mathrm{raw}}\in\arg\max_{r<M}G_r\). The normalized rule forms
\(H_r=rG_r\) and compares its maximum with
\[
t_{\alpha,M}=-\log(1-(1-\alpha)^{1/(M-1)}).
\]
If \(\max_{r<M}H_r\le t_{\alpha,M}\), the normalized rule returns
\(\widehat\Gamma=\varnothing\). Otherwise it sets
\(\widehat r_{\mathrm{norm}}\in\arg\max_{r<M}H_r\). In the nonempty cases,
the selected set is
\begin{equation}
\widehat\Gamma
=\{(j,k)\in\mathcal E_{\mathrm{stab}}:
\widehat p_{jk}\le\widehat p_{(\widehat r)}\},
\qquad
\widehat r=\widehat r_{\mathrm{raw}}\ \text{or}\ \widehat r_{\mathrm{norm}},
\label{eq:detailed-gap-rules}
\end{equation}
according to the chosen rule.
For each \(j\in[p]\), fit the Univariate OLS model \(\mathbf Y = \beta_{0,j}^{\mathrm{uni}}\mathbf 1 + \beta_{1,j}^{\mathrm{uni}}\widetilde{\mathbf X}_j + \boldsymbol\varepsilon.\)
\newline
For each \((j,k) \in \widehat{\Gamma}\), fit the Univariate OLS model \(\mathbf Y = \beta_{0,jk}^{\mathrm{uni}}\mathbf 1 + \beta_{1,jk}^{\mathrm{uni}}\widetilde{\mathbf X}_j \odot \widetilde{\mathbf X}_k +\boldsymbol\varepsilon.\)
For each \(i\in[n]\), compute the leave-one-out predictions
\[
\widehat{\eta}_j^{(-i)}=\widehat\beta_{0,j}^{(-i)\mathrm{uni}}+\widehat\beta_{1,j}^{(-i)\mathrm{uni}}\widetilde X_{ij},
\qquad \widehat{\eta}_{jk}^{(-i)}=\widehat\beta_{0,jk}^{(-i)\mathrm{uni}}+\widehat\beta_{1,jk}^{(-i)\mathrm{uni}}\widetilde X_{ij}\widetilde X_{ik}
\]
Given a penalty level $\lambda>0$, solve

\begin{mini*}
    {\theta_0^s\in\mathbb{R},\ \theta^s\in\mathbb{R}^{p+|\widehat\Gamma|}}
    {
    \frac{1}{n}\sum_{i=1}^n
    \Big(Y_i-\theta_0^s-\sum_{j=1}^{p}\theta_j^s\widehat{\eta}_j^{(-i)}-\sum_{(j,k)\in \widehat\Gamma}\theta_{jk}^s\widehat{\eta}_{jk}^{(-i)}\Big)^2
    \;+\;
    \lambda\sum_{j=1}^{p}|\theta_j^s|
    \;+\;
    \lambda\sum_{(j,k)\in \widehat\Gamma}|\theta_{jk}^s|
    }{}{}
    \addConstraint{\forall j\in [p]\quad \theta_j^s}{\ge 0}
    \addConstraint{\forall (j,k) \in \widehat\Gamma \quad\theta_{jk}^s}{\ge 0}
\end{mini*}
Select \(\lambda\) by \(K\)-fold cross-validation and refit at the chosen value. Denote a solution by \((\widehat\theta_0^s,\widehat\theta^s)\) and define 
\[
\widehat\beta_j^s\ =\ \widehat\theta_j^s\widehat\beta_{1,j}^{\mathrm{uni}},
\qquad
\widehat\beta_{jk}^s\ =\ \widehat\theta_{jk}^s\widehat\beta_{1,jk}^{\mathrm{uni}},
\qquad
\widehat\beta_0^s\ =\ \widehat\theta_0^s+\sum_{j=1}^p \widehat\theta_j^s\widehat\beta_{0,j}^{\mathrm{uni}}+\sum_{(j,k) \in \widehat{\Gamma}} \widehat\theta_{jk}^s\widehat\beta_{0,jk}^{\mathrm{uni}}
\]
Convert back to the original scale 
\[
\widehat\beta_{jk}
=
\frac{\widehat\beta_{jk}^{(s)}}{\sigma_j\sigma_k},
\qquad
\widehat\beta_j
=
\frac{\widehat\beta_j^{(s)}}{\sigma_j}
-
\frac{1}{\sigma_j}
\sum_{k\ne j}
\frac{\widehat\beta_{jk}^{(s)}}{\sigma_k}\,\mu_k,
\]
\[
\widehat\beta_0
=
\widehat\beta_0^{(s)}
-
\sum_{j=1}^p \frac{\widehat\beta_j^{(s)}\mu_j}{\sigma_j}
+
\sum_{(j,k) \in \widehat{\Gamma}} \frac{\widehat\beta_{jk}^{(s)}\mu_j\mu_k}{\sigma_j\sigma_k}
\]
Define the active sets as 
\[
\widehat S_M = \{j\in[p]:\ \widehat\beta_j\neq 0\},
\qquad
\widehat S_I = \{(j,k)\in\widehat\Gamma:\ \widehat\beta_{jk}\neq 0\}
\]
The fitted model is
\[
\widehat f(x)
\;=\;
\widehat\beta_0
+\sum_{j\in\widehat S_M}\widehat\beta_j\,x_j
+\sum_{(j,k)\in\widehat S_I}\widehat\beta_{jk}\,x_jx_k
\]

\subsection {The uniPairs-2stage procedure}
\label{sec:unipairs-2stage-details}
For each \(j \in [p]\), let 
\(\mu_j = \frac{1}{n}\sum_{i=1}^n X_{ij}\) and \(\sigma_j^2 = \frac{1}{n-1}\sum_{i=1}^n(X_{ij}-\mu_j)^2\). Let \(\widetilde{\mathbf X}\) have entries
\[
\widetilde X_{ij} = \frac{X_{ij}-\mu_j}{\sigma_j}.
\]
For each \(j\in[p]\), fit the Univariate OLS model \(\mathbf Y = \beta_{0,j}^{\mathrm{uni}}\mathbf 1 + \beta_{1,j}^{\mathrm{uni}}\widetilde{\mathbf X}_j + \boldsymbol\varepsilon.\)
\newline
For each \(i\in[n]\), compute the leave-one-out predictions
\[
\widehat{\eta}_j^{(-i)}=\widehat\beta_{0,j}^{(-i)\mathrm{uni}}+\widehat\beta_{1,j}^{(-i)\mathrm{uni}}\widetilde X_{ij}
\]
Given a penalty level $\lambda_1>0$, solve
\begin{mini*}
    {\theta_0^s\in\mathbb{R},\ \theta^s\in\mathbb{R}^{p}}
    {
    \frac{1}{n}\sum_{i=1}^n
    \Big(Y_i-\theta_0^s-\sum_{j=1}^{p}\theta_j^s\widehat{\eta}_j^{(-i)}\big)^2
    \;+\;
    \lambda_1\sum_{j=1}^{p}|\theta_j^s|
    }{}{}
    \addConstraint{\forall j\in [p]\quad \theta_j^s}{\ge 0}
\end{mini*}
Select \(\lambda_1\) by \(K\)-fold cross-validation. To construct the
second-stage response, let \(\mathcal I_1,\ldots,\mathcal I_K\) denote the
validation folds. For each \(\lambda\) and each fold \(r\), fit the first-stage
UniLasso model using the observations outside \(\mathcal I_r\), and denote the
resulting predictor by \(\widehat f_{\lambda}^{(1,-r)}\). After selecting
\(\widehat\lambda_1\), define the prevalidated, or out-of-fold, predictions by
\[
\widehat Y_{\mathrm{PV},i}^{(1)}
=\widehat f_{\widehat\lambda_1}^{(1,-r)}(\mathbf X_{i\cdot}),
\qquad i\in\mathcal I_r.
\]
Thus, the prediction assigned to observation \(i\) is obtained from a model that was not fitted using observation \(i\). We use
\[
R_i=Y_i-\widehat Y_{\mathrm{PV},i}^{(1)}
\]
as the response in the second stage. The first-stage model is then refitted
on the full data at \(\widehat\lambda_1\) to obtain its final coefficients.
Denote a solution by \((\widehat\theta_0^s,\widehat\theta^s)\) and define 
\[
\widehat\beta_j^s\ =\ \widehat\theta_j^s\widehat\beta_{1,j}^{\mathrm{uni}},
\qquad
\widehat\beta_0^s\ =\ \widehat\theta_0^s+\sum_{j=1}^p \widehat\theta_j^s\widehat\beta_{0,j}^{\mathrm{uni}}
\]
Let $\widehat S_M^{(1)}=\{j\in[p]:\ \widehat\beta_j^s\neq 0\}$. For each \((j,k)\in\mathcal{P}\), fit the OLS model 
\[\mathbf Y=\beta_{0,jk}\mathbf 1+\beta_{j,jk}\widetilde{\mathbf X}_j+\beta_{k,jk}\widetilde{\mathbf X}_k+\beta_{jk,jk}\widetilde{\mathbf X}_{j} \odot \widetilde{\mathbf X}_{k}+\boldsymbol\varepsilon\]
Record
\(
\widehat\beta_{jk}=\widehat\beta_{jk,jk}\) and \[
\widehat p_{jk}=\text{two-sided $t$-test $p$-value for }\widehat\beta_{jk}
\]
Given a hierarchy regime $h\in\{\text{strong},\text{weak},\text{none}\}$, define the eligible set
\[
\mathcal{E}=
\begin{cases}
\{(j,k)\in\mathcal{P}:\ j\in \widehat S_M^{(1)}\ \text{and}\ k\in \widehat S_M^{(1)}\} &\text{ if  }h=\text{strong},\\
\{(j,k)\in\mathcal{P}:\ j\in \widehat S_M^{(1)}\ \text{or}\ k\in \widehat S_M^{(1)}\}  &\text{ if  }h=\text{weak},\\
\mathcal{P} &\text{ if  }h=\text{none}.
\end{cases}
\]
If $\widehat S_M^{(1)}=\emptyset$, take $h=\text{none}$. As in
\texttt{uniPairs}, remove numerically unstable triplet regressions before
screening. Denote the remaining eligible
pairs by \(\mathcal E_{\mathrm{stab}}\).

Sort their \(p\)-values and define \(G_r\), \(H_r=rG_r\), and
\(t_{\alpha,M}\), with \(M=|\mathcal E_{\mathrm{stab}}|\), as in the
\texttt{uniPairs} procedure above. The raw rule uses the largest \(G_r\).
The normalized rule returns \(\widehat\Gamma=\varnothing\) unless
\(\max_{r<M}H_r>t_{\alpha,M}\), and otherwise uses the largest \(H_r\).
In either nonempty case, \(\widehat\Gamma\) is given by
\eqref{eq:detailed-gap-rules}, with \(\mathcal E_{\mathrm{stab}}\) defined
from the hierarchy-restricted eligible set above.
Collect these prevalidated residuals in
\[
\mathbf R\ =\ \mathbf Y-\widehat{\mathbf Y}^{(1)}_{\mathrm{PV}}
\]
Given a penalty level $\lambda_2>0$, solve
\begin{mini*}
    {\alpha_0^s\in\mathbb{R},\ \alpha^s\in\mathbb{R}^{|\widehat\Gamma|}}
    {
    \frac{1}{n}\sum_{i=1}^n
    \Big(R_i-\alpha_0^s-\sum_{(j,k)\in\widehat\Gamma}\alpha_{jk}^s \widetilde X_{ij}\widetilde X_{ik}\big)^2
    \;+\;
    \lambda_2\sum_{(j,k)\in\widehat\Gamma}|\alpha_{jk}^s|
    }{}{}
\end{mini*}
Select \(\lambda_2\) by \(K\)-fold cross-validation, then convert back to the original scale 
\[
\widehat\beta_{jk}
=
\frac{\widehat\alpha_{jk}^{(s)}}{\sigma_j\sigma_k},
\qquad
\widehat\beta_j
=
\frac{\widehat\beta_j^{(s)}}{\sigma_j}
-
\frac{1}{\sigma_j}
\sum_{k\ne j}
\frac{\widehat\alpha_{jk}^{(s)}}{\sigma_k}\,\mu_k,
\]
\[
\widehat\beta_0
=
\widehat\beta_0^{(s)} + \widehat\alpha_0^{(s)}
-
\sum_{j=1}^p \frac{\widehat\beta_j^{(s)}\mu_j}{\sigma_j}
+
\sum_{(j,k) \in \widehat{\Gamma}} \frac{\widehat\alpha_{jk}^{(s)}\mu_j\mu_k}{\sigma_j\sigma_k}
\]
Define the active sets as 
\[
\widehat S_M = \{j\in[p]:\ \widehat\beta_j\neq 0\},
\qquad
\widehat S_I = \{(j,k)\in\widehat\Gamma:\ \widehat\beta_{jk}\neq 0\}
\]
The fitted model is
\[
\widehat f(x)
\;=\;
\widehat\beta_0
+\sum_{j\in\widehat S_M}\widehat\beta_j\,x_j
+\sum_{(j,k)\in\widehat S_I}\widehat\beta_{jk}\,x_jx_k
\]

\subsection{Effect of Standardizing the Main Effects}
\label{sec:std-main-effects}
Consider again the true model with main and pairwise interaction terms:
\[
Y_i
=
\beta_0^*
+
\sum_{j=1}^p \beta_j^*\, X_{ij}
+
\sum_{1 \le j < k \le p} \gamma_{jk}^*\, X_{ij} X_{ik}
+
\varepsilon_i
\]
Define the standardized covariates
\[
X_{ij}^s = \frac{X_{ij} - \mathbb{E}[X_{ij}]}{\sqrt{\mathrm{Var}(X_{ij})}}
=
\frac{X_{ij} - \mu_j}{\sigma_j}
\]
and rewrite the model in terms of standardized variables as
\[
Y
=
\beta_0^{*,s}
+
\sum_{j=1}^p \beta_j^{*,s}\, X_j^s
+
\sum_{1 \le j < k \le p} \gamma_{jk}^{*,s}\, X_j^s X_k^s
\]
Then
\[
\beta_0^*
=
\beta_0^{*,s}
-
\sum_{j=1}^p \frac{\beta_j^{*,s} \mu_j}{\sigma_j}
+
\sum_{1 \le j < k \le p} \frac{\gamma_{jk}^{*,s}}{\sigma_j \sigma_k}\, \mu_j\, \mu_k
\]
and for each \( j \in [p] \)
\begin{equation}
\beta_j^*
=
\frac{\beta_j^{*,s}}{\sigma_j}
-
\frac{1}{\sigma_j}\sum_{\substack{k=1 \\ k \ne j}}^p
\frac{\gamma_{jk}^{*,s}}{\sigma_k}\, \mu_k
\label{eqn:original-scale-hierarchy}
\end{equation}
and
\[
\gamma_{jk}^* = \frac{\gamma_{jk}^{*,s}}{\sigma_j \sigma_k}
\qquad
\forall\, 1 \le j < k \le p
\]
By convention,
\[
\gamma_{kj}^* = \gamma_{jk}^*,
\qquad
\gamma_{jj}^* = 0
\]
During the fitting procedure, the main effects \( X_j \) are standardized to \(X_j^s\) and the \textsc{{\tt uniPairs-2stage}} / \textsc{\tt uniPairs} model is fit on the \( X_j^s \)’s.

\smallskip
\noindent
When the fitting is complete, the fitted linear model in terms of the standardized covariates \( X_j^s \) is rewritten back in terms of the original covariates \( X_j \).
Hence, there are two sources of hierarchy:
\begin{enumerate}
\item The explicit enforcement of hierarchy during fitting. This acts on the \( Z_j \)’s. This holds when the option hierarchy in {\tt uniPairs-2stage} is set to strong or weak. By default, it is set to None, in which case no hierarchy is enforced and the triplet scans are performed on all $\binom{p}{2}$ pairs of \( X_j^s \)’s. 
\item The implicit hierarchy induced by re-expressing the fitted model in the original variables \( X_j \). This comes from the term \(\frac{1}{\sigma_j}\sum_{\substack{k=1 \\ k \ne j}}^p
\frac{\gamma_{jk}^{*,s}}{\sigma_k}\, \mu_k\) in \ref{eqn:original-scale-hierarchy}. 
\end{enumerate}

\noindent
Therefore, unless the \( X_j \)’s are already mean-zero, this final conversion step will enforce hierarchy automatically.

\section{Related work}
\label{sec:related-work}
Here, we summarize the main ideas from related work on interaction modeling.

\subsection{On the need for hierarchy}
\label{sec:hierarchy}
In interaction modeling, practitioners often choose to enforce some form of \textit{hierarchy}, which can be viewed as a type of regularization analogous to sparsity constraints. Two common forms of hierarchy are:
\begin{itemize}
    \item Strong hierarchy requires that if an interaction term \( \widehat{\gamma}_{jk} \neq 0 \), then both associated main effects \( \widehat{\beta}_j \) and \( \widehat{\beta}_k \) must also be nonzero. 
    
    \item Weak hierarchy only requires that at least one of the main effects is nonzero when an interaction is present i.e \( \widehat{\gamma}_{jk} \neq 0 \Rightarrow \widehat{\beta}_j \neq 0 \) or \( \widehat{\beta}_k \neq 0 \). 
\end{itemize}
The following two motivations are discussed in more details in \citet{bien2013lasso}. First, a classical justification for enforcing strong hierarchy, as discussed by \citet{mccullagh}, is the following: suppose the model takes the form \( Y = \beta_0 + (\beta_1 + \gamma_{12} X_2)X_1 + \ldots \) If \( \beta_1 = 0 \) while \( \gamma_{12} \neq 0 \) and under the assumption that there is no reason to distinguish \( X_2 \) over \( X_2 + c \) from some non-zero constant \(c\), then the model with $X_2+c$ instead of $X_2$ is strongly hierarchical. Second, another justification, offered by \citet{cox}, is that large main effects are more likely to lead to significant interactions than small ones.

In \citet{lim2013learning}, the authors distinguish between parameter sparsity and practical sparsity. The first being number of non-zero coefficients while the latter is the number of raw features used to make a prediction. Consider $Y = X_1 + X_2 + X_1X_2$ then parameter sparsity is $3$ while practical sparsity is $2$. So for a given value of parameter sparsity, a strongly hierarchical model has smaller practical sparsity than a non-hierarchical method.
\subsection{Sprinter}
\label{sec:sprinter}
In \citet{yu2019reluctant}, the main guiding principle is that "one should prefer main effects over interactions if all else is equal".  This assumption is weaker than hierarchy constraints but still introduces a bias in favor of including main effects unless there is strong evidence supporting interactions. In particular, an interaction is kept only if it can not be explained by a linear combination of main effects. 

The paper introduces \texttt{Sprinter} which works as follows : 
\begin{enumerate}
    \item Fit a Lasso model on the main effects. Let \(\hat\theta\) be a solution to 
    \begin{mini*}
        {\theta\in\mathbb{R}^{p}}
        {
        \frac{1}{2n}\|\mathbf Y - \mathbf X\theta\|_2^2 + \lambda_1 \|\theta\|_1
        }{}{}
    \end{mini*}
    \item Compute residuals \( \mathbf R = \mathbf Y - \mathbf X\hat\theta \) and screen for interactions using residual correlation
    \[
    \widehat \Gamma = \left\{(j,k) \in [p]^2 : j < k, \; \widehat{\mathrm{SD}}(\mathbf R) |\widehat{\mathrm{Corr}}(\mathbf Z_{jk}, \mathbf R)| > \eta \right\}
    \]
    where \(\widehat{\mathrm{SD}}\)(resp. \(\widehat{\mathrm{Corr}}\)) is the empirical standard deviation (resp. correlation). This step is a form of Sure Independence Screening (SIS) introduced in \citet{fan2008sure}.
    
    \item Fit a joint Lasso model on the residuals using the selected interactions. Let \((\widehat\xi, \widehat\phi)\) be a solution to 
    \begin{mini*}
    {\xi\in\mathbb{R}^{p}, \phi\in\mathbb{R}^{|\widehat \Gamma|}}
    {
    \frac{1}{2n} \|\mathbf R - \mathbf X\xi - \mathbf Z_{\widehat \Gamma} \phi\|_2^2 + \lambda_3 \left( \|\xi\|_1 + \|\phi\|_1 \right)
    }{}{}
\end{mini*}
\end{enumerate}
The final predictive model for any input \( x \in \mathbb{R}^p \) is given by
\(
x^\top(\widehat\theta + \widehat\xi) + z_{\widehat \Gamma}^\top \widehat\phi
\).

If an interaction is highly correlated with a main effect, then \texttt{APL} may select the interaction while \texttt{Sprinter} by design will prioritize the main effect instead.

For comparison, Interaction Pursuit (IP), as introduced in \citet{fan2016interaction}, first selects a set of strong main effects and then considers only the interactions among them. This enforces strong hierarchy by construction.

To avoid a three-dimensional cross-validation on \( (\lambda_1, \eta, \lambda_3) \) in \texttt{Sprinter}, the authors propose the following : instead of thresholding residual-correlation by a fixed value \( \eta \) define
\[
\widehat \Gamma = \left\{(j,k) \in [p]^2 : j < k, \; |\widehat{\mathrm{Corr}}(\mathbf Z_{jk}, \mathbf R)| \text{ is among the top } m \text{ largest} \right\}
\]
where \( m \sim \frac{n}{\log n} \). Then a path algorithm is applied for two-dimensional cross-validation over \( (\lambda_1, \lambda_3) \).

An extension, \texttt{Sprinter+}, performs cross-validation over \( \lambda_1 \) prior to step 2. This variant is more computationally efficient but may suffer in terms of out-of-sample performance, particularly when the main effects are weak or absent.
\subsection{Glinternet}
\label{sec:glinternet}
In \citet{lim2013learning}, the authors define an interaction between variables \( x \) and \( y \) in a function \( f \) as being present if \( f(x, y) \) cannot be decomposed into the sum of two univariate functions, i.e., \( f(x, y) \ne g(x) + h(y) \) for any functions \( g \) and \( h \).
The paper introduces \texttt{Glinternet} which formulates the problem as a group lasso optimization involving \( p + \binom{p}{2} \) groups : one group for each main effect and one for each pairwise interaction. For continuous variables, \texttt{Glinternet} solves :
\begin{mini*}
  {\substack{\mu,\tilde{\mu},\alpha,\\
             \{\tilde\alpha_j^{(jk)},\tilde\alpha_k^{(jk)},\tilde\alpha_{jk}\}}}
  {
    \begin{aligned}[t]
      &\frac{1}{2n}\left\|\,\mathbf Y - \mu\mathbf{1} - \sum_{j=1}^p \mathbf X_j \alpha_j
      - \sum_{j<k} [\mathbf 1, \mathbf X_j, \mathbf X_k, \mathbf X_j \odot \mathbf X_k]
      \begin{pmatrix}
        \tilde{\mu} \\
        \tilde{\alpha}_j^{(jk)} \\
        \tilde{\alpha}_k^{(jk)} \\
        \tilde{\alpha}_{jk}
      \end{pmatrix}
      \right\|_2^2 \\
      &\quad + \lambda \sum_{j=1}^p |\alpha_j|
      + \lambda \sum_{1 \leq j < k \leq p}
      \sqrt{\tilde{\mu}^2 + (\tilde{\alpha}_j^{(jk)})^2
      + (\tilde{\alpha}_k^{(jk)})^2 + \tilde{\alpha}_{jk}^2}
    \end{aligned}
  }{}{}
\end{mini*}
\texttt{Glinternet} enforces strong hierarchy almost surely through a group lasso penalty structure. Generically, if an interaction term enters the model then the corresponding main effects are also active. 
 
In practice, \texttt{Glinternet} is scalable and can handle problems with tens of thousands of variables as evidenced by the GWAS example in \cite{lim2013learning} with $26801$ variables and $3500$ observations.
\subsection{HierNet}
\label{sec:hiernet}
In \citet{bien2013lasso}, the authors introduce \texttt{HierNet} which in the gaussian case solves the following convex optimization problem
\begin{mini*}
  {\substack{\beta_0 \in \mathbb{R},\; \beta^+, \beta^- \in \mathbb{R}^p,\\
             \Theta \in \mathbb{R}^{p \times p}}}
  {
    \begin{aligned}[t]
      &\frac{1}{2n} \sum_{i=1}^n 
        \left( 
          Y_i - \beta_0 
          - \sum_{j=1}^p X_{ij}(\beta_j^+ - \beta_j^-)
          - \sum_{1 \le j < k \le p} \Theta_{jk} X_{ij} X_{ik}
        \right)^2 \\
      &\quad + \lambda \sum_{j=1}^p (\beta_j^+ + \beta_j^-)
        + \lambda \sum_{1 \le j < k \le p} |\Theta_{jk}|
    \end{aligned}
  }{}{}
  \addConstraint{\Theta}{= \Theta^T}
  \addConstraint{\sum_{k \neq j} |\Theta_{jk}|}{\le \beta_j^+ + \beta_j^-,\quad \forall j \in [p]}
  \addConstraint{\beta_j^+,\, \beta_j^-}{\ge 0,\quad \forall j \in [p]}
\end{mini*}
\texttt{HierNet} is a convex relaxation of the All Pairs Lasso (\texttt{APL}) with the strong hierarchy constraint \( \|\Theta_j\|_1 \le |\beta_j| \), which is non-convex. It is equivalent to solving:
\begin{mini*}
  {\substack{\beta_0 \in \mathbb{R},\; \beta \in \mathbb{R}^p,\\
             \Theta \in \mathbb{R}^{p \times p}}}
  {
    \begin{aligned}[t]
      &\frac{1}{2n} \sum_{i=1}^n 
        \left(
          Y_i - \beta_0 
          - \sum_{j=1}^p X_{ij}\beta_j
          - \sum_{1 \le j < k \le p} \Theta_{jk} X_{ij}X_{ik}
        \right)^2 \\
      &\quad +\lambda \sum_{j=1}^p 
        \max\!\left(\,|\beta_j|,\; \|\Theta_j\|_1 \right)
        \;+\; \lambda \sum_{1 \le j < k \le p} |\Theta_{jk}|
    \end{aligned}
  }{}{}
  \addConstraint{\Theta}{= \Theta^T}
\end{mini*}
which shows that it uses the hierarchical penalty \( \max(|\beta_j|, \|\Theta_j\|_1) \) instead of the standard group-lasso penalty \( \|(\beta_j, \Theta_j)\|_2 \)

If the symmetry constraint \( \Theta = \Theta^T \) is removed, the resulting model enforces generically weak hierarchy. At a fixed level of parameter sparsity, \texttt{HierNet} typically achieves lower practical sparsity than \texttt{APL} and so uses fewer raw features for prediction. According to \citet{lim2013learning}, \texttt{HierNet} works in practice for problems with up to \( p < 1000 \) features.
\subsection{Group Lasso Approaches}
\label{sec:group-lasso}
As discussed in \citet{bien2013lasso}, one way to enforce hierarchical structure in interaction modeling is to apply Lasso regularization over both main effects and interaction terms using a penalty like 
\[
\sum_{1 \leq j < k \leq p} |\Theta_{jk}| + \|(\beta_j, \beta_k, \Theta_{jk})\|_2
\]
This encourages sparsity at both the interaction coefficient level through \( |\Theta_{jk}| \) and the group level through the \( \ell_2 \) norm across the group of interaction and corresponding main effects.

The idea is motivated by a general principle: a penalty of the form \( \|(\beta_i, \beta_j)\|_2 + \|\beta_j\|_2 \) induces a hierarchical dependence of \( \beta_j \) on \( \beta_i \). Specifically, under this penalty, the condition \( \beta_j \neq 0 \Rightarrow \beta_i \neq 0 \) is generically enforced. However, the converse does not necessarily hold i.e \( \beta_j = 0 \) does not imply \( \beta_i = 0 \).
\section{More on the HIV mutation example}
\label{sec:HIVmore}
Recall the HIV mutation example of Section \ref{sec:hiv}, where we saw that the {\tt uniPairs}
algorithms found interactions that display strong marginal interaction effects.
\begin{figure*}[!htbp]
    \centering
    \vskip -.15in
    \includegraphics[width=2.75in]{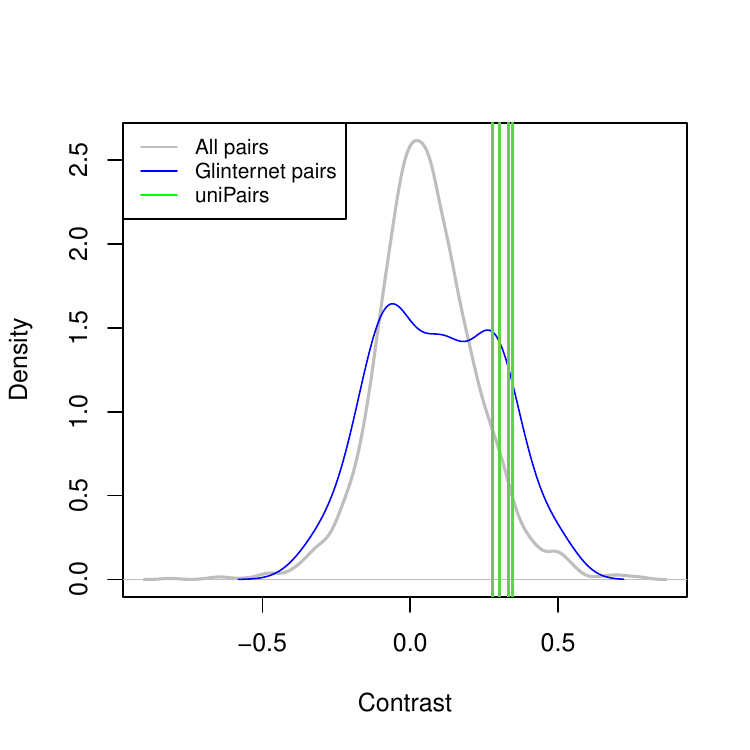}
    \caption{\em  Marginal interactions for all feature pairs (grey), the pairs found by {\tt Glinternet} (blue) and {\tt uniPairs} (green). The interactions are measured by  $\widehat\beta_{jk,jk}$ from the pairwise model (\ref{eqn:marginalpairwise}).
    }
    \label{fig:nrticontrast}
\end{figure*}
Figure \ref{fig:nrticontrast} compares the marginal interactions found by {\tt Glinternet} and {\tt uniPairs}
with the marginal interactions for all pairs. 
We see that {\tt Glinternet} pairs display fairly average marginal 
interactions, while the marginal interactions from {\tt uniPairs} are in the positive tail of the
distribution.

Figure \ref{fig:nrtiBox} shows results from four methods 
applied to $50$ train/test draws of the HIV mutation data.
We see that the uniPairs procedures have slightly higher (Test) MSE than {\tt Glinternet} but win in every other measure.

\begin{figure*}[!htbp]
\centering
\includegraphics[width=4in]{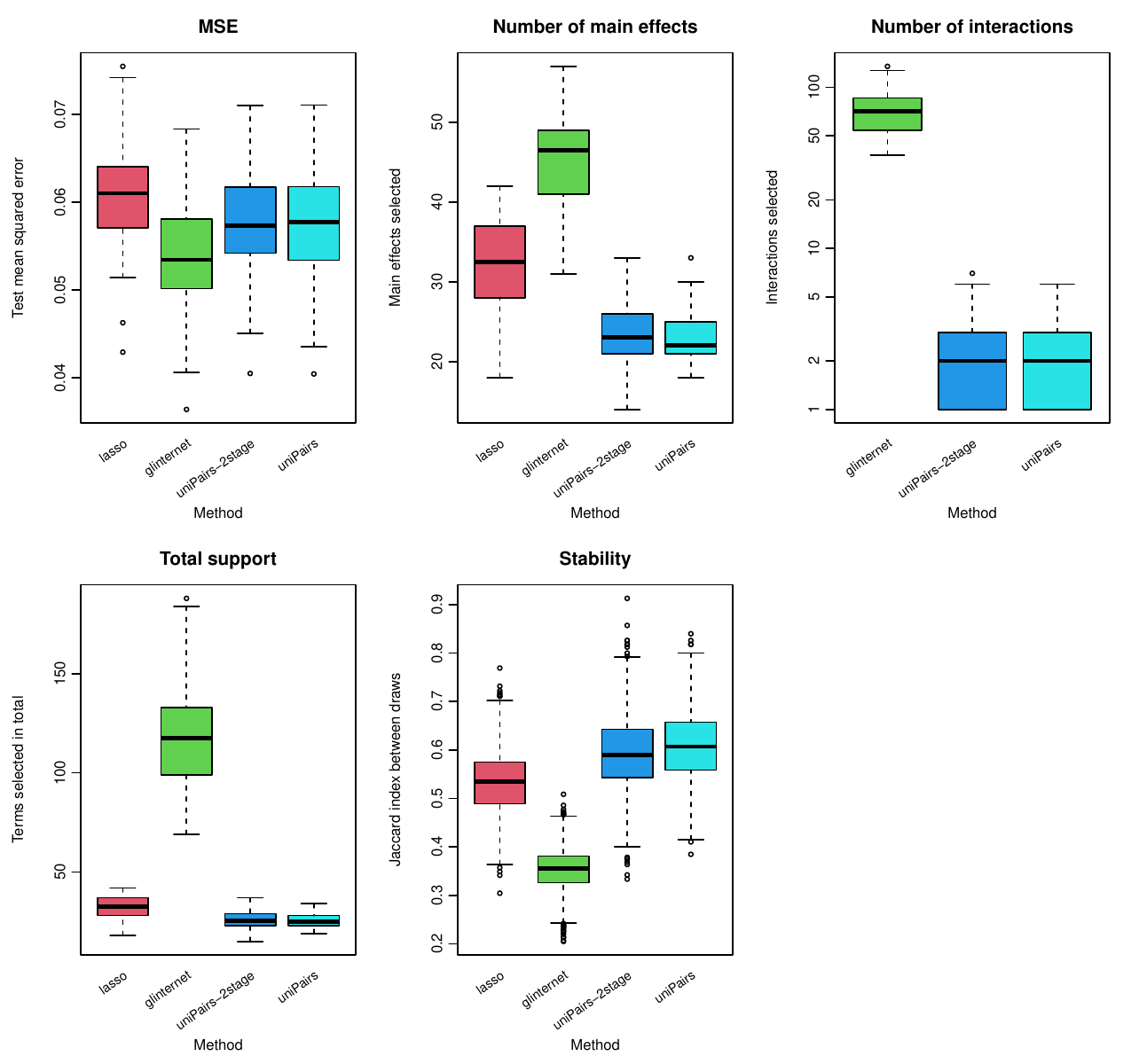}
\caption{\em  Results from 50 train/test draws of the HIV mutation data. The ``stability'' is the Jaccard index between the selected sets of each pair of the 50 draws, averaged over pairs. Interaction counts are on a logarithmic scale. The Lasso baseline is fit using only main-effects, and therefore can't select interaction terms.}
\label{fig:nrtiBox}
\end{figure*}

\section{A Simulation study}
\label{sec:simulation}
\subsection{Setup}
\label{sec:simulation-settings}
Below we compare the performance of \texttt{uniPairs} and
\texttt{uniPairs-2stage}, each with the largest-log-gap and normalized-log-gap
versions of \textsc{TripletScan}, with \texttt{Sprinter(1cv)} and
\texttt{Glinternet} on simulated data. \texttt{HierNet} was too slow to run.
We adopt the data generating process described in section 5.2 of
\citet{yu2019reluctant}, omitting squared effects. More precisely, let \( X \in \mathbb{R}^{n \times p} \) denote the design matrix of \( p \) features. Each row of \( X \) is independently sampled from a multivariate Gaussian distribution with zero mean and an AR(1) covariance structure
\(
\mathrm{Cov}(X_i, X_j) = \rho^{|i-j|}
\) where \(\rho \in [0,1)\). The data-generating mechanism depends on a specified structure which determines the sets of active main effects \(T_1 \subseteq [p]\) and active interaction pairs \(T_3 \subseteq \{(j,k): 1 \le j < k \le p\}\). 
 The considered structures are: 
\begin{itemize}
    \item Mixed: both main effects and interactions are present without structural constraints.  
    
    \(
    T_1 = \{0,1,2,3,4,5\}\) and \(
    T_3 = \{(0,4), (3,17), (9,10), (8,16), (0,12), (3,16)\}
    \)
    
    \item Hierarchical: interactions respect weak hierarchy.
    
    \(
    T_1 = \{0,1,2,3,4,5\}\) and \(
    T_3 = \{(0,2), (1,3), (2,3), (0,7), (1,7), (4,9)\}
    \)
    
    \item Anti-hierarchical: interactions occur only between features with no main effects (violating weak hierarchy).  
    
    \( 
    T_1 = \{0,1,2,3,4,5\}\) and \(
    T_3 = \{(10,12), (11,13), (12,13), (10,17), (11,17), (14,19)\}
    \)
    
    \item Interaction-only: only pairwise interactions are active.
    
    \( 
    T_1 = \varnothing \) and \(
    T_3 = \{(0,2), (1,3), (2,3), (0,7), (1,7), (4,9)\}
    \)
    
    \item Main-effects-only: only main effects are active.
    
    \( 
    T_1 = \{0,1,2,3,4,5\}\) and \( T_3 = \varnothing
    \)
\end{itemize}
A coefficient vector \(\beta \in \mathbb{R}^p\) is defined by \(\beta_j = 0\) unless \(j \in T_1\) in which case \(\beta_j = 2\). The main effect signal is then
\(
\mu_{\text{main}} = X \beta
\). For the active interactions \((j,k) \in T_3\), the interaction signal is
\(
\mu_{\text{interact}} = \sum_{(j,k) \in T_3} 3 \, (X_{\cdot j} \odot X_{\cdot k})
\). The interaction component is orthogonalized with respect to the column space of the active main-effect features. Let \( F = X_{[:,T_1]} \) and define the projection matrix \( P_F = F(F^\top F)^{-1}F^\top \). Then, we do 
\[
\mu_{\text{main}} \leftarrow \mu_{\text{main}} + P_F \mu_{\text{interact}} 
\quad \text{and} \quad
\mu_{\text{interact}} \leftarrow \mu_{\text{interact}} - P_F \mu_{\text{interact}}
\]
After orthogonalization, the interaction component is rescaled to match the variance of the main component:
\(
\mu_{\text{interact}} \leftarrow \mu_{\text{interact}} \times 
\sqrt{\frac{\mathrm{Var}(\mu_{\text{main}})}{\mathrm{Var}(\mu_{\text{interact}})}}
\). The total signal is then
\(
\mu = \mu_{\text{main}} + \mu_{\text{interact}}
\). Gaussian noise \(\varepsilon \sim \mathcal{N}(0, \sigma^2 I_n)\) is added, where the noise variance is set to achieve the desired signal-to-noise ratio 
\(
\sigma^2 = \frac{\mathrm{Var}(\mu)}{\text{SNR}}
\).
The observed response is then
\(
Y = \mu + \varepsilon
\).

Our evaluation metrics are : 
\begin{itemize}
    \item Test \(R^2\) and Train \(R^2\)
    \item Coverage : The fraction of true active variables correctly identified among the true actives, computed separately for main effects, interactions, and jointly. 
    \item False Discovery Rate (FDR) : The fraction of falsely selected variables among all predicted actives, computed for main effects, interactions, and both combined.
    \item Model size : The total number of selected active variables, reported for main effects, interactions, and overall.
\end{itemize}
We consider six algorithms, five data-generating structures, and three SNRs \([0.5, 1, 3]\). The number of features and samples varies across configurations according to  
\[
(n, p) \in \{(1000, 80), (100, 80), (300, 400), (100, 400), (100, 200), (300, 200)\}
\]
The correlation between features follows an AR(1) with \(\rho \in \{0.0, 0.2, 0.5, 0.8, 1.0^{-}\}\). For each combination of structure, SNR, \(\rho\), and \((n, p)\) pair, we perform 40 independent simulation replicates.
\subsection{Simulation results}
\label{sec:simulation-results}

\begin{table}[H]
\centering
\resizebox{\textwidth}{!}{%
    \begin{tabular}{lrrrrrr}
\toprule
\textbf{Method} & \texttt{Glinternet} & \texttt{Sprintr} & \texttt{uniPairs} & \shortstack{\texttt{uniPairs}\\\texttt{(norm. gap)}} & \texttt{uniPairs-2stage} & \shortstack{\texttt{uniPairs-2stage}\\\texttt{(norm. gap)}} \\
\midrule
Test $R^2$                     & \textcolor{blue}{3.18} & 3.79 & 4.06 & 3.51 & 3.21 & 3.26 \\
Train $R^2$                    & 3.45 & \textcolor{blue}{2.74} & 3.36 & 4.39 & 2.90 & 4.15 \\
Coverage Both                  & \textcolor{blue}{2.04} & 5.08 & 3.38 & 3.72 & 3.11 & 3.66 \\
Coverage Main                  & \textcolor{blue}{2.63} & 4.78 & 3.52 & 3.51 & 3.16 & 3.41 \\
Coverage Interactions          & \textcolor{blue}{2.26} & 4.33 & 3.39 & 3.81 & 3.39 & 3.81 \\
FDP Both                       & 4.77 & 5.12 & 2.99 & 2.53 & 3.07 & \textcolor{blue}{2.52} \\
FDP Main                       & 5.23 & 3.52 & 2.92 & \textcolor{blue}{2.91} & 3.42 & 3.00 \\
FDP Interactions               & 3.89 & 5.10 & 3.70 & 2.40 & 3.61 & \textcolor{blue}{2.29} \\
Model size Both                & 5.12 & 4.55 & 3.14 & \textcolor{blue}{2.38} & 3.33 & 2.48 \\
Model size Main                & 5.37 & \textcolor{blue}{2.97} & \textcolor{blue}{2.97} & \textcolor{blue}{2.97} & 3.62 & 3.11 \\
Model size Interactions        & 4.83 & 4.95 & 3.38 & 2.33 & 3.29 & \textcolor{blue}{2.22} \\
\bottomrule
\end{tabular}

}
\caption{\em Global average rank of each method across all simulations. \textcolor{blue}{Lower is better across all metrics}.}
\label{glob-rank}
\end{table}

In Table~\ref{glob-rank}, \texttt{Glinternet} and \texttt{uniPairs-2stage} attain the best average ranks in Test $R^2$, with the two normalized-gap variants close behind. Across coverage metrics, \texttt{Glinternet} ranks best. However, this comes at a substantial cost: \texttt{Glinternet} ranks worst on main-effect FDR and on overall and main-effect model size, whereas \texttt{Sprinter} ranks worst on FDR-Both, FDR-Interactions, and interaction model size. This pattern confirms that \texttt{Glinternet}'s high coverage is achieved through aggressive over-selection. In contrast, all four \texttt{uniPairs} procedures achieve the top ranks for FDR and the smallest model sizes. The normalized-gap variants rank best of all on interaction FDR and on interaction model size, since the threshold \(t_{\alpha,M}\) lets TripletScan return no pair at all when no gap is significant. 

\begin{table}[H]
\centering
\resizebox{\textwidth}{!}{%
    \newcommand{\up}{\(\uparrow\)}
\newcommand{\down}{\(\downarrow\)}

\begin{tabular}{lrrrrrr}
\multicolumn{7}{l}{\textit{Arrows indicate direction of improvement: }\up\textit{ = larger is better,\; }\down\textit{ = smaller is better.}}\\
\toprule
\textbf{Method} & \texttt{Glinternet} & \texttt{Sprintr} & \texttt{uniPairs} & \shortstack{\texttt{uniPairs}\\\texttt{(norm. gap)}} & \texttt{uniPairs-2stage} & \shortstack{\texttt{uniPairs-2stage}\\\texttt{(norm. gap)}} \\
\midrule
Test \(R^2\)\, \up & 0.35 & 0.34 & 0.30 & 0.30 & 0.33 & 0.30 \\
Train \(R^2\)\, \up & 0.50 & 0.52 & 0.49 & 0.41 & 0.51 & 0.41 \\
Coverage Both\, \up & 0.74 & 0.32 & 0.49 & 0.47 & 0.50 & 0.48 \\
Coverage Main\, \up & 0.67 & 0.43 & 0.53 & 0.53 & 0.57 & 0.54 \\
Coverage Interactions\, \up & 0.49 & 0.06 & 0.20 & 0.19 & 0.20 & 0.19 \\
FDP Both\, \down & 0.79 & 0.83 & 0.51 & 0.41 & 0.54 & 0.41 \\
FDP Main\, \down & 0.79 & 0.40 & 0.33 & 0.33 & 0.43 & 0.35 \\
FDP Interactions\, \down & 0.71 & 0.91 & 0.60 & 0.20 & 0.59 & 0.19 \\
Model size Both\, \down & 6226.88 & 27.70 & 10.24 & 11.60 & 11.39 & 10.58 \\
Model size Main\, \down & 68.55 & 7.20 & 6.23 & 6.22 & 7.49 & 6.44 \\
Model size Interactions\, \down & 6158.33 & 20.49 & 4.01 & 5.38 & 3.90 & 4.14 \\
\bottomrule
\end{tabular}

}
\caption{\em The average value of each metric for each method across all simulations.}
\label{glob-avg}
\end{table}

Table~\ref{glob-avg} confirms the patterns observed in Table~\ref{glob-rank}. All four \texttt{uniPairs} procedures provide the best balance between predictive performance and model parsimony. Coverage is highest for \texttt{Glinternet} across both main effects and interactions but this comes with the cost of substantially producing very large models with an average of over $6000$ selected terms compared to only $10-12$ for the \texttt{uniPairs} methods. The normalized-gap variants sharpen the interaction screen: interaction FDR falls from $0.60$ to $0.20$ for \texttt{uniPairs} and from $0.59$ to $0.19$ for \texttt{uniPairs-2stage}, at a cost of $0.01$ in interaction coverage and practically no loss in Test $R^2$ for \texttt{uniPairs}. Their mean interaction model size is nonetheless slightly larger, because the threshold \(t_{\alpha,M}\) is exceeded rarely but occasionally admits a very large set. 

\begin{table}[H]
\centering
\resizebox{\textwidth}{!}{%
    \newcommand{\up}{\(\uparrow\)}
\newcommand{\down}{\(\downarrow\)}

\begin{tabular}{lrrrrr}
\multicolumn{6}{l}{\textit{Arrows indicate direction of improvement: }\up\textit{ = larger is better,\; }\down\textit{ = smaller is better.}}\\
\toprule
\textbf{Method} & \texttt{Sprintr} & \texttt{uniPairs} & \shortstack{\texttt{uniPairs}\\\texttt{(norm. gap)}} & \texttt{uniPairs-2stage} & \shortstack{\texttt{uniPairs-2stage}\\\texttt{(norm. gap)}} \\
\midrule
Test \(R^2\)\, \down & 0.02 & 0.06 & 0.09 & 0.03 & 0.09 \\
Train \(R^2\)\, \down & -0.04 & 0.16 & 0.27 & 0.12 & 0.28 \\
Coverage Both\, \up & -0.83 & -0.48 & -0.44 & -0.44 & -0.44 \\
Coverage Main\, \up & -0.44 & -0.27 & -0.26 & -0.21 & -0.25 \\
Coverage Interactions\, \up & -1.31 & -0.70 & -0.21 & -0.70 & -0.22 \\
FDP Both\, \down & 0.09 & -0.43 & -0.48 & -0.41 & -0.49 \\
FDP Main\, \down & -0.42 & -0.52 & -0.52 & -0.44 & -0.50 \\
FDP Interactions\, \down & 0.15 & -0.11 & -0.18 & -0.12 & -0.22 \\
Model size Both\, \down & -1.53 & -2.31 & -2.33 & -2.16 & -2.35 \\
Model size Main\, \down & -1.72 & -1.58 & -1.58 & -1.35 & -1.59 \\
Model size Interactions\, \down & -1.33 & -3.00 & -2.24 & -2.98 & -2.23 \\
\bottomrule
\end{tabular}

}
\caption{\em Log relative performance of each method compared to \texttt{Glinternet} across all simulations. For Test \(R^2\) and Train \(R^2\), the reported values are \(\boldsymbol{\log\!\big[(1 - R^2_{\texttt{method}})/(1 - R^2_{\texttt{Glinternet}})\big]}\) so that a negative value indicate a reduction in squared error relative to \texttt{Ginternet}. For all other metrics \(\boldsymbol{M}\), the reported values are \(\boldsymbol{\log(M_{\texttt{method}}/M_{\texttt{Glinternet}})}\) with negative (resp. positive) values indicating an improvement over \texttt{Glinternet} in FDP and Model size (resp. Coverage). All cases where the logarithm is ill-defined are removed.}
\label{glob-log}
\end{table}

In Table~\ref{glob-log}, all four \texttt{uniPairs} procedures achieve substantially lower model sizes, but with slightly worse predictive performance relative to \texttt{Glinternet}: the Test $R^2$ entries are positive, ranging from $0.02$ to $0.09$, and therefore indicate slightly larger squared errors. The overall model-size log ratios of $-2.16$ to $-2.35$ correspond to size ratios of roughly $10\%$ to $12\%$ after exponentiation. Coverage also decreases relative to \texttt{Glinternet}, while FDR and model size improve. \texttt{Sprinter} shows weaker improvements in model size: it reduces model size relative to \texttt{Glinternet}, but not as strongly as the \texttt{uniPairs} variants. The interaction rows of this table should be read with care for the normalized-gap variants: a log ratio is undefined when a method selects nothing, so those entries average over the replicates in which the threshold \(t_{\alpha,M}\) fired, and are not comparable with the remaining columns. Table~\ref{glob-avg} gives the unconditional comparison.

The next four figures show the overall model size, Test $R^2$, overall coverage and overall FDR for the case \(n=300,p=400\). The other metrics are given in \hyperref[sec:full-simulations]{Appendix~C}.

\begin{figure*}[!htbp]
    \centering
    \includegraphics[width=\textwidth]{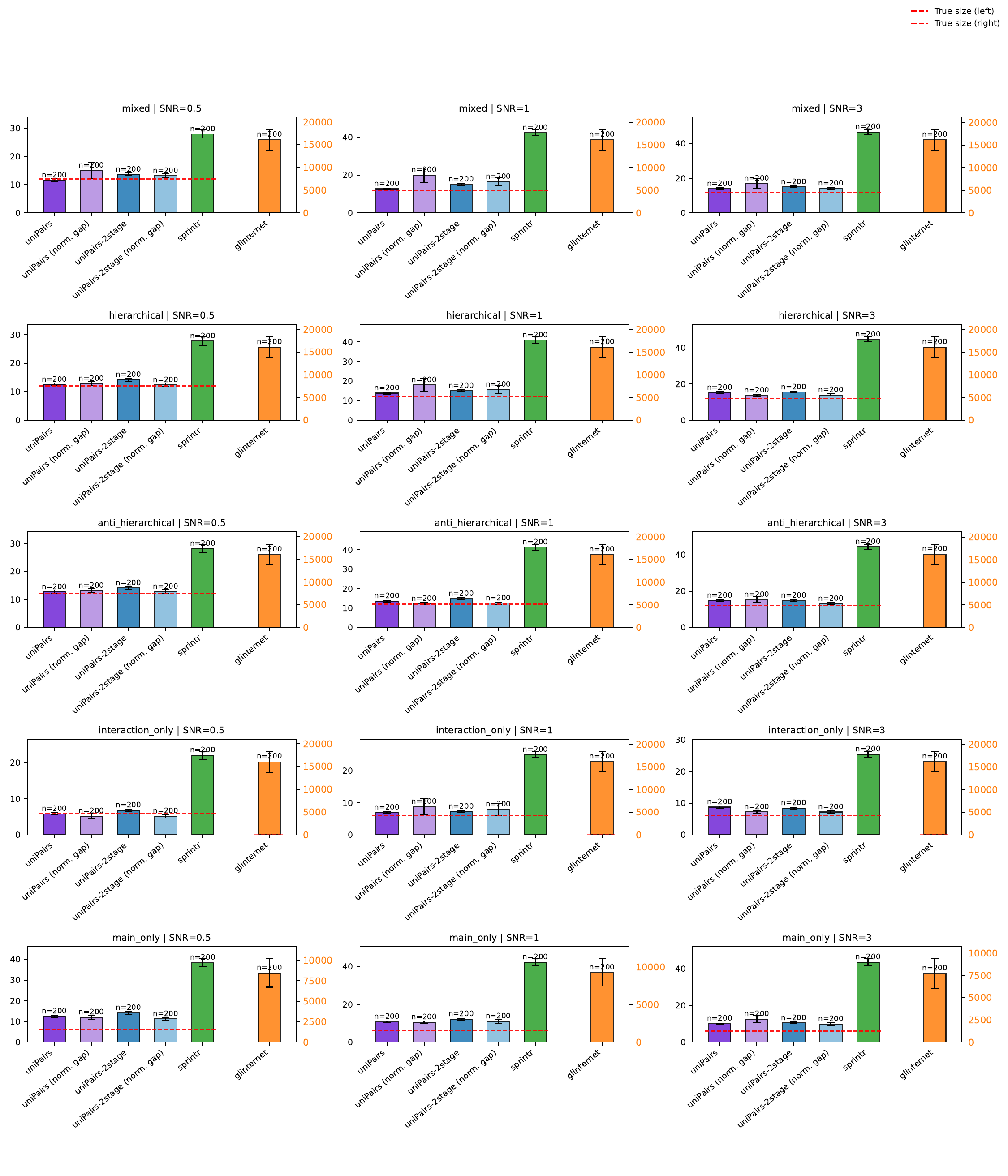}
    \caption{\em 
        Total model size (main + interaction) for $(n,p)=(300,400)$ aggregated over $\rho$.
        Each bar shows mean $\pm$ one standard error across 200 replicates.
        Rows correspond to structures and columns to SNR levels $(0.5,1,3)$.
        The red dashed line marks the true number of active effects. \texttt{Glinternet} is plotted against the right y-axis while the four \texttt{uniPairs} procedures and \texttt{Sprinter} use the left y-axis.
    }
    \label{fig:model_size_both_rho05}
\end{figure*}
In Figure~\ref{fig:model_size_both_rho05}, we see that across all data-generating structures and SNR levels, the four \texttt{uniPairs} procedures produce sparse models whose sizes are consistently close to the true number of active variables. \texttt{Sprinter} selects substantially more variables than needed and \texttt{Glinternet} selects orders of magnitude more terms than the true model, independent of structure and SNR level. Model size patterns remain stable across the different configurations.

\begin{figure*}[!htbp]
    \centering
    \includegraphics[width=\textwidth]{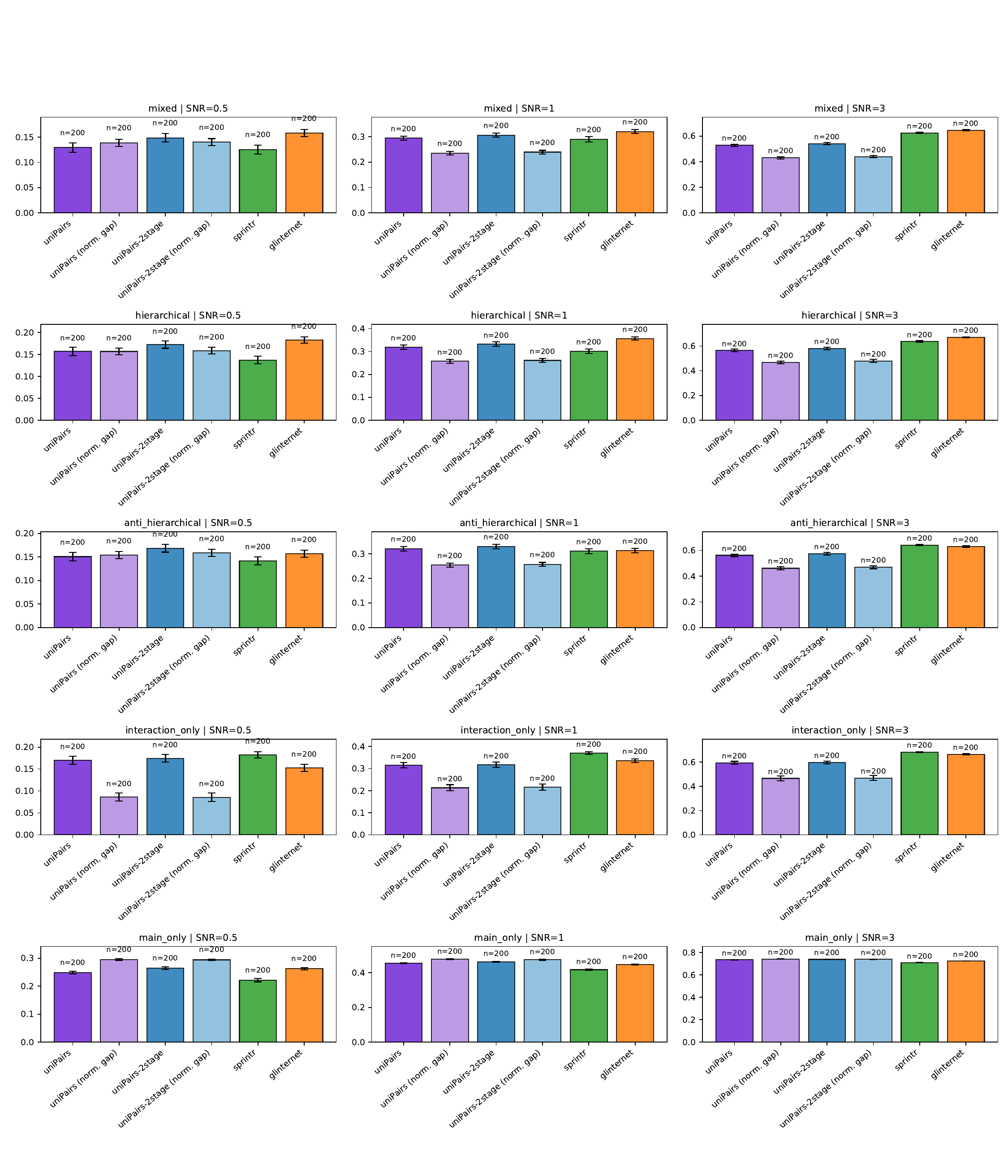}
    \caption{\em 
        Test $R^2$ for $(n,p)=(300,400)$ aggregated over $\rho$.
        Each bar shows mean $\pm$ one standard error across 200 replicates.
        Rows correspond to structures and columns to SNR levels $(0.5,1,3)$.
    }
    \label{fig:test_r2_rho05}
\end{figure*}

In Figure~\ref{fig:test_r2_rho05}, we see that Test $R^2$ increases with SNR as expected. Among the four \texttt{uniPairs} procedures, \texttt{uniPairs-2stage} is slightly ahead on average, but their relative performance varies across settings. The normalized-gap variants have lower Test $R^2$ than their raw-gap counterparts in several settings.
\begin{figure*}[!htbp]
    \centering
    \includegraphics[width=\textwidth]{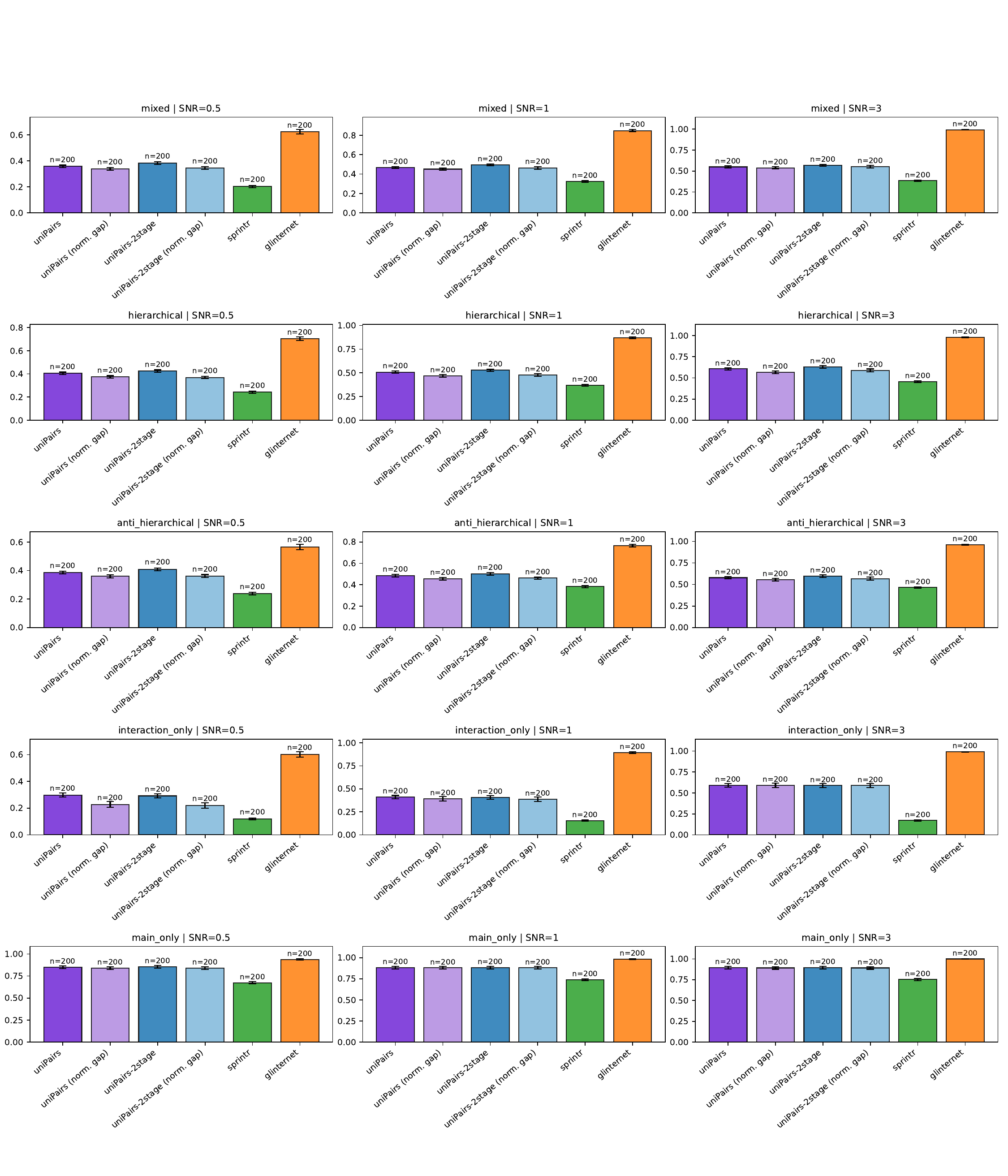}
    \caption{\em 
        Overall coverage (fraction of true active variables correctly identified) for $(n,p)=(300,400)$ aggregated over $\rho$.
        Each bar shows mean $\pm$ one standard error across 200 replicates.
        Rows correspond to structures and columns to SNR levels $(0.5,1,3)$.  
    }
    \label{fig:coverage_both_rho05}
\end{figure*}

In Figure~\ref{fig:coverage_both_rho05}, we see that \texttt{Glinternet} achieves the highest coverage across all structures and SNR levels. This reflects \texttt{Glinternet} agressive selection behavior : it includes nearly all true active variables but at the cost of extremely large model sizes. The four \texttt{uniPairs} procedures have moderate coverage and perform similarly across all structures, the normalized-gap variants losing a little coverage to their stricter screen. \texttt{Sprinter} exhibits the lowest coverage overall.

\begin{figure*}[!htbp]
    \centering
    \includegraphics[width=\textwidth]{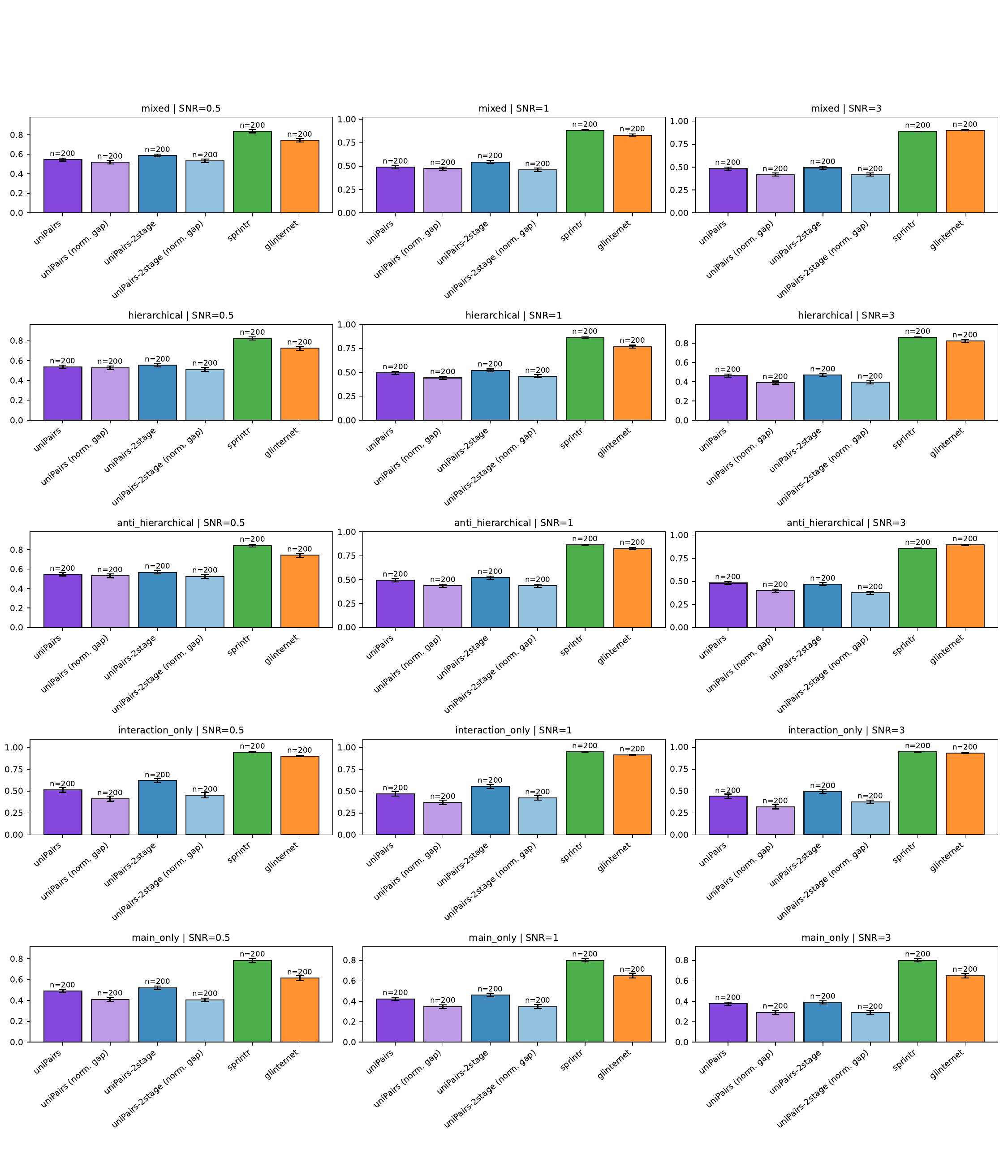}
    \caption{\em 
        Overall false discovery rate (FDR) for $(n,p)=(300,400)$ aggregated over $\rho$.
        Each bar shows mean $\pm$ one standard error across 200 replicates.
    }
    \label{fig:fdp_both_rho05}
\end{figure*}
In Figure~\ref{fig:fdp_both_rho05}, we see that both \texttt{Glinternet} and \texttt{Sprinter} exhibit high overall FDR, even for high SNRs. The \texttt{uniPairs} procedures achieve lower overall FDR, with small changes across structures and SNR levels, and the normalized-gap variants are lower still in every panel. 

\subsection{Summary of the simulation results}
\label{sec:simulation-summary}
First, across a wide range of \(n,p\), correlation levels~\(\rho\), SNRs and structures (mixed, hierarchical, anti-hierarchical, interaction-only, and main-effects-only), all four \texttt{uniPairs} procedures achieve Test \(R^2\) that is competitive with \texttt{Glinternet} and \texttt{Sprinter}, as seen in Figure~\ref{fig:test_r2_rho05}. However, the aggregate comparison in Table~\ref{glob-log} shows slightly larger squared errors for the \texttt{uniPairs} procedures relative to \texttt{Glinternet}. Their main advantages are substantially smaller models and lower FDR, with some predictive-performance tradeoff.

Second, \texttt{uniPairs} and \texttt{uniPairs-2stage} tend to produce substantially smaller predictive models in terms of the number of selected main effects and interactions, as seen in Figures~\ref{fig:model_size_both_rho05}~\ref{fig:model_size_main_rho05}~\ref{fig:model_size_interactions_rho05}. In many settings, their Test \(R^2\) is comparable to that of \texttt{Glinternet}, while they use noticeably fewer interaction terms, often orders of magnitude fewer, which leads to greater practical sparsity.

Third, the \texttt{uniPairs} procedures maintain a very low FDR compared to \texttt{Glinternet} and \texttt{Sprinter} as seen in Figures~\ref{fig:fdp_both_rho05}~\ref{fig:fdp_main_rho05}~\ref{fig:fdp_interactions_rho05}. Their coverage exceeds that of \texttt{Sprinter} and is below that of \texttt{Glinternet} as seen in Figures~\ref{fig:coverage_both_rho05}~\ref{fig:coverage_main_rho05}~\ref{fig:coverage_interactions_rho05}. It is worth noting, however, that \texttt{Glinternet} attains its higher coverage while producing models that are typically far larger than those of the other methods.

Fourth, replacing the raw largest-log-gap rule by the normalized-gap rule is the single most effective change to the interaction screen: it cuts interaction FDR by roughly two thirds for both procedures while costing $0.01$ in interaction coverage. 

\section{Theoretical results}
\label{sec:theory}

\paragraph{Model.}
Let \((X,Y)\) denote a generic observation satisfying
\begin{equation}
Y=\beta_0^*+\sum_{j\in S_M}\beta_j^*X_j
  +\sum_{(a,b)\in S_I}\gamma_{ab}^*X_aX_b+\varepsilon,
\label{eq:model}
\end{equation}
where \(S_M=\supp(\beta^*)\subseteq[p]\) and
\(S_I=\supp(\gamma^*)\subseteq\mathcal P\).

\subsection{The population target of TripletScan}

For \((j,k)\in\mathcal P\), recall \(Q_{j,k}=X_jX_k\), its projection residual
\(R_{j,k}=Q_{j,k}-\Pi_{j,k}Q_{j,k}\), and
\[
\tau_{j,k}=\frac{\Cov(Y,R_{j,k})}{\Var(R_{j,k})}
      =\frac{\E(R_{j,k}Y)}{\E(R_{j,k}^2)}.
\]
Throughout this subsection,
we assume \(\E(\varepsilon\mid X)=0\).

\begin{proposition}
\label{prop:exact-local-target}
If \(\E(R_{j,k}^2)>0\), then
\begin{align}
\tau_{j,k}={}&\gamma_{jk}^*
+\sum_{\ell\notin\{j,k\}}\beta_\ell^*
 \frac{\E(R_{j,k}X_\ell)}{\E(R_{j,k}^2)} \notag\\
&+\sum_{(a,b)\ne(j,k)}\gamma_{ab}^*
 \frac{\E(R_{j,k}X_aX_b)}{\E(R_{j,k}^2)}.
\label{eq:local-ch}
\end{align}
Therefore, neither \(S_I\subseteq S_L\) nor \(S_L\subseteq S_I\)
holds in general.
\end{proposition}

\begin{proposition}
\label{prop:gaussian-local-target}
Suppose \(X\sim N(0,\Sigma)\). Then
\(R_{j,k}=X_jX_k-\Sigma_{jk}\) and
\begin{equation}
\tau_{j,k}=
\frac{\displaystyle\sum_{a<b}\gamma_{ab}^*
 (\Sigma_{ja}\Sigma_{kb}+\Sigma_{jb}\Sigma_{ka})}
 {\Sigma_{jj}\Sigma_{kk}+\Sigma_{jk}^2}.
\label{eq:gaussian-local-formula}
\end{equation}
If the coordinates of \(X\) are independent, then \(\tau_{j,k}=\gamma_{jk}^*\).
\end{proposition}

\begin{corollary}
\label{cor:dense-local-support}
Let \(X\sim N(0,\Sigma)\), where
\[
\Sigma=(1-\rho)I_p+\rho\mathbf 1\mathbf 1^\top,
\qquad 0<\rho<1.
\]
Suppose
\(\gamma_{ab}^*=\gamma\ne0\) is the only
nonzero interaction coefficient. Then \(\tau_{j,k}\ne0\) for every
\(j<k\). More precisely,
\begin{align*}
\tau_{ab}&=\gamma,\\
\tau_{j,k}&=\frac{\gamma(\rho+\rho^2)}{1+\rho^2}
&&\text{if }(j,k)\text{ and }(a,b)\text{ share exactly one index},\\
\tau_{j,k}&=\frac{2\gamma\rho^2}{1+\rho^2}
&&\text{if }(j,k)\text{ and }(a,b)\text{ are disjoint}.
\end{align*}
\end{corollary}

\subsection{The UniLasso step}
\label{sec:unilasso-theory}

In this section, we give theoretical guarantees for the main-effects UniLasso fit in
\texttt{uniPairs-2stage} and for the augmented UniLasso fit in \texttt{uniPairs}.
The first result allows the interaction terms to be omitted from the main-effects fit. The second result allows
the TripletScan selected set to depend on the same observations as the UniLasso fit. We state the results without the initial sample standardization used by the algorithms,
as in Sections~\ref{sec:unipairs-2stage-details} and~\ref{sec:unipairs-details} with
\(\widetilde X_{ij}=X_{ij}\). Equivalently, the covariates have been standardized
using fixed population quantities.

\paragraph{Notation.}
Throughout, \((X_i,Y_i)\), \(i=1,\ldots,n\), are i.i.d.\ copies of
\((X,Y)\), with \(n\ge3\).
For a candidate feature \(V\), either a main effect \(X_j\) or an interaction \(X_jX_k\),
define
\[
b(V)=\frac{\Cov(V,Y)}{\Var(V)},
\qquad
a(V)=\E(Y)-b(V)\E(V)
\]
for its population univariate slope and intercept. We write
\(\widehat b^{(-i)}(V)\) and \(\widehat a^{(-i)}(V)\) for the corresponding
leave-one-out estimates, with \(i=0\) denoting the full-data estimates. Given a
collection \((V_q)\), the UniLasso estimate with penalty \(\lambda>0\) is any
solution of the nonnegative lasso on the leave-one-out fitted values
\[
\widehat\eta_{iq}
=\widehat a^{(-i)}(V_q)+\widehat b^{(-i)}(V_q)V_{iq}.
\]
If \((\widehat\theta_0,\widehat\theta)\) is such a solution, the UniLasso coefficients on the
original feature scale are
\[
\widehat\xi_q=\widehat\theta_q\widehat b^{(0)}(V_q),
\qquad
\widehat\xi_0=\widehat\theta_0+
 \sum_q\widehat\theta_q\widehat a^{(0)}(V_q).
\]
The precise objective is given in \hyperref[proof:theorem1]{Appendix~A}. For \(\nu\in\{2/3,1/2\}\), define
\begin{equation}
\phi_\nu(n,\lambda)
=\min(n\lambda^2,\,(n\lambda)^\nu,\,n^{1/2}).
\label{eq:phi}
\end{equation}
We call the range in which \(\phi_\nu(n,\lambda)=n\lambda^2\) the quadratic regime. 
Specifically, this regime corresponds to \(\lambda\le n^{-1/4}\) for \(\nu=2/3\) and to \(\lambda\le n^{-1/3}\) for \(\nu=1/2\).

The two results use the following common assumptions.
\begin{enumerate}[label=(B\arabic*),ref=B\arabic*]
    \item \label{ass:B1}
    \(\max_{j\in[p]}\|X_j\|_{\psi_2}\le K_X\),
    \(\|\varepsilon\|_{\psi_1}\le K_\varepsilon\), and
    \(\E(\varepsilon\mid X)=0\).
    \item \label{ass:B2}
    \(|\beta_0^*|+\|\beta^*\|_1+\|\gamma^*\|_1\le K_\beta\).
\end{enumerate}

\subsubsection*{The main-effect stage of \texttt{uniPairs-2stage}}

Let \(S_M'=S_M\cup\{0\}\), \(X_0=1\), and \(b_j=b(X_j)\). Define
\[
\beta_{0,M}^*
=\E(Y)-\sum_{j\in S_M}\beta_j^*\E(X_j)
=\beta_0^*+\sum_{(a,b)\in S_I}\gamma_{ab}^*\E(X_aX_b).
\]
This is the intercept of the main-effects-only representation of the model. Define
\begin{equation}
\Lambda_0=\max_{j\notin S_M}|b_j|,
\qquad
\Lambda_1=\max_{j\in S_M}
 \left|\Cov\left(X_j,\sum_{(a,b)\in S_I}\gamma_{ab}^*X_aX_b\right)\right|.
\label{eq:leakage}
\end{equation}
\(\Lambda_0\) is the largest univariate slope among the inactive main effects. \(\Lambda_1\) measures the covariance between each active main effect and the
omitted interaction signal. Both quantities vanish for
independent centered designs. \(\Lambda_1\) also vanishes for any centered Gaussian
design.

\begin{theorem}
\label{th:unipairs-2stage}
Suppose Assumptions~\ref{ass:B1}-\ref{ass:B2} hold and, in addition,
\begin{enumerate}[label=(B\arabic*),ref=B\arabic*,start=3]
    \item \label{ass:B3}
    \(\min_{j\in[p]}\Var(X_j)\ge c_0\),
    \item \label{ass:B4}
    \(\lambda_{\min}(\E(X_{S_M'}X_{S_M'}^\top))\ge\eta\),
    \item \label{ass:B5}
    \(\beta_j^*b_j>0\) and \(|b_j|\ge m\) for every \(j\in S_M\), and
    \(|S_M|\le s\).
\end{enumerate}
Let \((\widehat\beta_0,\widehat\beta)\) be any UniLasso solution on the \(p\)
main effects with penalty \(\lambda\). There are positive constants
\(K_1,\ldots,K_5\), depending only on
\((K_X,K_\varepsilon,K_\beta,c_0,\eta,m,s)\), such that if
\begin{equation}
K_1\Lambda_0\le\lambda\le K_2,
\label{eq:lambda-range-1}
\end{equation}
then, with probability at least
\begin{equation}
1-K_4pn\exp(-K_5\phi_{2/3}(n,\lambda)),
\label{eq:prob-1}
\end{equation}
we have
\begin{equation}
\begin{aligned}
\widehat\beta_j&=0 &&\text{for every }j\notin S_M,\\
\max_{j\in S_M}|\widehat\beta_j-\beta_j^*|
 &\le K_3(\lambda+\Lambda_1),\\
|\widehat\beta_0-\beta_{0,M}^*|
 &\le K_3(\lambda+\Lambda_1).
\end{aligned}
\label{eq:concl-1}
\end{equation}
\end{theorem}

\begin{corollary}
\label{cor:support-1}
On the event in Theorem~\ref{th:unipairs-2stage}, if
\[
\min_{j\in S_M}|\beta_j^*|>K_3(\lambda+\Lambda_1),
\]
then \(\supp(\widehat\beta)=S_M\) and
\(\sign(\widehat\beta_j)=\sign(\beta_j^*)\) for every \(j\in S_M\).
\end{corollary}

For \(\lambda\le n^{-1/4}\), the probability in~\eqref{eq:prob-1} is
\(1-K_4pn\exp(-K_5n\lambda^2)\), which has the same form as the rate in
Theorem~1 of \citet{chatterjee2024unilasso}. In the quadratic regime, one might choose
\[
\lambda=A\sqrt{\frac{\log(pn)}{n}},
\]
where \(A\) is a sufficiently large constant, provided
\(K_1\Lambda_0\leq\lambda\). If
\(\log(pn)=o(n^{1/2})\), then \(\lambda=o(n^{-1/4})\), so
\(\phi_{2/3}(n,\lambda)=n\lambda^2=A^2\log(pn)\) for all sufficiently large \(n\). The failure
probability in Theorem~\ref{th:unipairs-2stage} is therefore bounded by
\[
K_4(pn)^{1-K_5A^2},
\]
which tends to zero when \(A\) is sufficiently large.

\subsubsection*{The augmented stage of \texttt{uniPairs}}

Let \(\mathcal D_0=[p]\cup\mathcal P\) be all main effects and interactions,
where \(V_j=X_j\) for \(j\in[p]\) and
\(V_{(j,k)}=X_jX_k\) for \((j,k)\in\mathcal P\). Let
\[
S=S_M\cup S_I,\qquad S'=S\cup\{0\},\qquad V_0=1.
\]
Denote by \(\xi^*\) the true coefficient vector:
\(\xi_j^*=\beta_j^*\), \(\xi_{(j,k)}^*=\gamma_{jk}^*\), and
\(\xi_0^*=\beta_0^*\). Write \(b_q=b(V_q)\) and define
\begin{equation}
\Lambda_0=\max_{q\in\mathcal D_0\setminus S}|b_q|.
\label{eq:level-2}
\end{equation}
Let \(\widehat\Gamma\subseteq\mathcal P\) be the selected set of
TripletScan, computed from the same observations, and let
\((\widehat\xi_0,\widehat\xi)\) be any UniLasso solution on
\([p]\cup\widehat\Gamma\) with penalty \(\lambda\). Set
\(\widehat\xi_q=0\) for every \(q\in\mathcal P\setminus\widehat\Gamma\).

\begin{theorem}
\label{th:unipairs}
Suppose Assumptions~\ref{ass:B1}-\ref{ass:B2} hold and, in addition,
\begin{enumerate}[label=(B\arabic*$'$),ref=B\arabic*$'$,start=3]
    \item \label{ass:B3p}
    \(\min_{q\in\mathcal D_0}\Var(V_q)\ge c_0\),
    \item \label{ass:B4p}
    \(\lambda_{\min}(\E(V_{S'}V_{S'}^\top))\ge\eta\),
    \item \label{ass:B5p}
    \(\xi_q^*b_q>0\) and \(|b_q|\ge m\) for every \(q\in S\), and
    \(|S|\le s\).
\end{enumerate}
There are positive constants \(K_1,\ldots,K_5\), depending only on 
\((K_X,K_\varepsilon,K_\beta,c_0,\eta,m,s)\), such that if
\(K_1\Lambda_0\le\lambda\le K_2\), then
\begin{align}
&\Prob\left(
S_I\subseteq\widehat\Gamma\ \text{and}\
\left[
\widehat\xi_q\ne0\ \text{for some }q\in([p]\cup\widehat\Gamma)\setminus S
\ \text{or}\
\max_{q\in S'}|\widehat\xi_q-\xi_q^*|>K_3\lambda
\right]\right) \notag\\
&\hspace{35mm}\le
K_4p^2n\exp(-K_5\phi_{1/2}(n,\lambda)).
\label{eq:concl-2}
\end{align}
Therefore, if
\[
\mathcal G=
\left\{\widehat\xi_q=0\ \text{for every }q\notin S,\quad
\max_{q\in S'}|\widehat\xi_q-\xi_q^*|\le K_3\lambda\right\},
\]
then
\begin{equation}
\Prob(\mathcal G)
\ge \Prob(S_I\subseteq\widehat\Gamma)
   -K_4p^2n\exp(-K_5\phi_{1/2}(n,\lambda)).
\label{eq:concl-2b}
\end{equation}
\end{theorem}

\begin{corollary}
\label{cor:fixed-screen}
Suppose \(\widehat\Gamma\) is deterministic, or is computed from observations
independent of the \(n\) observations used in the UniLasso fit, and
\(S_I\subseteq\widehat\Gamma\). If the assumptions and the condition \(K_1\Lambda_0\le\lambda\le K_2\) in
Theorem~\ref{th:unipairs} hold with
\(\mathcal D_0=[p]\cup\widehat\Gamma\) (conditionally on
\(\widehat\Gamma\) in the independent-screen case), then
\[
\Prob(\mathcal G\mid\widehat\Gamma)
\ge1-K_4(p+|\widehat\Gamma|)n
       \exp(-K_5\phi_{1/2}(n,\lambda)).
\]
If, in addition,
\[
\min_{q\in S}|\xi_q^*|>K_3\lambda,
\]
then \(\supp(\widehat\xi)=S\), and
\(
\sign(\widehat\xi_q)=\sign(\xi_q^*)
\) for every \(q\in S\).
\end{corollary}
For \(\lambda\le n^{-1/3}\), the probability in~\eqref{eq:concl-2} is
\(K_4p^2n\exp(-K_5n\lambda^2)\). In the quadratic regime, one might choose
\(\lambda=A\sqrt{\log(p^2n)/n}\), with \(A\) a sufficiently large fixed
constant, provided \(K_1\Lambda_0\le\lambda\). If
\(\log(p^2n)=o(n^{1/3})\), then \(\lambda=o(n^{-1/3})\), so
\(\phi_{1/2}(n,\lambda)=n\lambda^2=A^2\log(p^2n)\) for all sufficiently large \(n\). 
The failure
probability in Theorem~\ref{th:unipairs} is therefore bounded by
\[
K_4(p^2n)^{1-K_5A^2},
\]
which tends to zero when \(A\) is sufficiently large. Note that no independence between the screen and the fit is required in
Theorem~\ref{th:unipairs}. Because model~\eqref{eq:model} is correctly specified once
\(S_I\subseteq\widehat\Gamma\), there is no analogue of \(\Lambda_1\) in Theorem~\ref{th:unipairs}.

\subsection{The log-gap screen rule}
\label{sec:triplet-scan-theory}
For each interaction pair \((j,k)\in\mathcal P\),
\textsc{TripletScan} measures the contribution of \(Q_{j,k}=X_jX_k\)
after adjusting for the two corresponding main effects.
This corresponds to conditional sure independence screening (CSIS) in 
\citet{csis}, applied with conditioning variables
\((X_j,X_k)\) and candidate variable \(Q_{j,k}\).
Recall that \(\Pi_{j,k}\) denotes the population linear projection onto
\(\operatorname{span}\{1,X_j,X_k\}\). The conditional linear covariance is
\[
\operatorname{Cov}_L(Y,Q_{j,k}\mid X_j,X_k)
=\E\big((Y-\Pi_{j,k}Y)R_{j,k}\big)
=\tau_{j,k}\E(R_{j,k}^2).
\]
Therefore, when \(\E(R_{j,k}^2)>0\),
\[
\tau_{j,k}=0
\quad\Longleftrightarrow\quad
\operatorname{Cov}_L(Y,Q_{j,k}\mid X_j,X_k)=0.
\]

Let \(\widehat\tau_{j,k}\) be the estimated coefficient for
\(\tau_{j,k}\) in the local regression. CSIS uses
\(|\widehat\tau_{j,k}|\) as the screening statistic for the interaction
\(Q_{j,k}\). Its screening rule is
\[
\widehat S_{\gamma}
=\{(j,k)\in\mathcal P:|\widehat\tau_{j,k}|\ge\gamma\}
\]
for some threshold \(\gamma>0\).
Then, under the standard CSIS signal strength condition, namely that
\[
\min_{(j,k)\in S_L}
\big|\operatorname{Cov}_L(Y,Q_{j,k}\mid X_j,X_k)\big|
\ge c n^{-\kappa}
\quad\text{for some}\quad \kappa<\tfrac12\quad\text{and}\quad c>0,
\]
together with regularity conditions stated in \citet{csis}, applied
uniformly to these pair-specific local regressions, the
CSIS procedure satisfies the sure-screening property, that is,
\[
\Prob(S_L\subseteq\widehat S_{\gamma})
\xrightarrow[n\to\infty]{}1
\]
for a threshold \(\gamma\asymp n^{-\kappa}\).

The \(p\)-values used by \textsc{TripletScan} are the
ordinary homoskedastic \(t\)-test \(p\)-values of the local regression in
Algorithm~\ref{alg:tripletscan}. Its error includes omitted main effects and
interactions and is generally heteroskedastic. 
Preliminary experiments, not reported here, 
suggested that screening based on heteroskedasticity-consistent covariance estimators HC0–HC3\citep{hc_pvalues} did not work as well as the ordinary
\(t\)-test.

The normalized-log-gap rule accounts for the rank dependence of the
log gaps under an independent-uniform global null, in which the \(r\)th log gap has an exponential distribution with mean
\(1/r\). Multiplication by \(r\) therefore puts the gaps on a common
\(\operatorname{Exp}(1)\) scale as shown in Proposition~\ref{prop:gap-null-law} below.

Fix \(M\ge2\), and let \(P_1,\ldots,P_M\) be i.i.d.\ uniform \(p\)-values, under the global null. Denote their order statistics by
\(P_{(1)}\le\cdots\le P_{(M)}\), and define
\(G_r=\log P_{(r+1)}-\log P_{(r)}\).

\begin{proposition}
\label{prop:gap-null-law}
Under the independent-uniform global null,
\((G_r)_{r\in[M-1]}\) are mutually independent and satisfy
\[
G_r\sim\operatorname{Exp}(r)
\qquad \forall r\in[M-1].
\]
Therefore,
\((H_r)_{r\in[M-1]}:=(rG_r)_{r\in[M-1]}\) are i.i.d.\
\(\operatorname{Exp}(1)\), and
\[
\Prob\left(\max_{r<M}H_r\leq t\right)
=\left(1-e^{-t}\right)^{M-1}
\qquad \forall t\geq0.
\]
At level \(\alpha\), the normalized-gap threshold
\(t_{\alpha,M}\) is
\[
t_{\alpha,M}=-\log\left(1-(1-\alpha)^{1/(M-1)}\right),
\]
for which
\[
\Prob\left(\max_{r<M}H_r\leq t_{\alpha,M}\right)=1-\alpha.
\]
\begin{enumerate}[label=(\roman*),leftmargin=*]
\item The \emph{largest-log-gap rule} chooses
\[
\widehat r_{\rm raw}\in
\arg\max_{r<M}G_r
\]
and retains the pairs corresponding to
\(P_{(1)},\ldots,P_{(\widehat r_{\rm raw})}\). It has no global-null
threshold and therefore selects at least one pair whenever \(M\geq2\).
\item The \emph{normalized-gap rule} uses
\((H_r)_{r\in[M-1]}\). It returns the empty set when
\[
\max_{r<M}H_r\leq t_{\alpha,M}.
\]
Otherwise, it chooses
\[
\widehat r_{\rm norm}\in
\arg\max_{r<M}H_r
\]
and retains the pairs corresponding to
\(P_{(1)},\ldots,P_{(\widehat r_{\rm norm})}\).
\end{enumerate}
\end{proposition}

\newpage

\section{Discussion}
\label{discussion}
We have introduced \texttt{uniPairs} and \texttt{uniPairs-2stage}, two univariate guided procedures for learning sparse interaction models in high dimensions. Both methods use a data-driven screening rule that never requires choosing how many pairs to retain, and they leverage the UniLasso framework to combine univariate fits into a multivariate predictor.
Empirically, the proposed methods have competitive predictive performance relative to existing methods such as \texttt{Sprinter} and \texttt{Glinternet}, while selecting fewer interaction terms and thus producing models that are easier to interpret.
Of the two threshold rules, the rank-normalized gap is the better default: it cuts the interaction FDR of both procedures by roughly two thirds for a coverage cost of about \(0.01\), at the price of a level \(\alpha\) and of possibly returning a main-effects model.
Theoretically, under the stated conditions, the UniLasso stages in both methods exclude inactive variables and satisfy finite-sample \(\ell_{\infty}\) error bounds, extending the guarantees of \citet{chatterjee2024unilasso}.
The TripletScan criterion can be viewed as a conditional sure independence screening mechanism. Directions for future work are : 
\begin{itemize}
    \item Extend \texttt{uniPairs} and \texttt{uniPairs-2stage} to more general feature engineering pipelines.
    \item The TripletScan screening step is embarrassingly parallel, so exploring GPU and more efficient parallel/vectorized implementations could further improve scalability when \(p\) is large. 
    \item Explore a sample-splitting version of both \texttt{uniPairs} and \texttt{uniPairs-2stage}, using one subsample for TripletScan and the remaining observations for model fitting. This separation may improve out-of-sample performance, but at the cost of reducing the effective sample size available for fitting.
\end{itemize}
Both packages {\tt uniPairs} and {\tt uniPairs-2stage} are available on the PyPI repository
\href{https://pypi.org/project/uniPairs/}{https://pypi.org/project/uniPairs/}. Install via 
{\tt pip install uniPairs}. The full documentation is available at \href{https://aymenecharghaoui.github.io/uniPairs/}{https://aymenecharghaoui.github.io/uniPairs/} while the Github repository is at \href{https://github.com/AymenEcharghaoui/uniPairs}{https://github.com/AymenEcharghaoui/uniPairs}.

\medskip

{\bf Acknowledgements.}  We thank Trevor Hastie  for helpful comments. R.T. was supported
by the NIH (5R01EB001988-16) and the NSF (19DMS1208164).

\bibliography{references}

@article{yu2019reluctant,
  title={Reluctant Interaction Modeling},
  author={Yu, Guo and Bien, Jacob and Tibshirani, Ryan},
  journal={arXiv preprint arXiv:1907.08414},
  year={2019}
}

@article{chatterjee2024unilasso,
  title={Univariate-Guided Sparse Regression},
  author={Chatterjee, Sourav and Hastie, Trevor and Tibshirani, Robert},
  journal={arXiv preprint arXiv:2501.18360},
  year={2024}
}

@article{rhee2003,
title={ Human immunodeficiency virus reverse transcriptase and pro-
tease sequence database},
author={Rhee, S.-Y. and Gonzales, M. J. and Kantor, R. and Betts, B. J. and Ravela, J. and
Shafer, R. W.},
journal={Nucleic Acids Research},
volume=31,
pages={298-303},
year=2003,
}

@article{bien2013lasso,
  title={A lasso for hierarchical interactions},
  author={Bien, Jacob and Taylor, Jonathan and Tibshirani, Robert},
  journal={Annals of Statistics},
  volume={41},
  number={3},
  pages={1111--1141},
  year={2013},
  publisher={Institute of Mathematical Statistics},
  doi={10.1214/13-AOS1096}
}

@article{lim2013learning,
  title={Learning interactions through hierarchical group-lasso regularization},
  author={Lim, Michael and Hastie, Trevor},
  journal={arXiv preprint arXiv:1308.2719},
  year={2013}
}

@article{csis,
  title={Conditional sure independence screening},
  author={Barut, Emre and Fan, Jianqing and Verhasselt, Anneleen},
  journal={Journal of the American Statistical Association},
  volume={111},
  number={515},
  pages={1266--1277},
  year={2016},
  publisher={Taylor \& Francis}
}

@book{vershynin2018HDP,
  author    = {Roman Vershynin},
  title     = {High-Dimensional Probability: An Introduction with Applications in Data Science},
  series    = {Cambridge Series in Statistical and Probabilistic Mathematics},
  volume    = {47},
  year      = {2018},
  publisher = {Cambridge University Press},
  address   = {Cambridge, UK},
  doi       = {10.1017/9781108231596},
  isbn      = {9781108415194},
  url       = {https://doi.org/10.1017/9781108231596}
}

@article{rad2018scalable,
  title={A scalable estimate of the out-of-sample prediction error via approximate leave-one-out cross-validation},
  author={Rad, Kamiar Rahnama and Maleki, Arian},
  journal={Journal of the Royal Statistical Society Series B: Statistical Methodology},
  volume={82},
  number={4},
  pages={965--996},
  year={2020},
  publisher={Oxford University Press}
}

@book{mccullagh,
  title={Generalized Linear Models},
  author={McCullagh, P. and Nelder, J.A.},
  series={Chapman \& Hall, London},
  year={1983},
}

@article{cox,
author = {Cox, D. R.},
year = {1984},
pages = {1-31},
title = {Interaction},
volume = {52},
journal = {Internat. Statist. Rev.},
}

@article{fan2008sure,
  title={Sure independence screening for ultrahigh dimensional feature space},
  author={Fan, Jianqing and Lv, Jinchi},
  journal={Journal of the Royal Statistical Society Series B: Statistical Methodology},
  volume={70},
  number={5},
  pages={849--911},
  year={2008},
  publisher={Oxford University Press}
}

@article{fan2016interaction,
  title={Interaction pursuit with feature screening and selection},
  author={Fan, Yingying and Kong, Yinfei and Li, Daoji and Lv, Jinchi},
  journal={arXiv preprint arXiv:1605.08933},
  year={2016}
}

@article{kuchibhotla2022moving,
  title={Moving beyond sub-Gaussianity in high-dimensional statistics: Applications in covariance estimation and linear regression},
  author={Kuchibhotla, Arun Kumar and Chakrabortty, Abhishek},
  journal={Information and Inference: A Journal of the IMA},
  volume={11},
  number={4},
  pages={1389--1456},
  year={2022},
  publisher={Oxford University Press}
}

@article{renyi1953order,
  author={R{\'e}nyi, Alfr{\'e}d},
  title={On the theory of order statistics},
  journal={Acta Mathematica Academiae Scientiarum Hungarica},
  volume={4},
  number={3},
  pages={191--231},
  year={1953}
}

@article{hc_pvalues,
 ISSN = {00031305, 15372731},
 URL = {http://www.jstor.org/stable/2685594},
 author = {J. Scott Long and Laurie H. Ervin},
 journal = {The American Statistician},
 number = {3},
 pages = {217--224},
 publisher = {[American Statistical Association, Taylor & Francis, Ltd.]},
 title = {Using Heteroscedasticity Consistent Standard Errors in the Linear Regression Model},
 urldate = {2026-09-27},
 volume = {54},
 year = {2000}
}
\section*{Appendix A: Proofs}
\label{proof:theorem1}
\subsection*{The population target of TripletScan}

\paragraph{Proof of Proposition \ref{prop:exact-local-target}.}
Let \((j,k)\in\mathcal P\), and
\[
\mathcal L_{j,k}=\operatorname{span}\{1,X_j,X_k\}.
\]
Because \(\Pi_{j,k}Q_{j,k}\) is the orthogonal projection of
\(Q_{j,k}\) onto \(\mathcal L_{j,k}\), the residual
\(R_{j,k}=Q_{j,k}-\Pi_{j,k}Q_{j,k}\) satisfies
\begin{equation}
\E(R_{j,k})=0,
\qquad
\E(R_{j,k}X_j)=0,
\qquad
\E(R_{j,k}X_k)=0.
\label{eq:local-projection-orthogonality}
\end{equation}
In particular,
\[
\Cov(Y,R_{j,k})=\E(R_{j,k}Y),
\qquad
\Var(R_{j,k})=\E(R_{j,k}^2).
\]
Moreover, \(\Pi_{j,k}Q_{j,k}\in\mathcal L_{j,k}\) and
\(R_{j,k}\perp\mathcal L_{j,k}\). Hence, 
\begin{align}
\E(R_{j,k}Q_{j,k})
&=\E\left(R_{j,k}
  (R_{j,k}+\Pi_{j,k}Q_{j,k})\right) \notag\\
&=\E(R_{j,k}^2).
\label{eq:local-product-identity}
\end{align}
\(R_{j,k}\) is measurable with respect to \(X\), so
\begin{align}
\E(R_{j,k}\varepsilon)
&=\E\!\left(\E(R_{j,k}\varepsilon\mid X)\right) =\E\!\left(R_{j,k}\E(\varepsilon\mid X)\right)=0.
\label{eq:local-noise-orthogonality}
\end{align}
Model~\eqref{eq:model} then implies
\begin{align*}
\E(R_{j,k}Y)
={}&\beta_0^*\E(R_{j,k})
+\sum_{\ell=1}^p\beta_\ell^*\E(R_{j,k}X_\ell)+\sum_{a<b}\gamma_{ab}^*\E(R_{j,k}X_aX_b)
+\E(R_{j,k}\varepsilon)\\
={}&\gamma_{jk}^*\E(R_{j,k}^2)
+\sum_{\ell\notin\{j,k\}}\beta_\ell^*
  \E(R_{j,k}X_\ell)+\sum_{(a,b)\ne(j,k)}\gamma_{ab}^*
  \E(R_{j,k}X_aX_b),
\end{align*}
which proves \eqref{eq:local-ch}.

We now give two examples showing that neither inclusion between \(S_I\) and
\(S_L\) holds in general. Take
\[
\beta^*=0, 
\qquad
\varepsilon=0,
\qquad
X=(X_1,X_2,X_3)^\top\sim N(0,\Sigma),
\qquad
\Sigma=(1-\rho)I_3+\rho\mathbf 1\mathbf 1^\top,
\qquad 0<\rho<1.
\]
For every
\((j,k)\in\mathcal P\),
\[
\E(X_jX_k-\rho)=0,
\qquad
\E((X_jX_k-\rho)X_j)=0,
\qquad
\E((X_jX_k-\rho)X_k)=0.
\]
Thus, \(R_{j,k}=X_jX_k-\rho\). By Isserlis' identity,
\begin{align*}
\E(R_{j,k}^2)
&=\E(X_j^2X_k^2)-\rho^2
 =1+\rho^2,\\
\E(R_{1,2}X_1X_3)
&=\E(X_1^2X_2X_3)-\rho\E(X_1X_3)=\rho+\rho^2.
\end{align*}
For \(S_I\nsubseteq S_L\), set
\[
\gamma_{12}^*=1,
\qquad
\gamma_{13}^*=-\frac{1+\rho^2}{\rho+\rho^2},
\qquad
\gamma_{23}^*=0.
\]
Then \((1,2)\in S_I\), and 
\[
\tau_{1,2}
=1+\gamma_{13}^*\frac{\rho+\rho^2}{1+\rho^2}=0,
\]
so \((1,2)\notin S_L\). 
For \(S_L\nsubseteq S_I\), set
\(\gamma_{12}^*=1\) and \(\gamma_{23}^*=\gamma_{13}^*=0\).
Then \((1,3)\notin S_I\), and
\[
\tau_{1,3}=\frac{\rho+\rho^2}{1+\rho^2}>0,
\]
so \((1,3)\in S_L\). 
This completes the proof. \hfill\(\square\)

\paragraph{Proof of Proposition \ref{prop:gaussian-local-target}.}
Let \((j,k)\in\mathcal P\), then
\begin{align*}
\E(X_jX_k-\Sigma_{jk})=0,\\
\E((X_jX_k-\Sigma_{jk})X_j)=0,\\
\E((X_jX_k-\Sigma_{jk})X_k)=0.
\end{align*}
Thus, \(R_{j,k}=X_jX_k-\Sigma_{jk}\).
For every
\((a,b)\in\mathcal P\), by Isserlis' identity,
\begin{align*}
\E(X_jX_kX_aX_b)
={}&\Sigma_{jk}\Sigma_{ab}
+\Sigma_{ja}\Sigma_{kb}
+\Sigma_{jb}\Sigma_{ka}.
\end{align*}
Therefore,
\begin{align}
\E(R_{j,k}X_aX_b)
&=\Sigma_{ja}\Sigma_{kb}+\Sigma_{jb}\Sigma_{ka}.
\label{eq:gaussian-product-covariance}
\end{align}
In particular,
\[
\E(X_j^2X_k^2)
=\Sigma_{jj}\Sigma_{kk}+2\Sigma_{jk}^2,
\]
and
\begin{align}
\E(R_{j,k}^2)
&=\E(X_j^2X_k^2)-\Sigma_{jk}^2
 =\Sigma_{jj}\Sigma_{kk}+\Sigma_{jk}^2.
\label{eq:gaussian-product-variance}
\end{align}
For every \(\ell\in[p]\),
\[
\E(R_{j,k}X_\ell)
=\E(X_jX_kX_\ell)-\Sigma_{jk}\E(X_\ell)=0,
\]
As shown in \eqref{eq:local-noise-orthogonality}, \(\E(R_{j,k}\varepsilon)=0\). Therefore,
\[
\E(R_{j,k}Y)
=\sum_{a<b}\gamma_{ab}^*
  \big(\Sigma_{ja}\Sigma_{kb}+\Sigma_{jb}\Sigma_{ka}\big),
\]
which proves
\eqref{eq:gaussian-local-formula}.

If the coordinates of \(X\) are independent, then the summand in
\eqref{eq:gaussian-local-formula} is zero unless \((a,b)=(j,k)\) for which
\[
\Sigma_{ja}\Sigma_{kb}+\Sigma_{jb}\Sigma_{ka}
=\Sigma_{jj}\Sigma_{kk},
\]
which equals the denominator since \(\Sigma_{jk}=0\). Therefore
\(\tau_{j,k}=\gamma_{jk}^*\). This completes the proof. \hfill\(\square\)

\paragraph{Proof of Corollary \ref{cor:dense-local-support}.}
Because \(\gamma_{ab}^*=\gamma\) is the only nonzero interaction
coefficient, Proposition~\ref{prop:gaussian-local-target} gives, for every
\((j,k)\in\mathcal P\),
\[
\tau_{j,k}
=\gamma\,
\frac{\Sigma_{ja}\Sigma_{kb}+\Sigma_{jb}\Sigma_{ka}}
{\Sigma_{jj}\Sigma_{kk}+\Sigma_{jk}^2} = \gamma\,
\frac{\Sigma_{ja}\Sigma_{kb}+\Sigma_{jb}\Sigma_{ka}}
{1+\rho^2}.
\]
If \((j,k)=(a,b)\), the numerator is
\(\gamma(1+\rho^2)\). If \((j,k)\) and \((a,b)\) share exactly one index,
the numerator is \(\gamma(\rho+\rho^2)\). If they are
disjoint,
the numerator is \(2\gamma\rho^2\). Therefore, \(\tau_{j,k}\ne0\) for every
\((j,k)\in\mathcal P\). \(\square\)

\subsection*{The log-gap screen rule}

\paragraph{Proof of Proposition \ref{prop:gap-null-law}.}
Let \((P_m)_{m\in[M]}\) be i.i.d. 
\(\operatorname{Unif}(0,1)\), and set
\(E_m=-\log P_m\) for every \(m\in[M]\). Then
\((E_m)_{m\in[M]}\) are i.i.d.\ \(\operatorname{Exp}(1)\). Let
\(E_{(1)}\le\cdots\le E_{(M)}\) be their order statistics.
Set \(E_{(0)}=0\), and define
\[
\Delta_k=E_{(k)}-E_{(k-1)},\qquad \forall k\in[M].
\]
By R\'{e}nyi's representation theorem \citep{renyi1953order},
\((\Delta_k)_{k\in[M]}\) are mutually independent and
\[
\Delta_k\sim\operatorname{Exp}(M-k+1)
\qquad \forall k\in[M].
\]
We have \(P_{(r)}=\exp(-E_{(M-r+1)})\). Thus, for every
\(r\in[M-1]\),
\[
\begin{aligned}
G_r
&=\log P_{(r+1)}-\log P_{(r)}\\
&=E_{(M-r+1)}-E_{(M-r)}\\
&=\Delta_{M-r+1}
\sim\operatorname{Exp}(r).
\end{aligned}
\]
In particular, \((G_r)_{r\in[M-1]}\) are mutually independent. Therefore,
\((H_r)_{r\in[M-1]}=(rG_r)_{r\in[M-1]}\) are i.i.d.\
\(\operatorname{Exp}(1)\) variables. Hence, for every \(t\geq0\),
\[
\begin{aligned}
\Prob\left(\max_{r\in[M-1]}H_r\leq t\right)
&=\Prob\left(\bigcap_{r\in[M-1]}\{H_r\leq t\}\right)\\
&=\prod_{r=1}^{M-1}\Prob(H_r\leq t)\\
&=(1-e^{-t})^{M-1}.
\end{aligned}
\]
Finally, \(t_{\alpha,M}\) satisfies
\[
(1-e^{-t_{\alpha,M}})^{M-1}=1-\alpha.
\]
So,
\(
t_{\alpha,M}
=-\log\left(1-(1-\alpha)^{1/(M-1)}\right).
\)
\hfill\(\square\)

\subsection*{The UniLasso step}

Theorems~\ref{th:unipairs-2stage} and~\ref{th:unipairs} are consequences of Theorem~\ref{th:main} about the leave-one-out nonnegative lasso.
Throughout, for any function \(f:[n]\to\mathbb R\), define
\[
P_n[f]=\frac{1}{n}\sum_{l=1}^n f(l)
\]
and, for every \(i\in[n]\), define
\[
L_{n,i}[f]=\frac{1}{n-1}\sum_{\substack{l=1\\l\ne i}}^n f(l).
\]
We also set \(L_{n,0}[f]=P_n[f]\). When \(f(l)\) is given by an expression
involving the \(l\)th observation, we use that expression as the argument of
the operator. For example, \(P_n[V_qY]=n^{-1}\sum_{l=1}^nV_{lq}Y_l\).
All constants denoted $C,c,C_1,\dots$ are positive and depend only on the
parameters listed in the corresponding statement.

\subsubsection*{Sub-Weibull tools}

For $\alpha>0$, the sub-Weibull norm of a random variable $Z$ is
\[
\|Z\|_{\psi_\alpha}=\inf\Big\{c>0:\ \E\big(\exp\big(|Z|^\alpha/c^\alpha\big)\big)\le2\Big\},
\]
so that $\psi_2$ is the sub-Gaussian and $\psi_1$ the sub-exponential norm.
By monotone convergence, when the infimum is finite, it is attained. Therefore,
$\|Z\|_{\psi_\alpha}\le K$ implies $\E\big(\exp(|Z|^\alpha/K^\alpha)\big)\le2$.

\begin{lemma}
\label{lem:orlicz}
Let $\alpha,\beta>0$ and let $U,V,Z,Z_1,Z_2$ be random variables.
\begin{enumerate}[label=(\alph*),leftmargin=*]
\item If $\|Z\|_{\psi_\alpha}\le K$, then $\Prob(|Z|\ge t)\le2e^{-(t/K)^\alpha}$ for all
$t\ge0$, and for every $k\ge1$,
$\E(|Z|^k)\le C_{\alpha,k}K^k$ with $C_{\alpha,k}=2\frac{k}{\alpha}\Gamma(\frac k\alpha)$.
\item If $0<\alpha'<\alpha$, there exists a constant $C_{\alpha,\alpha'}>0$ with
$\|Z\|_{\psi_{\alpha'}}\le C_{\alpha,\alpha'}\|Z\|_{\psi_\alpha}$.
\item If $1/\gamma=1/\alpha+1/\beta$, then
$\|UV\|_{\psi_\gamma}\le\|U\|_{\psi_\alpha}\|V\|_{\psi_\beta}$.
\item
\[
\|Z_1+Z_2\|_{\psi_\alpha}
\le2^{1+1/\alpha}\max(\|Z_1\|_{\psi_\alpha},\|Z_2\|_{\psi_\alpha}).
\]
For $\alpha\ge1$, $\|\cdot\|_{\psi_\alpha}$ is a norm.
Moreover, $\|1\|_{\psi_\alpha}=(\log2)^{-1/\alpha}$ and
\[
\|Z-\E(Z)\|_{\psi_\alpha}\le C^{\mathrm c}_\alpha\|Z\|_{\psi_\alpha},
\qquad
C^{\mathrm c}_\alpha
=2^{1+1/\alpha}\max(1,C_{\alpha,1}(\log2)^{-1/\alpha}).
\]
\end{enumerate}
\end{lemma}

\begin{proof}
For (a), by Markov's inequality, for every \(t\geq0\),
\[
\begin{aligned}
\Prob(|Z|\geq t)
&=\Prob\left(\exp(|Z|^\alpha/K^\alpha)
       \geq \exp(t^\alpha/K^\alpha)\right)\\
&\leq \exp(-t^\alpha/K^\alpha)
       \E\left(\exp(|Z|^\alpha/K^\alpha)\right)\\
&\leq 2\exp\left(-(t/K)^\alpha\right).
\end{aligned}
\]
Using
\(\E(|Z|^k)=\int_0^\infty kt^{k-1}\Prob(|Z|\geq t)\,dt\) and taking \(y=(t/K)^\alpha\), we get
\[
\begin{aligned}
\E(|Z|^k)
&\leq 2\int_0^\infty kt^{k-1}\exp\left(-(t/K)^\alpha\right)dt\\
&=2K^k\frac{k}{\alpha}\int_0^\infty
  y^{k/\alpha-1}e^{-y}\,dy
=C_{\alpha,k}K^k.
\end{aligned}
\]

For (b), let \(K=\|Z\|_{\psi_\alpha}\) and
\(\vartheta=\alpha'/\alpha\in(0,1)\). Without loss of generality, assume \(K>0\). Since
\(y^\vartheta\leq\vartheta y+1-\vartheta\) for every \(y\geq0\), we get that for any
\(c\geq1\),
\[
\left(\frac{|Z|}{cK}\right)^{\alpha'}
\leq c^{-\alpha'}\left(\vartheta\left(\frac{|Z|}{K}\right)^\alpha
      +1-\vartheta\right).
\]
\(x\mapsto x^{c^{-\alpha'}\vartheta}\) is concave on \((0,\infty)\). Hence, by
Jensen's inequality,
\[
\begin{aligned}
\E\left(\exp\left((|Z|/(cK))^{\alpha'}\right)\right)
&\leq \exp\left(c^{-\alpha'}(1-\vartheta)\right)
 \left(\E\left(\exp(|Z|^\alpha/K^\alpha)\right)\right)^{c^{-\alpha'}\vartheta}\\
&\leq \exp\left(c^{-\alpha'}(1-\vartheta)\right)
 2^{c^{-\alpha'}\vartheta}.
\end{aligned}
\]
The last expression is at most \(2\) whenever
\[
c^{-\alpha'}(1-\vartheta+\vartheta\log 2)\leq\log 2.
\]
Therefore, we take
\[
C_{\alpha,\alpha'}
=\left(\frac{1-\vartheta+\vartheta\log 2}{\log 2}\right)^{1/\alpha'}.
\]

For (c), take \(K_U>\|U\|_{\psi_\alpha}\) and
\(K_V>\|V\|_{\psi_\beta}\). Since
\(\gamma/\alpha+\gamma/\beta=1\), by Young's inequality,
\[
\left(\frac{|UV|}{K_UK_V}\right)^\gamma
\leq \frac{\gamma}{\alpha}\left(\frac{|U|}{K_U}\right)^\alpha + 
    \frac{\gamma}{\beta}\left(\frac{|V|}{K_V}\right)^\beta.
\]
Applying H\"older's inequality, we get
\[
\begin{aligned}
\E\left(\exp\left((|UV|/(K_UK_V))^\gamma\right)\right)
&\leq
\left(\E\left(\exp((|U|/K_U)^\alpha)\right)\right)^{\gamma/\alpha}\times
\left(\E\left(\exp((|V|/K_V)^\beta)\right)\right)^{\gamma/\beta}
\leq 2.
\end{aligned}
\]
Therfore, \(\|UV\|_{\psi_\gamma}\leq K_UK_V\). Taking
\(K_U\downarrow\|U\|_{\psi_\alpha}\) and
\(K_V\downarrow\|V\|_{\psi_\beta}\) concludes the proof.

For (d), let
\[
M=\max\left(\|Z_1\|_{\psi_\alpha},\|Z_2\|_{\psi_\alpha}\right)
\quad\text{and}\quad c=2^{1+1/\alpha}M.
\]
Since \((a+b)^\alpha\leq2^\alpha(a^\alpha+b^\alpha)\) for all
\(a,b\geq0\), Cauchy-Schwarz implies
\[
\begin{aligned}
\E\left(\exp\left((|Z_1+Z_2|/c)^\alpha\right)\right)
&\leq
\left(\E\left(\exp(2^{\alpha+1}|Z_1|^\alpha/c^\alpha)\right)\right)^{1/2}\times
\left(\E\left(\exp(2^{\alpha+1}|Z_2|^\alpha/c^\alpha)\right)\right)^{1/2}
\leq2,
\end{aligned}
\]
If \(\alpha\geq1\), \(x\mapsto e^{x^\alpha}-1\) is convex and increasing. Hence,
\(\|\cdot\|_{\psi_\alpha}\) is its Orlicz norm \citep[Chapter~2.7.1]{vershynin2018HDP}. Moreover, by definition of \(\|\cdot\|_{\psi_\alpha}\),
\[
\|1\|_{\psi_\alpha}=(\log 2)^{-1/\alpha}.
\]
Finally, from (a),
\[
\|\E(Z)\|_{\psi_\alpha}
=|\E(Z)|(\log 2)^{-1/\alpha}
\leq C_{\alpha,1}(\log 2)^{-1/\alpha}\|Z\|_{\psi_\alpha}.
\]
Applying the first inequality of (d) to \(Z\) and \(-\E(Z)\) concludes the proof.
\end{proof}

The next inequality is the generalized Bernstein inequality of \citet[Theorem~3.1]{kuchibhotla2022moving}
in the i.i.d.\ case. For $\nu=1$, it is Bernstein's inequality
\citep[Corollary~2.8.3]{vershynin2018HDP}.

\begin{lemma}
\label{lem:two-regime}
Let $0<\nu\le1$ and let $U_1,\dots,U_N$ be i.i.d.\ with
$\|U_1\|_{\psi_\nu}\le K$. Recall that
$P_N[U]=N^{-1}\sum_{i=1}^N U_i$. Then, there is a constant $c_\nu>0$ depending only on
$\nu$, such that, for all $x>0$,
\begin{equation}
\Prob\left(|P_N[U]-\E(U_1)|\geq x\right)
\leq 2\exp\left(-c_\nu\min\left(\frac{Nx^2}{K^2},
\left(\frac{Nx}{K}\right)^\nu\right)\right).
\label{eq:two-regime}
\end{equation}
\end{lemma}

\begin{proof}
Let \(\widetilde U_i=U_i-\E(U_i)\) for every \(i\in[N]\). By
Lemma~\ref{lem:orlicz}(d),
\[
\|\widetilde U_1\|_{\psi_\nu}
\leq C^{\mathrm c}_\nu\|U_1\|_{\psi_\nu}
\leq C^{\mathrm c}_\nu K.
\]
Theorem~3.1 of \citet{kuchibhotla2022moving}, applied to the i.i.d.\ mean-zero
variables \((\widetilde U_i)_{i\in[N]}\), implies that there exists a constant
\(C_\nu>0\), depending only on \(\nu\), such that, for all \(t\geq0\),
\[
\Prob\left(
|P_N[U]-\E(U_1)|
\geq C_\nu C^{\mathrm c}_\nu K
\left(\sqrt{\frac{t}{N}}+\frac{t^{1/\nu}}{N}\right)
\right)
\leq2e^{-t}.
\]
Let \(A_\nu=C_\nu C^{\mathrm c}_\nu\). For \(x>0\), take
\[
t_x=\min\left(\frac{Nx^2}{4A_\nu^2K^2},
\left(\frac{Nx}{2A_\nu K}\right)^\nu\right).
\]
We have
\[
A_\nu K\sqrt{\frac{t_x}{N}}\leq\frac{x}{2}
\quad\text{and}\quad
A_\nu K\frac{t_x^{1/\nu}}{N}\leq\frac{x}{2}.
\]
Therefore,
\[
\Prob\left(|P_N[U]-\E(U_1)|\geq x\right)
\leq2e^{-t_x}.
\]
Let
\[
c_\nu=\min\left((2A_\nu)^{-2},(2A_\nu)^{-\nu}\right)>0.
\]
Then,
\[
t_x\geq c_\nu\min\left(\frac{Nx^2}{K^2},
\left(\frac{Nx}{K}\right)^\nu\right),
\]
which concludes the proof of \eqref{eq:two-regime}.
\end{proof}

\subsubsection*{The LOO nonnegative lasso}

Let \(\mathcal D_0\) be a finite set with
\(d_0=|\mathcal D_0|\geq1\), and
\((V_q)_{q\in\mathcal D_0}\), \(Y\), and \(u\) be random variables defined on
the same probability space. Let \(S\subseteq\mathcal D_0\) with
\(|S|\leq s\). Introduce \(0\notin\mathcal D_0\), and set
\[
S'=S\cup\{0\},\qquad V_0=1.
\]
Suppose that \(\xi^*=(\xi_q^*)_{q\in S'}\) satisfies
\begin{equation}
Y=\xi_0^*+\sum_{q\in S}\xi_q^*V_q+u,
\qquad \E(u)=0.
\label{eq:active-model}
\end{equation}
Let
\[
\left((V_{iq})_{q\in\mathcal D_0},Y_i,u_i\right)_{i\in[n]}
\]
be i.i.d.\ copies of \(((V_q)_{q\in\mathcal D_0},Y,u)\), where \(n\geq3\).
For every \(q\in\mathcal D_0\), define the population slope and intercept of
the univariate regression of \(Y\) on \(V_q\) by
\[
b_q=\frac{\Cov(V_q,Y)}{\Var(V_q)},
\qquad
a_q=\E(Y)-b_q\E(V_q).
\]
For every \(q\in\mathcal D_0\) and \(i\in\{0,1,\ldots,n\}\), define
\begin{equation}
\widehat b_q^{(-i)}=\frac{L_{n,i}[V_qY]-L_{n,i}[V_q]\,L_{n,i}[Y]}
{L_{n,i}[V_q^2]-(L_{n,i}[V_q])^2},
\qquad
\widehat a_q^{(-i)}=L_{n,i}[Y]-\widehat b_q^{(-i)}L_{n,i}[V_q],
\label{eq:loo-coefs}
\end{equation}
where \(i=0\) corresponds to the full-data coefficients (\(L_{n,0}=P_n\)). 
If \(L_{n,i}[V_q^2]=(L_{n,i}[V_q])^2\) for some \(i\), set \(b_q^{(-i)}=0\) and \(a_q^{(-i)}=L_{n,i}[Y]\).

For every \(q\in\mathcal D_0\) and \(i\in[n]\), define the LOO fitted value
\[
\widehat\eta_{iq}
=\widehat a_q^{(-i)}+\widehat b_q^{(-i)}V_{iq},
\]
and write \(\widehat\eta_q=(\widehat\eta_{iq})_{i\in[n]}\). Given
\(\mathcal D\subseteq\mathcal D_0\) and a penalty \(\lambda>0\), let
\((\widehat\theta_0,\widehat\theta)\) be any solution to
\begin{equation}
\min_{\substack{\theta_0\in\mathbb R\\
                  \theta\in\mathbb R_{\geq0}^{\mathcal D}}}
F_{\mathcal D}(\theta_0,\theta),
\qquad
F_{\mathcal D}(\theta_0,\theta)
=P_n\left[\left(Y-\theta_0-
\sum_{q\in\mathcal D}\theta_q\widehat\eta_q\right)^2\right]
+\lambda\sum_{q\in\mathcal D}\theta_q.
\label{eq:objective}
\end{equation}
The corresponding coefficients on the original scale are
\begin{equation}
\widehat\xi_q=\widehat\theta_q\widehat b_q^{(0)}
\qquad \forall q\in\mathcal D,
\qquad
\widehat\xi_0=\widehat\theta_0+\sum_{q\in\mathcal D}\widehat\theta_q\,\widehat a_q^{(0)}.
\label{eq:map-back}
\end{equation}
The function \(F_{\mathcal D}\) is continuous and coercive on the nonempty closed
feasible set \(\mathbb R\times\mathbb R_{\geq0}^{\mathcal D}\). Hence, it attains
its minimum. It is also convex. All the
results below hold for every minimizer. When \(\mathcal D_0=\mathcal D=[p]\) and
\(V_j=X_j\), it is the
main-effect UniLasso fit in \texttt{uniPairs-2stage}. When
\(\mathcal D_0=[p]\cup\mathcal P\), \(V_j=X_j\),
\(V_{(j,k)}=X_jX_k\), and
\(\mathcal D=[p]\cup\widehat\Gamma\), it is the UniLasso fit in
\texttt{uniPairs}.

\paragraph{Assumptions.} Fix $\nu\in[\frac12,1]$ and constants
$K,c_0,\eta,m,K_\xi>0$.
\begin{enumerate}[label=(\Alph*),leftmargin=*]
\item[(T1)] $\|Y\|_{\psi_\nu}\le K$,
\[
\|V_q\|_{\psi_\nu}\le K,\qquad
\|V_q^2\|_{\psi_\nu}\le K,\qquad
\|V_qY\|_{\psi_\nu}\le K
\qquad \forall q\in\mathcal D_0,
\]
and
\[
\|V_qu\|_{\psi_\nu}\le K
\qquad \forall q\in S'.
\]
\item[(T2)] $\|Y^2\|_{\psi_{1/2}}\le K$, and
\[
\|V_qV_{q'}\|_{\psi_{1/2}}\le K
\qquad \forall (q,q')\in S'\times S'.
\]
\item[(V)] $\Var(V_q)\ge c_0\quad \forall q\in\mathcal D_0$.
\item[(E)] $\lambda_{\min}\big(\E(V_{S'}V_{S'}^\top)\big)\ge\eta$.
\item[(S)] $\xi_q^*b_q>0$, $|b_q|\ge m$, and
$|\xi_q^*|\le K_\xi\quad \forall q\in S$.
\end{enumerate}
Define
\begin{equation}
\Lambda_0=\max_{q\in\mathcal D_0\setminus S}|b_q|,\qquad
\Lambda_1=\max_{q\in S}|\E(V_qu)|
\label{eq:level-m}
\end{equation}
where a maximum over an empty set is defined as $0$, and recall that
\[
\phi_\nu(n,\lambda)=\min(n\lambda^2,(n\lambda)^{\nu},n^{1/2}).
\]

\begin{theorem}
\label{th:main}
Under {\rm(T1)}, {\rm(T2)}, {\rm(V)}, {\rm(E)} and {\rm(S)}, there are
constants $K_1,\dots,K_5>0$ depending only on $(\nu,K,c_0,\eta,m,s,K_\xi)$ such that
if
\begin{equation}
K_1\Lambda_0\le\lambda\le K_2,
\label{eq:lambda-m}
\end{equation}
then there exists an event $\mathcal E$ satisfying
\[
\Prob(\mathcal E)
\geq1-K_4d_0n\exp\left(-K_5\phi_\nu(n,\lambda)\right),
\]
such that for every
$\mathcal D$ satisfying $S\subseteq\mathcal D\subseteq\mathcal D_0$ and every
minimizer \((\widehat\theta_0,\widehat\theta)\) of $F_{\mathcal D}$ :
\begin{enumerate}[label=(\roman*),leftmargin=*]
\item $\forall q\in\mathcal D\setminus S, \quad \widehat\theta_q=\widehat\xi_q=0$,
\item $\max_{q\in S'}|\widehat\xi_q-\xi_q^*|\le K_3(\lambda+\Lambda_1)$,
\item
\[
\widehat\xi_q=0
\quad\text{or}\quad
\sign(\widehat\xi_q)=\sign(\xi_q^*)
\qquad \forall q\in S.
\]
Moreover, for all \(q\in S\), if 
\(|\xi_q^*|>K_3(\lambda+\Lambda_1)\), then
\(\widehat\xi_q\ne0\).
\end{enumerate}
\end{theorem}
Without loss of generality, we can assume $K_2\le1$. 
We complete the proof of Theorem~\ref{th:main} in four steps.

\subsubsection*{Step 1}

Under (T1), (T2), and (V), by Lemma~\ref{lem:orlicz}(a),
\begin{equation}
\begin{gathered}
\max\left(|\E(V_q)|,\E(V_q^2),|\E(Y)|,|\E(V_qY)|\right)
\leq B_1
\qquad \forall q\in\mathcal D_0,\\
\E(|V_qu|)\leq B_1
\qquad \forall q\in S',\\
\E(Y^2)\leq B_Y-1,
\end{gathered}
\label{eq:moment-bounds}
\end{equation}
where 
\[
B_1=C_{\nu,1}K,
\qquad
B_Y=C_{1/2,1}K+1.
\]
Therefore, for every \(q\in\mathcal D_0\),
\begin{equation}
\begin{gathered}
|\Cov(V_q,Y)|
\leq |\E(V_qY)|+|\E(V_q)|\,|\E(Y)|
\leq B_2:=B_1+B_1^2,\\
|b_q|
\leq\sqrt{\frac{\Var(Y)}{\Var(V_q)}}
\leq B_b:=\sqrt{\frac{B_Y}{c_0}},\\
|a_q|
\leq |\E(Y)|+|b_q|\,|\E(V_q)|
\leq B_1(1+B_b).
\end{gathered}
\label{eq:coef-bounds}
\end{equation}
For every \(x>0\), define the event
\begin{equation}
\mathcal A(x)
=\left\{
\max_{q\in\mathcal D_0}
\max_{i\in\{0\}\cup[n]}
\max_{f\in\{V_q,V_q^2,Y,V_qY\}}
\left|L_{n,i}[f]-\E(f)\right|
\leq x
\right\}.
\label{eq:eventA}
\end{equation}

\begin{lemma}
\label{lem:loo-coefs}
Assume {\rm(T1)}, {\rm(T2)} and {\rm(V)}.
\begin{enumerate}[label=(\roman*),leftmargin=*]
\item Let
\[
c_A=c_\nu\min\left((2K^2)^{-1},(2K)^{-\nu}\right).
\]
Then
\[
\Prob(\mathcal A(x)^c)
\leq8d_0(n+1)
\exp\left(-c_A\min\left(nx^2,(nx)^\nu\right)\right)
\qquad \forall x>0.
\]
\item Let
\[
x_0=\min\left(1,\frac{c_0}{4(B_1+1)}\right)
\]
and
\[
C_1=1+(2B_1+2)
\left(\frac{2}{c_0}+\frac{2B_2}{c_0^2}\right)(B_1+1)+B_b.
\]
Then, for every \(x\in(0,x_0]\), on \(\mathcal A(x)\), we have :
\begin{equation}
\begin{gathered}
L_{n,i}[V_q^2]-(L_{n,i}[V_q])^2\geq\frac{c_0}{2}
\qquad
\forall (q,i)\in\mathcal D_0\times(\{0\}\cup[n]),\\
|\widehat b_q^{(-i)}-b_q|\leq C_1x,
\qquad
|\widehat a_q^{(-i)}-a_q|\leq C_1x
\qquad
\forall (q,i)\in\mathcal D_0\times(\{0\}\cup[n]).
\end{gathered}
\label{eq:loo-coef-dev}
\end{equation}
In particular,
\[
|\widehat b_q^{(-i)}-\widehat b_q^{(0)}|\leq2C_1x,
\qquad
|\widehat a_q^{(-i)}-\widehat a_q^{(0)}|\leq2C_1x
\qquad
\forall (q,i)\in\mathcal D_0\times[n],
\]
and
\[
\max_{i\in\{0\}\cup[n]}|\widehat b_q^{(-i)}|
\leq|b_q|+C_1x,
\qquad
P_n[V_q^2]\leq B_1+1
\qquad \forall q\in\mathcal D_0.
\]
\end{enumerate}
\end{lemma}

\begin{proof}
For (i), fix \(q\in\mathcal D_0\),
\(i\in\{0\}\cup[n]\), and
\(f\in\{V_q,V_q^2,Y,V_qY\}\). By (T1),
\(\|f\|_{\psi_\nu}\leq K\). Let
\[
N_i=
\begin{cases}
n, & i=0,\\
n-1, & i\in[n].
\end{cases}
\]
Then \(L_{n,i}[f]\) is an average of \(N_i\) i.i.d.\ copies of \(f\).
Since \(n\geq3\), we have \(N_i\geq n/2\). By Lemma~\ref{lem:two-regime}, for every \(x>0\),
\begin{align*}
\Prob\left(|L_{n,i}[f]-\E(f)|>x\right)
&\leq2\exp\left(-c_\nu
\min\left(\frac{N_ix^2}{K^2},
\left(\frac{N_ix}{K}\right)^\nu\right)\right)\\
&\leq2\exp\left(-c_\nu
\min\left(\frac{nx^2}{2K^2},
\left(\frac{nx}{2K}\right)^\nu\right)\right)\\
&\leq2\exp\left(-c_A\min\left(nx^2,(nx)^\nu\right)\right).
\end{align*}
Applying the union bound proves (i).

For (ii), fix \(q\in\mathcal D_0\),
\(i\in\{0\}\cup[n]\), and \(x\in(0,x_0]\). Define
\[
\begin{gathered}
\widehat m_1=L_{n,i}[V_q],
\qquad \widehat m_2=L_{n,i}[V_q^2],
\qquad \widehat m_Y=L_{n,i}[Y],
\qquad \widehat m_{VY}=L_{n,i}[V_qY],\\
m_1=\E(V_q),
\qquad m_2=\E(V_q^2),
\qquad m_Y=\E(Y),
\qquad m_{VY}=\E(V_qY).
\end{gathered}
\]
Also, define
\[
\begin{aligned}
\widehat D&=\widehat m_2-\widehat m_1^2,
&D&=m_2-m_1^2=\Var(V_q),\\
\widehat N&=\widehat m_{VY}-\widehat m_1\widehat m_Y,
&N&=m_{VY}-m_1m_Y=\Cov(V_q,Y).
\end{aligned}
\]
Therefore,
\[
\widehat b_q^{(-i)}=\frac{\widehat N}{\widehat D}
\quad\text{and}\quad
b_q=\frac{N}{D}.
\]

Suppose that \(\mathcal A(x)\) holds. Since \(x\leq x_0\leq1\),
by \eqref{eq:moment-bounds},
\[
|\widehat m_1|\leq |m_1|+x\leq B_1+1,
\qquad
|\widehat m_Y|\leq |m_Y|+x\leq B_1+1.
\]
Therefore,
\begin{align*}
|\widehat m_1^2-m_1^2|
&=|\widehat m_1-m_1|\,|\widehat m_1+m_1|
\leq (2B_1+1)x,\\
|\widehat D-D|
&\leq|\widehat m_2-m_2|+|\widehat m_1^2-m_1^2|
\leq(2B_1+2)x,\\
|\widehat N-N|
&\leq|\widehat m_{VY}-m_{VY}|
+|\widehat m_1|\,|\widehat m_Y-m_Y|
+|m_Y|\,|\widehat m_1-m_1|\\
&\leq(2B_1+2)x.
\end{align*}
By (V), \(D\geq c_0\). Moreover,
\[
(2B_1+2)x\leq\frac{c_0}{2}.
\]
Consequently,
\[
\widehat D\geq D-|\widehat D-D|\geq\frac{c_0}{2}.
\]
Using
\[
\widehat b_q^{(-i)}-b_q
=\frac{\widehat N-N}{\widehat D}
+N\frac{D-\widehat D}{\widehat D D},
\]
with \(|N|\leq B_2\), \(D\geq c_0\), and
\(\widehat D\geq c_0/2\), we get
\[
|\widehat b_q^{(-i)}-b_q|
\leq (2B_1+2)
\left(\frac{2}{c_0}+\frac{2B_2}{c_0^2}\right)x
=C_1^{\mathrm b}x,
\]
where
\[
C_1^{\mathrm b}
=(2B_1+2)\left(\frac{2}{c_0}+\frac{2B_2}{c_0^2}\right).
\]
Moreover,
\[
\widehat a_q^{(-i)}-a_q
=(\widehat m_Y-m_Y)
-(\widehat b_q^{(-i)}-b_q)\widehat m_1
-b_q(\widehat m_1-m_1).
\]
Using \(|\widehat m_1|\leq B_1+1\) and \(|b_q|\leq B_b\),
\[
|\widehat a_q^{(-i)}-a_q|
\leq x+C_1^{\mathrm b}x(B_1+1)+B_bx
=C_1x.
\]
Since \(C_1\geq C_1^{\mathrm b}\), we also get
\(|\widehat b_q^{(-i)}-b_q|\leq C_1x\). Therefore, using the estimates with \(i\)
and \(0\) yields
\[
|\widehat b_q^{(-i)}-\widehat b_q^{(0)}|
\leq2C_1x,
\qquad
|\widehat a_q^{(-i)}-\widehat a_q^{(0)}|
\leq2C_1x.
\]
The two remaining parts follow from
\[
|\widehat b_q^{(-i)}|
\leq|b_q|+|\widehat b_q^{(-i)}-b_q|
\leq|b_q|+C_1x
\]
and
\[
P_n[V_q^2]
\leq\E(V_q^2)+x
\leq B_1+1.
\]
This completes the proof.
\end{proof}

\subsubsection*{Step 2}

\begin{lemma}
\label{lem:kkt}
Let \(\mathcal D\subseteq\mathcal D_0\), \(\lambda>0\),
\((\widehat\theta_0,\widehat\theta)\) a minimizer of
\(F_{\mathcal D}\). Define the residuals
\[
r_i=Y_i-\widehat\theta_0
-\sum_{q\in\mathcal D}\widehat\theta_q\widehat\eta_{iq}
\qquad \forall i\in[n].
\]
Then,
\begin{enumerate}[label=(\alph*),leftmargin=*]
\item $P_n[r]=0$.
\item
\[
2P_n[(\widehat\eta_q-c)r]\leq\lambda
\qquad \forall(q,c)\in\mathcal D\times\mathbb R.
\]
Moreover, equality holds for every \(c\in\mathbb R\) whenever
\(\widehat\theta_q>0\).
\item $P_n[r^2]\le P_n[Y^2]$.
\end{enumerate}
\end{lemma}

\begin{proof}
The function \(F_{\mathcal D}\) is convex and differentiable. At
\((\widehat\theta_0,\widehat\theta)\), we have
\[
\frac{\partial F_{\mathcal D}}{\partial\theta_0}
=-2P_n[r]
\]
and
\[
\frac{\partial F_{\mathcal D}}{\partial\theta_q}
=-2P_n[\widehat\eta_qr]+\lambda,
\qquad \forall q\in\mathcal D.
\]
Because \(\theta_0\) is unconstrained, the first-order optimality condition gives
\[
\frac{\partial F_{\mathcal D}}{\partial\theta_0}=0.
\]
Therefore, \(P_n[r]=0\). Fix \(q\in\mathcal D\). If \(\widehat\theta_q=0\), then
\[
\frac{\partial F_{\mathcal D}}{\partial\theta_q}\geq0.
\]
Otherwise \(\widehat\theta_q>0\), in which case
\[
\frac{\partial F_{\mathcal D}}{\partial\theta_q}=0.
\]
Thus,
\[
2P_n[\widehat\eta_qr]\leq\lambda,
\]
with equality whenever \(\widehat\theta_q>0\). For every \(c\in\mathbb R\),
by (a),
\[
P_n[(\widehat\eta_q-c)r]
=P_n[\widehat\eta_qr]-cP_n[r]
=P_n[\widehat\eta_qr].
\]
Since \((0,0)\) is feasible and \((\widehat\theta_0,\widehat\theta)\) is optimal, we get
\[
F_{\mathcal D}(\widehat\theta_0,\widehat\theta)
\leq F_{\mathcal D}(0,0)
=P_n[Y^2].
\]
Since \(\widehat\theta_q\geq0\) for every \(q\in\mathcal D\),
\[
F_{\mathcal D}(\widehat\theta_0,\widehat\theta)
=P_n[r^2]+\lambda\sum_{q\in\mathcal D}\widehat\theta_q
\geq P_n[r^2],
\]
This concludes the proof.
\end{proof}

\subsubsection*{Step 3}

Let $\mathcal B=\{P_n[Y^2]\le B_Y\}$.

\begin{lemma}
\label{lem:nofp}
Assume {\rm(T1)}, {\rm(T2)} and {\rm(V)}, and let
\[
K_1=12\sqrt{B_Y(B_1+1)},
\qquad
K_1'=\max\left(1,12C_1\sqrt{B_Y(B_1+1)}\right).
\]
Take \(\lambda>0\) such that \(\lambda\geq K_1\Lambda_0\). Let
\(x=\lambda/K_1'\), and assume that \(x\leq x_0\). On
\(\mathcal A(x)\cap\mathcal B\), for every \(\mathcal D\subseteq\mathcal D_0\) and every minimizer \((\widehat\theta_0,\widehat\theta)\) of
\(F_{\mathcal D}\):
\[
\widehat\theta_q=\widehat\xi_q=0 \qquad \forall q\in\mathcal D\setminus S.
\]
\end{lemma}

\begin{proof}
Fix \(\mathcal D\subseteq\mathcal D_0\) and a minimizer
\((\widehat\theta_0,\widehat\theta)\) of \(F_{\mathcal D}\). Suppose that
\(\widehat\theta_q>0\) for some \(q\in\mathcal D\setminus S\). By
Lemma~\ref{lem:kkt}(b), with \(c=\widehat a_q^{(0)}\),
\[
\frac{\lambda}{2}
=P_n[(\widehat\eta_q-\widehat a_q^{(0)})r].
\]
By Cauchy-Schwarz and Lemma~\ref{lem:kkt}(c),
\[
\frac{\lambda}{2}
\leq
\left(P_n[(\widehat\eta_q-\widehat a_q^{(0)})^2]\right)^{1/2}
\left(P_n[r^2]\right)^{1/2}
\leq
\left(P_n[(\widehat\eta_q-\widehat a_q^{(0)})^2]\right)^{1/2}
\left(P_n[Y^2]\right)^{1/2}.
\]
For every \(i\in[n]\),
\[
\widehat\eta_{iq}-\widehat a_q^{(0)}
=(\widehat a_q^{(-i)}-\widehat a_q^{(0)})
+\widehat b_q^{(-i)}V_{iq}.
\]
Using \((a+b)^2\leq2a^2+2b^2\), we get
\[
\begin{aligned}
P_n[(\widehat\eta_q-\widehat a_q^{(0)})^2]
&\leq2\max_{i\in[n]}
(\widehat a_q^{(-i)}-\widehat a_q^{(0)})^2 + 2\max_{i\in[n]}(\widehat b_q^{(-i)})^2P_n[V_q^2].
\end{aligned}
\]
On \(\mathcal A(x)\), by Lemma~\ref{lem:loo-coefs}(ii),
\[
\max_{i\in[n]}|\widehat a_q^{(-i)}-\widehat a_q^{(0)}|
\leq2C_1x,
\qquad
\max_{i\in[n]}|\widehat b_q^{(-i)}|
\leq|b_q|+C_1x,
\]
and \(P_n[V_q^2]\leq B_1+1\). Since
\(q\in\mathcal D\setminus S\), by definition of \(\Lambda_0\),
\(|b_q|\leq\Lambda_0\). Thus,
\[
P_n[(\widehat\eta_q-\widehat a_q^{(0)})^2]
\leq8C_1^2x^2
+2(B_1+1)(\Lambda_0+C_1x)^2.
\]
On \(\mathcal B\), \(P_n[Y^2]\leq B_Y\). Therefore,
\begin{align*}
\lambda^2
&\leq4B_Y\left(8C_1^2x^2
+2(B_1+1)(\Lambda_0+C_1x)^2\right)\\
&\leq4B_Y\lambda^2\left(
\frac{8C_1^2}{K_1'^2}
+2(B_1+1)\left(\frac1{K_1}+\frac{C_1}{K_1'}\right)^2
\right).
\end{align*}
The last inequality follows from \(x=\lambda/K_1'\) and
\(\Lambda_0\leq\lambda/K_1\). By definitions of \(K_1'\) and \(K_1\),
\[
4B_Y\frac{8C_1^2}{K_1'^2}
\leq\frac{2}{9(B_1+1)}
\leq\frac{2}{9}
<\frac14,
\]
and
\[
\frac1{K_1}+\frac{C_1}{K_1'}
\leq\frac{1}{6\sqrt{B_Y(B_1+1)}}.
\]
Therefore,
\[
4B_Y\,2(B_1+1)
\left(\frac1{K_1}+\frac{C_1}{K_1'}\right)^2
\leq\frac{2}{9}
<\frac14.
\]
Consequently, \(\lambda^2\leq\lambda^2/2\), which contradicts
\(\lambda>0\). Thus, \(\widehat\theta_q=0\) for every
\(q\in\mathcal D\setminus S\). By \eqref{eq:map-back},
\[
\widehat\xi_q
=\widehat\theta_q\widehat b_q^{(0)}
=0
\qquad \forall q\in\mathcal D\setminus S.
\]
\end{proof}

\subsubsection*{Step 4}

Set \(V_{i0}=1\quad \forall i\in[n]\). Write
\[
V_{S'}=(V_q)_{q\in S'}
\quad\text{and}\quad
V_{i,S'}=(V_{iq})_{q\in S'}
\qquad \forall i\in[n].
\]
Define the population and empirical Gram matrices by
\[
G=\E(V_{S'}V_{S'}^\top)
\quad\text{and}\quad
\widehat G=\frac{1}{n}\sum_{i=1}^nV_{i,S'}V_{i,S'}^\top.
\]
Let
\[
\mathcal C_1
=\left\{\lambda_{\min}(\widehat G)\geq\frac{\eta}{2}\right\},
\]
\[
\mathcal C_2(x)
=\left\{
\max_{q\in S'}\left|P_n[V_qu]-\E(V_qu)\right|\leq x
\right\}
\qquad \forall x>0,
\]
and
\[
\mathcal C_3
=\left\{
\max_{q\in S'}P_n[|V_qu|]\leq B_1+1
\right\}.
\]

\begin{lemma}
\label{lem:linf}
Assume {\rm(T1)}, {\rm(T2)}, {\rm(V)}, {\rm(E)} and {\rm(S)}. Then, there exist
constants \(C_4,C_5>0\), depending only on
\((\nu,K,c_0,m,s,K_\xi)\), such that, for any \(\mathcal D\)
satisfying \(S\subseteq\mathcal D\subseteq\mathcal D_0\), any minimizer
\((\widehat\theta_0,\widehat\theta)\) of \(F_{\mathcal D}\) satisfying
\[
\widehat\theta_q=0
\qquad \forall q\in\mathcal D\setminus S,
\]
and any \(x\) satisfying \(0<x\leq\min(x_0,x_2)\), on
\(\mathcal A(x)\cap\mathcal C_1\cap\mathcal C_2(x)\cap\mathcal C_3\), we have:
\begin{equation}
\Delta
\leq\frac{4}{\eta}
\left(\frac{s\lambda}{2m}+(s+1)(\Lambda_1+x)+C_4x\right),
\label{eq:linf-bound}
\end{equation}
and
\[
\widehat\xi_q=0
\quad\text{or}\quad
\sign(\widehat\xi_q)=\sign(\xi_q^*)
\qquad \forall q\in S.
\]
where
\[
x_2=\min\left(\frac{m}{2C_1},\frac{\eta}{4C_5}\right)
\qquad\text{and}\qquad
\Delta=\max_{q\in S'}|\widehat\xi_q-\xi_q^*|.
\]
\end{lemma}

\begin{proof}
Throughout the proof, \(0<x\leq\min(x_0,x_2)\), and we work on the event 
\[
\mathcal A(x)\cap\mathcal C_1\cap\mathcal C_2(x)\cap\mathcal C_3.
\]

By Lemma~\ref{lem:loo-coefs}(ii), for every
\((q,i)\in S\times(\{0\}\cup[n])\),
\begin{equation}
|\widehat a_q^{(-i)}-\widehat a_q^{(0)}|\leq2C_1x,
\quad
|\widehat b_q^{(-i)}-\widehat b_q^{(0)}|\leq2C_1x,
\quad
|\widehat b_q^{(-i)}-b_q|\leq C_1x.
\label{eq:P1}
\end{equation}
Moreover,
\[
|\widehat b_q^{(0)}-b_q|
\leq C_1x
\leq C_1x_2
\leq\frac{m}{2}
\leq\frac{|b_q|}{2}
\qquad \forall q\in S.
\]
If \(b_q>0\), then
\[
\widehat b_q^{(0)}
\geq b_q-|\widehat b_q^{(0)}-b_q|
\geq\frac{|b_q|}{2}>0.
\]
Otherwise \(b_q<0\), in which case 
\[
\widehat b_q^{(0)}
\leq b_q+|\widehat b_q^{(0)}-b_q|
\leq-\frac{|b_q|}{2}<0.
\]
Hence, \(\sign(\widehat b_q^{(0)})=\sign(b_q)\). By (S), \(\xi_q^*b_q>0\). Therefore, 
\begin{equation}
|\widehat b_q^{(0)}|\geq|b_q|-C_1x\geq\frac{m}{2},
\qquad
\sign(\widehat b_q^{(0)})=\sign(b_q)=\sign(\xi_q^*)
\qquad \forall q\in S.
\label{eq:P2}
\end{equation}
Since \(V_0=1\), by Lemma~\ref{lem:loo-coefs}(ii) and Cauchy-Schwarz,
\begin{equation}
\begin{gathered}
P_n[V_q^2]\leq B_1+1,
\qquad
P_n[|V_q|]\leq\sqrt{B_1+1}\leq B_1+1
\qquad \forall q\in S',\\
P_n[|V_qV_{q'}|]
\leq\left(P_n[V_q^2]P_n[V_{q'}^2]\right)^{1/2}
\leq B_1+1
\qquad \forall(q,q')\in S'\times S'.
\end{gathered}
\label{eq:P3}
\end{equation}
Since \(\widehat\theta_q\geq0\) and
\(\widehat\xi_q=\widehat\theta_q\widehat b_q^{(0)}\), \eqref{eq:P2} implies
\[
\widehat\xi_q=0
\quad\text{or}\quad
\sign(\widehat\xi_q)=\sign(\xi_q^*)
\qquad \forall q\in S.
\]
It also implies
\begin{equation}
\widehat\theta_q=\frac{|\widehat\xi_q|}{|\widehat b_q^{(0)}|}
\leq\frac{2(|\xi_q^*|+\Delta)}{m}
\leq\frac{2(K_\xi+\Delta)}{m}
\qquad \forall q\in S.
\label{eq:P4}
\end{equation}
Since \(\widehat\theta_q=0\) for every \(q\in\mathcal D\setminus S\),
\[
r_i=Y_i-\widehat\theta_0
-\sum_{q\in S}\widehat\theta_q
\left(\widehat a_q^{(-i)}+\widehat b_q^{(-i)}V_{iq}\right).
\]
By \eqref{eq:active-model} and \eqref{eq:map-back},
\begin{equation}
\begin{gathered}
r_i=e_i+u_i-w_i
\qquad \forall i\in[n],\\
e_i=\sum_{q\in S'}(\xi_q^*-\widehat\xi_q)V_{iq},\\
w_i=\sum_{q\in S}\widehat\theta_q
\left((\widehat a_q^{(-i)}-\widehat a_q^{(0)})
+(\widehat b_q^{(-i)}-\widehat b_q^{(0)})V_{iq}\right).
\end{gathered}
\label{eq:resid-decomp}
\end{equation}
By \eqref{eq:P1},
\begin{equation}
|e_i|\leq\Delta\sum_{q'\in S'}|V_{iq'}|,
\qquad
|w_i|\leq2C_1x\sum_{q\in S}\widehat\theta_q(1+|V_{iq}|)
\qquad \forall i\in[n].
\label{eq:w-bounds}
\end{equation}
For every \(q\in S\), define
\[
\delta_q=\frac{\xi_q^*-\widehat\xi_q}{b_q}.
\]
Since \(|b_q|\geq m\),
\[
|\delta_q|\leq\frac{\Delta}{m}
\qquad \forall q\in S.
\]
Moreover,
\begin{equation}
\delta_q\left(\lambda
-2P_n[(\widehat\eta_q-\widehat a_q^{(0)})r]\right)\geq0
\qquad \forall q\in S.
\label{eq:kkt-signed}
\end{equation}
In fact, if \(\widehat\theta_q>0\), the second factor in
\eqref{eq:kkt-signed} is zero by Lemma~\ref{lem:kkt}(b). If
\(\widehat\theta_q=0\), then \(\widehat\xi_q=0\) and
\(\delta_q=\xi_q^*/b_q>0\) by (S), while the second factor is nonnegative by
Lemma~\ref{lem:kkt}(b).
For every \((q,i)\in S\times[n]\), define
\[
\rho_{iq}
=(\widehat b_q^{(-i)}-b_q)V_{iq}
+(\widehat a_q^{(-i)}-\widehat a_q^{(0)}).
\]
Then
\[
\widehat\eta_{iq}-\widehat a_q^{(0)}=b_qV_{iq}+\rho_{iq},
\qquad
|\rho_{iq}|\leq C_1x(|V_{iq}|+2),
\]
where the second inequality follows from \eqref{eq:P1}. Since
\(\delta_qb_q=\xi_q^*-\widehat\xi_q\), by \eqref{eq:kkt-signed},
\begin{equation}
(\xi_q^*-\widehat\xi_q)P_n[V_qr]
\leq\frac{\lambda}{2}|\delta_q|
+|\delta_q|\,|P_n[\rho_qr]|
\qquad \forall q\in S,
\label{eq:kkt-q}
\end{equation}
where \(\rho_q=(\rho_{iq})_{i\in[n]}\). For \(q=0\), by Lemma~\ref{lem:kkt}(a),
\[
(\xi_0^*-\widehat\xi_0)P_n[V_0r]
=(\xi_0^*-\widehat\xi_0)P_n[r]
=0.
\]
Summing \eqref{eq:kkt-q} over \(q\in S\), adding the previous identity for
\(q=0\), and using \eqref{eq:resid-decomp}, we get
\[
\begin{aligned}
P_n[er]
&=P_n[e^2]+P_n[eu]-P_n[ew]\\
&\leq\frac{\lambda}{2}\sum_{q\in S}|\delta_q|
+\sum_{q\in S}|\delta_q|\,|P_n[\rho_qr]|.
\end{aligned}
\]
Let
\[
v=(\xi_q^*-\widehat\xi_q)_{q\in S'}.
\]
Then, \(P_n[e^2]=v^\top\widehat Gv\). On \(\mathcal C_1\),
\[
P_n[e^2]
\geq\frac{\eta}{2}\|v\|_2^2
\geq\frac{\eta}{2}\Delta^2.
\]
Therefore,
\begin{equation}
\frac{\eta}{2}\Delta^2
\leq\frac{\lambda}{2}\sum_{q\in S}|\delta_q|
+\sum_{q\in S}|\delta_q|\,|P_n[\rho_qr]|
+|P_n[eu]|+|P_n[ew]|.
\label{eq:ineq-m}
\end{equation}
We have
\[
\frac{\lambda}{2}\sum_{q\in S}|\delta_q|
\leq\frac{s\lambda}{2m}\Delta.
\]
Since \(|\E(V_qu)|\leq\Lambda_1\) for every \(q\in S\) and
\(\E(V_0u)=\E(u)=0\), on \(\mathcal C_2(x)\) we have
\[
|P_n[eu]|
\leq\Delta\sum_{q\in S'}|P_n[V_qu]|
\leq(s+1)(\Lambda_1+x)\Delta.
\]
By \eqref{eq:w-bounds}, \eqref{eq:P3} and \eqref{eq:P4},
\begin{align*}
|P_n[ew]|
&\leq2C_1x\Delta\sum_{q\in S}\widehat\theta_q
P_n\left[(1+|V_q|)\sum_{q'\in S'}|V_{q'}|\right]\\
&\leq2C_1x\Delta
\frac{2s(K_\xi+\Delta)}{m}
2(s+1)(B_1+1)\\
&=C_2x\Delta(K_\xi+\Delta),
\end{align*}
where
\[
C_2=\frac{8C_1(B_1+1)s(s+1)}{m}.
\]
For every \(q\in S\), by \eqref{eq:resid-decomp},
\[
|P_n[\rho_qr]|
\leq P_n[|\rho_q||e|]
+P_n[|\rho_q||u|]
+P_n[|\rho_q||w|].
\]
Therefore, on the event \(\mathcal C_3\), we have
\begin{align*}
P_n[|\rho_q||e|]
&\leq C_1x\Delta
P_n\left[(|V_q|+2)\sum_{q'\in S'}|V_{q'}|\right]\\
&\leq3C_1(s+1)(B_1+1)x\Delta,\\
P_n[|\rho_q||u|]
&\leq C_1x\left(P_n[|V_qu|]+2P_n[|u|]\right)\\
&\leq3C_1(B_1+1)x,\\
P_n[|\rho_q||w|]
&\leq2C_1^2x^2\sum_{q''\in S}\widehat\theta_{q''}
P_n[(|V_q|+2)(1+|V_{q''}|)]\\
&\leq\frac{24C_1^2(B_1+1)s}{m}x^2(K_\xi+\Delta).
\end{align*}
Since \(|\delta_q|\leq\Delta/m\), \(|S|\leq s\) and \(x\leq1\), we have
\[
\sum_{q\in S}|\delta_q|\,|P_n[\rho_qr]|
\leq\frac{3C_1(B_1+1)s}{m}x\Delta
\left((s+1)\Delta+1+\frac{8C_1s}{m}(K_\xi+\Delta)\right).
\]
Define
\begin{equation}
\begin{aligned}
C_4&=C_2K_\xi
+\frac{3C_1(B_1+1)s}{m}
\left(1+\frac{8C_1sK_\xi}{m}\right),\\
C_5&=C_2
+\frac{3C_1(B_1+1)s}{m}
\left(s+1+\frac{8C_1s}{m}\right).
\end{aligned}
\label{eq:C4C5}
\end{equation}
Substituting the last four bounds into \eqref{eq:ineq-m}, we obtain
\[
\frac{\eta}{2}\Delta^2
\leq\Delta\left(
\frac{s\lambda}{2m}+(s+1)(\Lambda_1+x)+C_4x
\right)+C_5x\Delta^2.
\]
Since \(x\leq x_2\leq\eta/(4C_5)\), we have
\[
\frac{\eta}{4}\Delta^2
\leq\Delta\left(
\frac{s\lambda}{2m}+(s+1)(\Lambda_1+x)+C_4x
\right).
\]
which concludes the proof of \eqref{eq:linf-bound} (\(\Delta \geq 0\)).
\end{proof}

\subsubsection*{Proof of Theorem~\ref{th:main}}

\begin{proof}
Let \(K_1,K_1'\) be the constants in Lemma~\ref{lem:nofp}, \(x_0\) as in Lemma~\ref{lem:loo-coefs}, and \(x_2\) as in
Lemma~\ref{lem:linf}. Define
\begin{align*}
K_2&=\min\left(1,K_1'\min(x_0,x_2)\right),\\
K_3&=\frac{4}{\eta}\max\left(
\frac{s}{2m}+\frac{s+1+C_4}{K_1'},s+1
\right).
\end{align*}
Fix a penalty \(\lambda\) satisfying \eqref{eq:lambda-m}, and set
\[
x=\frac{\lambda}{K_1'}.
\]
Since \(\lambda\leq K_2\leq K_1'\min(x_0,x_2)\), we have
\(0<x\leq\min(x_0,x_2)\). Define the event 
\begin{equation}
\mathcal E
=\mathcal A(x)\cap\mathcal B\cap\mathcal C_1
\cap\mathcal C_2(x)\cap\mathcal C_3.
\label{eq:event-m}
\end{equation}

Fix \(\mathcal D\) satisfying
\(S\subseteq\mathcal D\subseteq\mathcal D_0\), and let 
\((\widehat\theta_0,\widehat\theta)\) be any minimizer of
\(F_{\mathcal D}\). Since \(\lambda\geq K_1\Lambda_0\),
by Lemma~\ref{lem:nofp},
\[
\widehat\theta_q=\widehat \xi_q=0
\qquad \forall q\in\mathcal D\setminus S,
\]
which proves (i). By Lemma~\ref{lem:linf},
\begin{align*}
\Delta
&\leq\frac{4}{\eta}\left(
\frac{s\lambda}{2m}+(s+1)\Lambda_1
+\frac{(s+1+C_4)\lambda}{K_1'}
\right)\\
&\leq K_3(\lambda+\Lambda_1),
\end{align*}
which proves (ii). The first assertion in (iii) follows from Lemma~\ref{lem:linf}. Moreover, if \(q\in S\) satisfies
\[
|\xi_q^*|>K_3(\lambda+\Lambda_1),
\]
then by (ii) :
\[
|\widehat\xi_q-\xi_q^*|<|\xi_q^*|.
\]
Therefore, \(\widehat\xi_q\neq0\), which proves the second assertion in
(iii). The event in \eqref{eq:event-m} does not depend on
\(\mathcal D\) or \((\widehat\theta_0,\widehat\theta)\). Hence, these conclusions
hold simultaneously for every such \(\mathcal D\) and every minimizer of \(F_{\mathcal D}\).
Next, we upper-bound \(\Prob(\mathcal E^c)\). First, Lemma~\ref{lem:loo-coefs}(i), \(x=\lambda/K_1'\), and
\(K_1'\geq1\) imply
\begin{align}
\Prob(\mathcal A(x)^c)
&\leq8d_0(n+1)\exp\left(
-c_A\min\left(
\frac{n\lambda^2}{(K_1')^2},
\left(\frac{n\lambda}{K_1'}\right)^\nu
\right)\right)\notag\\
&\leq16d_0n\exp\left(
-\frac{c_A}{(K_1')^2}\phi_\nu(n,\lambda)
\right).
\label{eq:m-prob-A}
\end{align}
Since \(\E(Y^2)\leq B_Y-1\),
\[
\mathcal B^c
\subseteq
\left\{P_n[Y^2]-\E(Y^2)>1\right\}.
\]
By Assumption (T2) and Lemma~\ref{lem:two-regime},
\begin{equation}
\Prob(\mathcal B^c)
\leq2\exp(-c_B\sqrt n),
\label{eq:m-prob-B}
\end{equation}
where \(c_B>0\) depends only on \(K\).
For the empirical Gram matrix, by Weyl's inequality :
\[
\lambda_{\min}(\widehat G)
\geq\lambda_{\min}(G)-\|\widehat G-G\|_{\mathrm{op}}.
\]
Moreover,
\[
\|\widehat G-G\|_{\mathrm{op}}
\leq\|\widehat G-G\|_F
\leq(s+1)\|\widehat G-G\|_{\max}.
\]
Since \(\lambda_{\min}(G)\geq\eta\) by (E), we obtain
\[
\mathcal C_1^c
\subseteq
\left\{
\max_{(q,q')\in S'\times S'}
\left|P_n[V_qV_{q'}]-\E(V_qV_{q'})\right|
>\frac{\eta}{2(s+1)}
\right\}.
\]
By (T2), Lemma~\ref{lem:two-regime}, and a union bound over
\((q,q')\in S'\times S'\), there exists \(c_G>0\), depending only on
\((K,\eta,s)\), such that
\begin{equation}
\Prob(\mathcal C_1^c)
\leq2(s+1)^2\exp(-c_G\sqrt n).
\label{eq:m-prob-C1}
\end{equation}
For \(\mathcal C_2(x)\), (T1), Lemma~\ref{lem:two-regime}, and a union
bound over \(q\in S'\) imply
\begin{align}
\Prob(\mathcal C_2(x)^c)
&\leq2(s+1)\exp\left(
-c_\nu\min\left(
\frac{n\lambda^2}{(K_1')^2K^2},
\left(\frac{n\lambda}{K_1'K}\right)^\nu
\right)\right)\notag\\
&\leq2(s+1)\exp\left(-c_2\phi_\nu(n,\lambda)\right),
\label{eq:m-prob-C2}
\end{align}
where \(c_2>0\) depends only on \((\nu,K,K_1')\).
Finally, \(\||V_qu|\|_{\psi_\nu}=\|V_qu\|_{\psi_\nu}\leq K\) and
\(\E(|V_qu|)\leq B_1\) for every \(q\in S'\). Hence,
\[
\mathcal C_3^c
\subseteq
\left\{
\max_{q\in S'}\left(P_n[|V_qu|]-\E(|V_qu|)\right)>1
\right\}.
\]
Since \(\nu\geq1/2\), Lemma~\ref{lem:two-regime} and a union bound imply
\begin{equation}
\Prob(\mathcal C_3^c)
\leq2(s+1)\exp(-c_3\sqrt n),
\label{eq:m-prob-C3}
\end{equation}
where \(c_3>0\) depends only on \((\nu,K)\).
Let
\[
K_5=\min\left(
\frac{c_A}{(K_1')^2},c_B,c_G,c_2,c_3
\right).
\]
Since \(\phi_\nu(n,\lambda)\leq\sqrt n\), by
\eqref{eq:m-prob-A}-\eqref{eq:m-prob-C3}, we get :
\begin{align*}
\Prob(\mathcal E^c)
&\leq\left(16d_0n+2+2(s+1)^2+4(s+1)\right)
\exp\left(-K_5\phi_\nu(n,\lambda)\right)\\
&\leq K_4d_0n\exp\left(-K_5\phi_\nu(n,\lambda)\right),
\end{align*}
where
\[
K_4=18+2(s+1)^2+4(s+1).
\]
The constants \(K_1,\ldots,K_5\) depend only on
\((\nu,K,c_0,\eta,m,s,K_\xi)\). This concludes the proof.
\end{proof}

\subsubsection*{Proofs of Theorems~\ref{th:unipairs-2stage} and~\ref{th:unipairs}}

By Lemma~\ref{lem:orlicz}(b)-(d):
\begin{equation}
\begin{gathered}
\|X_j\|_{\psi_1}\leq C_{2,1}K_X
\qquad \forall j\in[p],\\
\|X_aX_b\|_{\psi_1}\leq K_X^2,
\qquad
\|X_aX_b-\E(X_aX_b)\|_{\psi_1}
\leq C^{\mathrm c}_1K_X^2
\qquad \forall(a,b)\in\mathcal P.
\end{gathered}
\label{eq:tails-b}
\end{equation}
Define
\[
K_Y
=K_\beta\max\left((\log 2)^{-1},C_{2,1}K_X,K_X^2\right)
+K_\varepsilon
\]
and
\[
K_u=K_\beta C^{\mathrm c}_1K_X^2+K_\varepsilon.
\]
Since \(\|\cdot\|_{\psi_1}\) is a norm and
\(\|1\|_{\psi_1}=(\log 2)^{-1}\), by Assumption~\ref{ass:B2},
\begin{equation}
\begin{gathered}
\|Y\|_{\psi_1}\leq K_Y,\\
\left\|
\sum_{(a,b)\in S_I}\gamma^*_{ab}
\left(X_aX_b-\E(X_aX_b)\right)+\varepsilon
\right\|_{\psi_1}
\leq K_u.
\end{gathered}
\label{eq:tails-y}
\end{equation}
Both \(K_Y\) and \(K_u\) depend only on
\((K_X,K_\varepsilon,K_\beta)\).

\paragraph{Proof of Theorem~\ref{th:unipairs-2stage} and Corollary~\ref{cor:support-1}.}
Apply Theorem~\ref{th:main} with
\[
\mathcal D_0=\mathcal D=[p],
\qquad
V_j=X_j,
\qquad
S=S_M,
\qquad
\xi_j^*=\beta_j^*\quad \forall j\in S_M,
\qquad
\xi_0^*=\beta_{0,M}^*,
\]
\[
u=Y-\beta^*_{0,M}-\sum_{j\in S_M}\beta_j^*X_j
\;=\sum_{(a,b)\in S_I}\gamma^*_{ab}
\left(X_aX_b-\E(X_aX_b)\right)+\varepsilon,
\]
and \(\nu=2/3\). Here, the estimator defined in
\eqref{eq:objective}-\eqref{eq:map-back} is the main-effect UniLasso
estimator in \texttt{uniPairs-2stage}. In particular,
\(\widehat\xi_0=\widehat\beta_0\) and
\(\widehat\xi_j=\widehat\beta_j\) for every \(j\in[p]\).
We next verify the assumptions of Theorem~\ref{th:main}. By the definition of \(\beta_{0,M}^*\) and
\(\E(\varepsilon)=\E(\E(\varepsilon\mid X))=0\), we have
\(\E(u)=0\). By Lemma~\ref{lem:orlicz}(b)-(c) and
\eqref{eq:tails-b}-\eqref{eq:tails-y},
\[
\begin{gathered}
\|X_j\|_{\psi_{2/3}}\leq C_{2,2/3}K_X,
\qquad
\|X_j^2\|_{\psi_{2/3}}\leq C_{1,2/3}K_X^2
\qquad \forall j\in[p],\\
\|Y\|_{\psi_{2/3}}\leq C_{1,2/3}K_Y,
\qquad
\|u\|_{\psi_{2/3}}\leq C_{1,2/3}K_u,\\
\|X_jY\|_{\psi_{2/3}}
\leq\|X_j\|_{\psi_2}\|Y\|_{\psi_1}
\leq K_XK_Y
\qquad \forall j\in[p],\\
\|X_ju\|_{\psi_{2/3}}
\leq\|X_j\|_{\psi_2}\|u\|_{\psi_1}
\leq K_XK_u
\qquad \forall j\in S_M.
\end{gathered}
\]
These are the bounds required for (T1). Moreover,
\[
\|Y^2\|_{\psi_{1/2}}
\leq\|Y\|_{\psi_1}^2
\leq K_Y^2,
\]
and
\[
\begin{gathered}
\|X_jX_k\|_{\psi_{1/2}}
\leq C_{1,1/2}K_X^2
\qquad \forall(j,k)\in[p]\times[p],\\
\|X_j\|_{\psi_{1/2}}\leq C_{2,1/2}K_X
\qquad \forall j\in[p],
\qquad
\|1\|_{\psi_{1/2}}=(\log 2)^{-2}.
\end{gathered}
\]
Thus, (T2) also holds. The same constant \(K\), depending only on
\((K_X,K_\varepsilon,K_\beta)\), may be chosen in (T1) and (T2).
Assumptions~\ref{ass:B3}, \ref{ass:B4}, and \ref{ass:B5} verify (V), (E),
and (S), respectively. To verify Assumption~\ref{ass:B5}, we take
\(K_\xi=K_\beta\), because
\[
\max_{j\in S_M}|\beta_j^*|
\leq\|\beta^*\|_1
\leq K_\beta.
\]
\(\Lambda_0\) is the first quantity
in \eqref{eq:leakage}. For every \(j\in S_M\),
\[
\begin{aligned}
\E(X_ju)
&=\Cov\left(
X_j,\sum_{(a,b)\in S_I}\gamma_{ab}^*X_aX_b
\right)+\E(X_j\varepsilon),\\
\E(X_j\varepsilon)
&=\E\left(X_j\E(\varepsilon\mid X)\right)=0.
\end{aligned}
\]
Therefore, \(\Lambda_1\) in \eqref{eq:level-m} is the
second quantity in \eqref{eq:leakage}.

Theorem~\ref{th:main}, with \(d_0=p\), shows that there exist constants
\(K_1,\ldots,K_5\), and an event of probability at least
\[
1-K_4pn\exp\left(-K_5\phi_{2/3}(n,\lambda)\right)
\]
on which every UniLasso solution satisfies
\[
\widehat\beta_j=0
\qquad \forall j\notin S_M,
\]
\[
\max_{j\in S_M}|\widehat\beta_j-\beta_j^*|
 \leq K_3(\lambda+\Lambda_1),
\qquad
|\widehat\beta_0-\beta^*_{0,M}|
 \leq K_3(\lambda+\Lambda_1),
\]
which proves \eqref{eq:concl-1}. If, in addition,
\[
\min_{j\in S_M}|\beta_j^*|>K_3(\lambda+\Lambda_1),
\]
part (iii) of Theorem~\ref{th:main} implies
\[
\widehat\beta_j\neq0,
\qquad
\sign(\widehat\beta_j)=\sign(\beta_j^*)
\qquad \forall j\in S_M.
\]
Together with \(\widehat\beta_j=0\) for every \(j\notin S_M\), this proves
Corollary~\ref{cor:support-1}.
\hfill\(\square\)

\paragraph{Proof of Theorem~\ref{th:unipairs} and Corollary~\ref{cor:fixed-screen}.}
Apply Theorem~\ref{th:main} with
\[
\mathcal D_0=[p]\cup\mathcal P,
\qquad
V_j=X_j\quad \forall j\in[p],
\qquad
V_{(j,k)}=X_jX_k\quad \forall(j,k)\in\mathcal P,
\]
\[
S=S_M\cup S_I,
\qquad
\xi_j^*=\beta_j^*\quad \forall j\in S_M,
\qquad
\xi_{(j,k)}^*=\gamma_{jk}^*\quad \forall(j,k)\in S_I,
\qquad
\xi_0^*=\beta_0^*,
\]
and take \(u=\varepsilon\) and \(\nu=1/2\). Then
\[
d_0=p+\binom{p}{2}\leq p^2,
\qquad
\E(u)=\E(\varepsilon)=0.
\]
Moreover, for every \(q\in S'\),
\[
\E(V_qu)
=\E\left(V_q\E(\varepsilon\mid X)\right)
=0.
\]
Thus, \(\Lambda_1=0\), while \(\Lambda_0\) in
\eqref{eq:level-m} is the quantity defined in
\eqref{eq:level-2}. To verify (T1), let \(q=(j,k)\in\mathcal P\). By Lemma~\ref{lem:orlicz}(b)-(c) and
\eqref{eq:tails-b}-\eqref{eq:tails-y},
\[
\begin{gathered}
\|V_q\|_{\psi_{1/2}}
\leq C_{1,1/2}K_X^2,
\qquad
\|V_q^2\|_{\psi_{1/2}}\leq K_X^4,\\
\|V_qY\|_{\psi_{1/2}}
\leq\|X_jX_k\|_{\psi_1}\|Y\|_{\psi_1}
\leq K_X^2K_Y,\\
\|V_q\varepsilon\|_{\psi_{1/2}}
\leq\|X_jX_k\|_{\psi_1}\|\varepsilon\|_{\psi_1}
\leq K_X^2K_\varepsilon.
\end{gathered}
\]
For every \(j\in[p]\), the corresponding main-effect bounds are
\[
\begin{gathered}
\|X_j\|_{\psi_{1/2}}\leq C_{2,1/2}K_X,
\qquad
\|X_j^2\|_{\psi_{1/2}}\leq C_{1,1/2}K_X^2,\\
\|X_jY\|_{\psi_{1/2}}
\leq C_{2,1}K_XK_Y,
\qquad
\|X_j\varepsilon\|_{\psi_{1/2}}
\leq C_{2,1}K_XK_\varepsilon.
\end{gathered}
\]
In addition,
\[
\|Y\|_{\psi_{1/2}}\leq C_{1,1/2}K_Y,
\qquad
\|\varepsilon\|_{\psi_{1/2}}\leq C_{1,1/2}K_\varepsilon,
\qquad
\|Y^2\|_{\psi_{1/2}}\leq K_Y^2.
\]
For (T2), for every \((q,q')\in S'\times S'\), the random variable
\(V_qV_{q'}\) is a
product of at most four sub-Gaussian variables. Hence, by Lemma~\ref{lem:orlicz}(b)-(d), there exists a constant
\(C_V>0\) such that
\[
\|V_qV_{q'}\|_{\psi_{1/2}}
\leq C_V\max(1,K_X)^4
\qquad \forall(q,q')\in S'\times S',
\]
which verifies (T2). Consequently, the same constant \(K\), depending only
on \((K_X,K_\varepsilon,K_\beta)\), may be used in (T1) and (T2).
Assumptions~\ref{ass:B3p}, \ref{ass:B4p}, and \ref{ass:B5p} verify (V),
(E), and (S), respectively. We take \(K_\xi=K_\beta\), because
\[
\max_{q\in S'}|\xi_q^*|
\leq|\beta_0^*|+\|\beta^*\|_1+\|\gamma^*\|_1
\leq K_\beta.
\]
Let \(\mathcal E\) be the event as defined by Theorem~\ref{th:main}, and
fix
\[
\omega\in\mathcal E\cap\{S_I\subseteq\widehat\Gamma\}.
\]
Then,
\[
\mathcal D(\omega)=[p]\cup\widehat\Gamma(\omega)
\]
satisfies
\[
S\subseteq\mathcal D(\omega)\subseteq\mathcal D_0.
\]
The \texttt{uniPairs} estimator at \(\omega\) is a minimizer of
\(F_{\mathcal D(\omega)}\) evaluated at \(\omega\). The conclusions of
Theorem~\ref{th:main} hold for the
\texttt{uniPairs} estimator at \(\omega\), even though
\(\mathcal D(\omega)\) is data dependent. It follows that the event on the
left-hand side of \eqref{eq:concl-2} is contained in \(\mathcal E^c\).
Since \(d_0\leq p^2\), Theorem~\ref{th:main} implies
\[
\begin{aligned}
&\Prob\left(
S_I\subseteq\widehat\Gamma\ \text{and}\
\left[
\widehat\xi_q\neq0
\ \text{for some }q\in([p]\cup\widehat\Gamma)\setminus S
\ \text{or}\
\max_{q\in S'}|\widehat\xi_q-\xi_q^*|>K_3\lambda
\right]\right)\\
&\qquad\leq
K_4p^2n\exp\left(-K_5\phi_{1/2}(n,\lambda)\right),
\end{aligned}
\]
which proves \eqref{eq:concl-2}.
On the event \(\{S_I\subseteq\widehat\Gamma\}\), we have
\[
\mathcal G^c
\subseteq
\left\{
\widehat\xi_q\neq0
\ \text{for some }q\in([p]\cup\widehat\Gamma)\setminus S
\ \text{or}\
\max_{q\in S'}|\widehat\xi_q-\xi_q^*|>K_3\lambda
\right\}.
\]
Therefore,
\[
\begin{aligned}
\Prob(\mathcal G)
&\geq\Prob\left(\mathcal G\cap
\{S_I\subseteq\widehat\Gamma\}\right)\\
&=\Prob(S_I\subseteq\widehat\Gamma)
-\Prob\left(\mathcal G^c\cap
\{S_I\subseteq\widehat\Gamma\}\right)\\
&\geq\Prob(S_I\subseteq\widehat\Gamma)
-K_4p^2n\exp\left(-K_5\phi_{1/2}(n,\lambda)\right),
\end{aligned}
\]
First suppose that \(\widehat\Gamma=\Gamma\) is deterministic and that
\(S_I\subseteq\Gamma\). Apply Theorem~\ref{th:main} with
\[
\mathcal D_0=\mathcal D=[p]\cup\Gamma.
\]
Then, \(d_0=p+|\Gamma|\), and  \(\Lambda_1=0\), Therefore, Theorem~\ref{th:main} implies
\[
\Prob(\mathcal G)
\geq1-K_4(p+|\Gamma|)n
\exp\left(-K_5\phi_{1/2}(n,\lambda)\right).
\]
If \(\widehat\Gamma\) is computed from observations independent of the
\(n\) fitting observations, condition on \(\widehat\Gamma\). The fitting
sample retains its original distribution under this conditioning, so the
same argument applies with
\(\Gamma=\widehat\Gamma\). Therefore,
\[
\Prob(\mathcal G\mid\widehat\Gamma)
\geq1-K_4(p+|\widehat\Gamma|)n
\exp\left(-K_5\phi_{1/2}(n,\lambda)\right).
\]
Finally, when
\[
\min_{q\in S}|\xi_q^*|>K_3\lambda,
\]
parts (i) and (iii) of Theorem~\ref{th:main} imply
\[
\supp(\widehat\xi)=S,
\qquad
\sign(\widehat\xi_q)=\sign(\xi_q^*)
\qquad \forall q\in S.
\]
This proves Corollary~\ref{cor:fixed-screen}.
\hfill\(\square\)

\paragraph{Remark.}
The union bound over the \(n+1\) regressions in the proof of
Lemma~\ref{lem:loo-coefs}(i) can be sharpened using 
\[
L_{n,i}[f]-\E(f)
=\frac{n}{n-1}\left(P_n[f]-\E(f)\right)
-\frac{f_i-\E(f)}{n-1}
\qquad \forall i\in[n].
\]
Consequently,
\[
\max_{i\in[n]}|L_{n,i}[f]-\E(f)|
\leq2|P_n[f]-\E(f)|
+\frac{\max_{i\in[n]}|f_i|+|\E(f)|}{n-1}.
\]
Using this bound in the proof shows that the failure probability is bounded by
\[
K_4d_0\left(
\exp(-K_5n\lambda^2)
+n\exp\left(-K_5(n\lambda)^\nu\right)
\right)
+K_4\exp(-K_5\sqrt n).
\]

\section*{Appendix B: Generalization to the binomial GLM and the Cox proportional hazards model}
\label{sec:glm-generalization}
We provide implementations that extend uniPairs and uniPairs-2stage to the Binomial generalized linear model with logit link (logistic regression) and the Cox proportional hazards model. The algorithmic structure remains the same with two important changes. First, we use the approximation \citet{rad2018scalable} of the LOO prediction used in UniLasso since no exact formula exists. Second, the t-tests performed in the Triplet-Scan are changed to likelihood ratio tests with unpenalized GLM fitting instead of OLS.
  
The model is written in terms of the linear predictor
\[
\boldsymbol\eta = \mathbf X\beta + \text{offset} 
\]
with inverse link \(\mu\). Each model provides a fitted linear predictor \(\widehat{\boldsymbol\eta}\in \mathbb{R}^n\), a log-likelihood \(\ell(\boldsymbol\eta;y)\) (partial likelihood for Cox), together with its gradient and Hessian evaluated at \(\widehat{\boldsymbol\eta}\):
\[
\widehat g = g(\widehat{\boldsymbol\eta}) \;=\; \left.\frac{\partial \ell(\boldsymbol\eta;y)}{\partial \boldsymbol\eta}\right|_{\boldsymbol\eta=\widehat{\boldsymbol\eta}} \in \mathbb{R}^n
\quad \text{and} \quad 
\widehat H =  H(\widehat{\boldsymbol\eta}) \;=\; \left.\frac{\partial^2 \ell(\boldsymbol\eta;y)}{\partial \boldsymbol\eta \partial\boldsymbol\eta^\top}\right|_{\boldsymbol\eta=\widehat{\boldsymbol\eta}}  \in \mathbb{R}^{n \times n}
\]
When fitting Univariate models inside UniLasso both in \texttt{uniPairs} and \texttt{uniPairs-2stage}, for each \(j\in[p]\) and each \((j,k)\in\widehat\Gamma\) (the second case only happens for \texttt{uniPairs}), we fit the Univariate GLM models
\[
\boldsymbol\eta_j = \beta_{0,j}^{\mathrm{uni}}\mathbf 1 + \beta_{1,j}^{\mathrm{uni}}\widetilde{\mathbf X}_j,
\quad \text{and} \quad
\boldsymbol\eta_{jk} = \beta_{0,jk}^{\mathrm{uni}}\mathbf 1 + \beta_{1,jk}^{\mathrm{uni}}(\widetilde{\mathbf X}_j\odot\widetilde{\mathbf X}_k)
\]
Exact LOO formulas are unavailable for GLMs, so we use the approximation \citet{rad2018scalable}:
\[
\widehat\eta_i^{(-i)}
\approx
\widehat\eta_i+
\frac{\widehat h_i}{1-\widehat h_i}
\frac{\widehat g_i}{\widehat H_i},
\]
where \(d_{ij}=(1,\widetilde X_{ij})^\top\) for the Binomial model and \(d_{ij}=\widetilde X_{ij}\) for the Cox model, and
\[
\widehat h_i
=\widehat H_i d_{ij}^\top
\left(\sum_{r=1}^n\widehat H_r d_{rj}d_{rj}^\top\right)^\dagger
d_{ij}.
\]

For each pair \((j,k)\in\mathcal{P}\), TripletScan fits the unpenalized GLM
\[
\boldsymbol\eta_{jk}
= \beta_{0,jk}\mathbf 1 + \beta_{j,jk}\widetilde{\mathbf X}_j
+ \beta_{k,jk}\widetilde{\mathbf X}_k
+ \beta_{jk,jk}(\widetilde{\mathbf X}_j\odot\widetilde{\mathbf X}_k)
\]
where the likelihood corresponds to the Binomial or Cox model. To test the interaction coefficient \(\beta_{jk,jk}\), we use a likelihood ratio test:
\[
\Lambda_{jk}
= 2\big(\ell_{\mathrm{full}} - \ell_{\mathrm{null}}\big)
\quad \text{and} \quad
p_{jk} = \mathbb{P}\!\big(\chi^2_1 \ge \Lambda_{jk}\big)
\]
where the null model excludes the interaction term.

\section*{Appendix C: Full simulation results}
\label{sec:full-simulations}

The next figures complement the ones shown in Section~\ref{sec:simulation}. 
They show the main-effects model size, the interactions model size, Train $R^2$, main-effects FDR, interactions FDR, main-effects coverage, and interactions coverage, as well as the Jaccard index between the predicted active sets of the two raw-gap procedures and of the two normalized-gap variants. 

\begin{figure*}[!htbp]
    \centering
    \includegraphics[width=\textwidth]{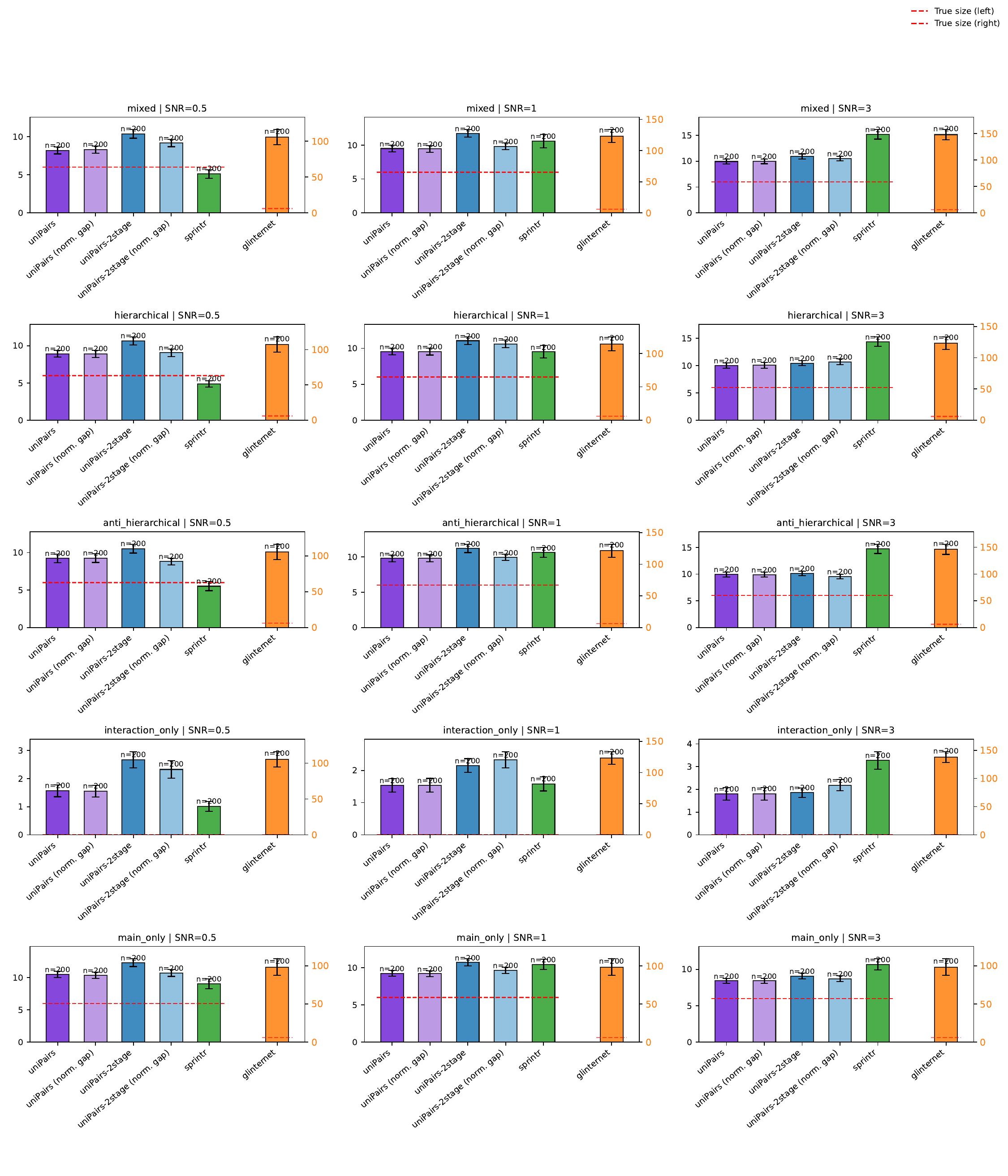}
    \caption{\em 
        Main-effects model size for $(n,p)=(300,400)$ aggregated over $\rho$.
        Each bar shows mean $\pm$ one standard error across 200 replicates.
        Rows correspond to structures and columns to SNR levels $(0.5,1,3)$.
        The red dashed line marks the true number of active main effects. \texttt{Glinternet} is plotted against the right y-axis while \texttt{uniPairs}, \texttt{uniPairs-2stage} and \texttt{Sprinter} use the left y-axis.
    }
    \label{fig:model_size_main_rho05}
\end{figure*}
In Figure~\ref{fig:model_size_main_rho05}, we see that \texttt{uniPairs} and \texttt{uniPairs-2stage} produce main-effects model sizes that are close to the true number of active main effects. Their selection tends to be slightly above the truth. \texttt{Sprinter} tends to select more main-effects at high SNR levels but overall stays close to the two variants. In contrast, \texttt{Glinternet} selects a very large number of main-effects in every scenario, often exceeding one hundred even when only six main effects are truly active. 

\begin{figure*}[!htbp]
    \centering
    \includegraphics[width=\textwidth]{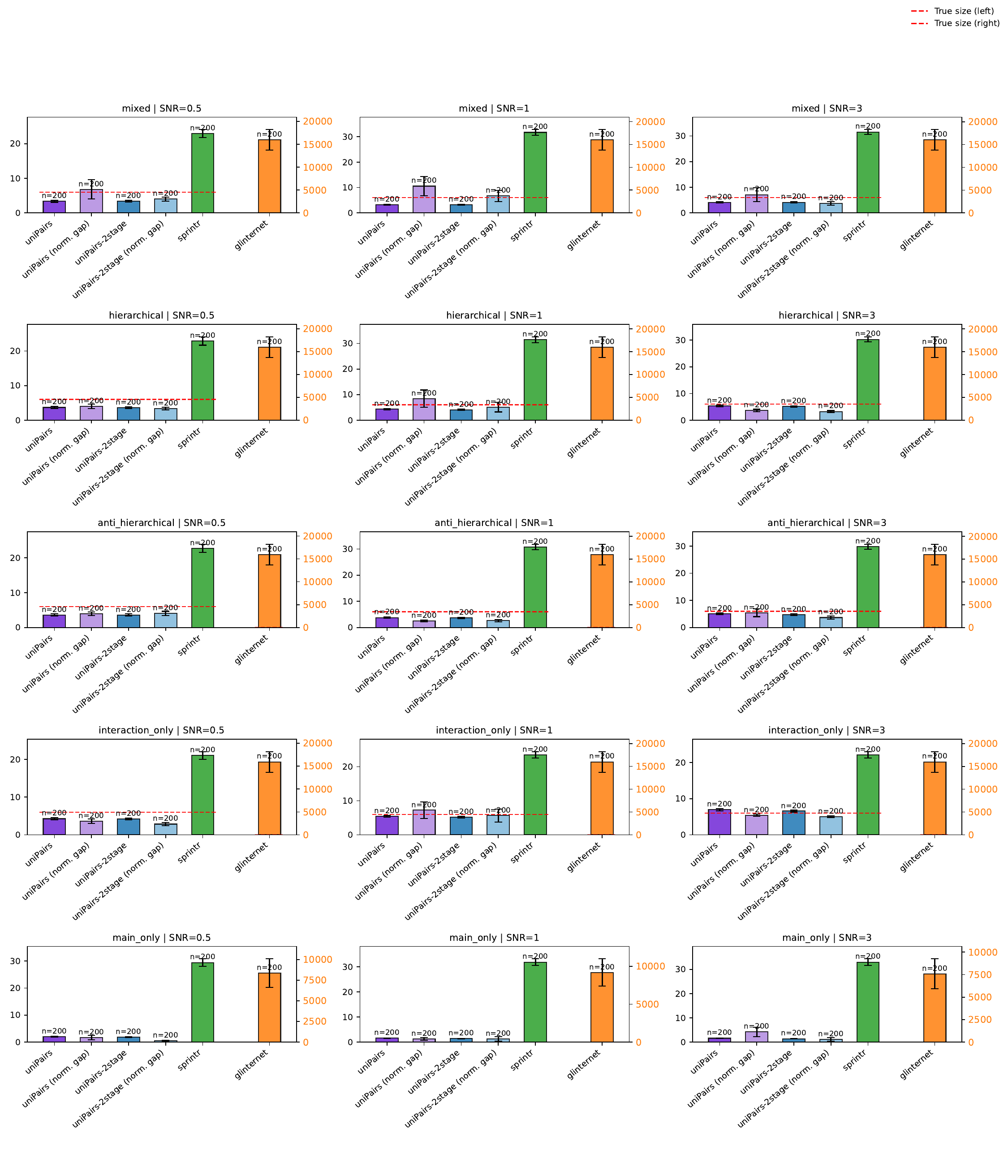}
    \caption{\em  
        Interactions model size for $(n,p)=(300,400)$ aggregated over $\rho$.
        Each bar shows mean $\pm$ one standard error across 200 replicates.
        Rows correspond to structures and columns to SNR levels $(0.5,1,3)$.
        The red dashed line marks the true number of active interactions. \texttt{Glinternet} is plotted against the right y-axis while \texttt{uniPairs}, \texttt{uniPairs-2stage} and \texttt{Sprinter} use the left y-axis.
    }
    \label{fig:model_size_interactions_rho05}
\end{figure*}

In Figure~\ref{fig:model_size_interactions_rho05}, we see that the \texttt{uniPairs} procedures consistently select very few interaction terms, typically very close to the true number of active pairs. This holds across all structures and SNR levels and their behavior remains very stable. The selected interaction counts for \texttt{Sprinter} are often several times larger than the truth. \texttt{Glinternet} selects extremely large numbers of interactions, in the order of thousands, irrespective of structure and SNR level. Even in the main-effects only case, where the true number of interactions is zero, \texttt{Glinternet} produces a large interaction set. This explains its high coverage in Figure~\ref{fig:coverage_interactions_rho05} which is achieved through aggressive over-selection. 

\begin{figure*}[!htbp]
    \centering
    \includegraphics[width=\textwidth]{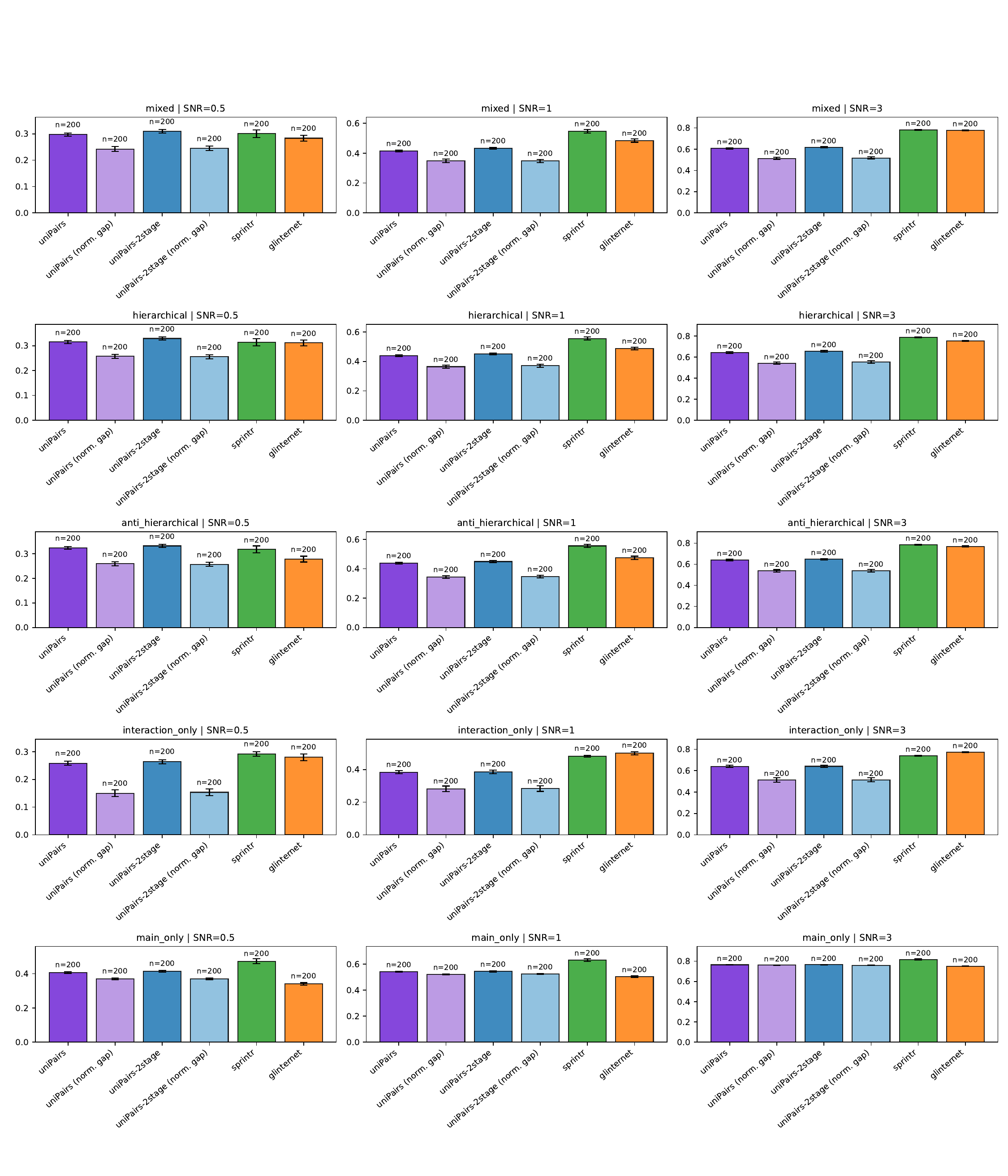}
    \caption{\em 
        Train $R^2$ for $(n,p)=(300,400)$ aggregated over $\rho$.
        Each bar shows mean $\pm$ one standard error across 200 replicates.
        Rows correspond to structures and columns to SNR levels $(0.5,1,3)$.
    }
    \label{fig:train_r2_rho05}
\end{figure*}

In Figure~\ref{fig:train_r2_rho05}, we see that Train $R^2$ increases with SNR for all methods, as expected. The \texttt{uniPairs} procedures obtain slightly lower Train $R^2$ than \texttt{Glinternet} and \texttt{Sprinter}, especially at high SNR levels, but with much sparser models as seen in Figure~\ref{fig:model_size_both_rho05}.

\begin{figure*}[!htbp]
    \centering
    \includegraphics[width=\textwidth]{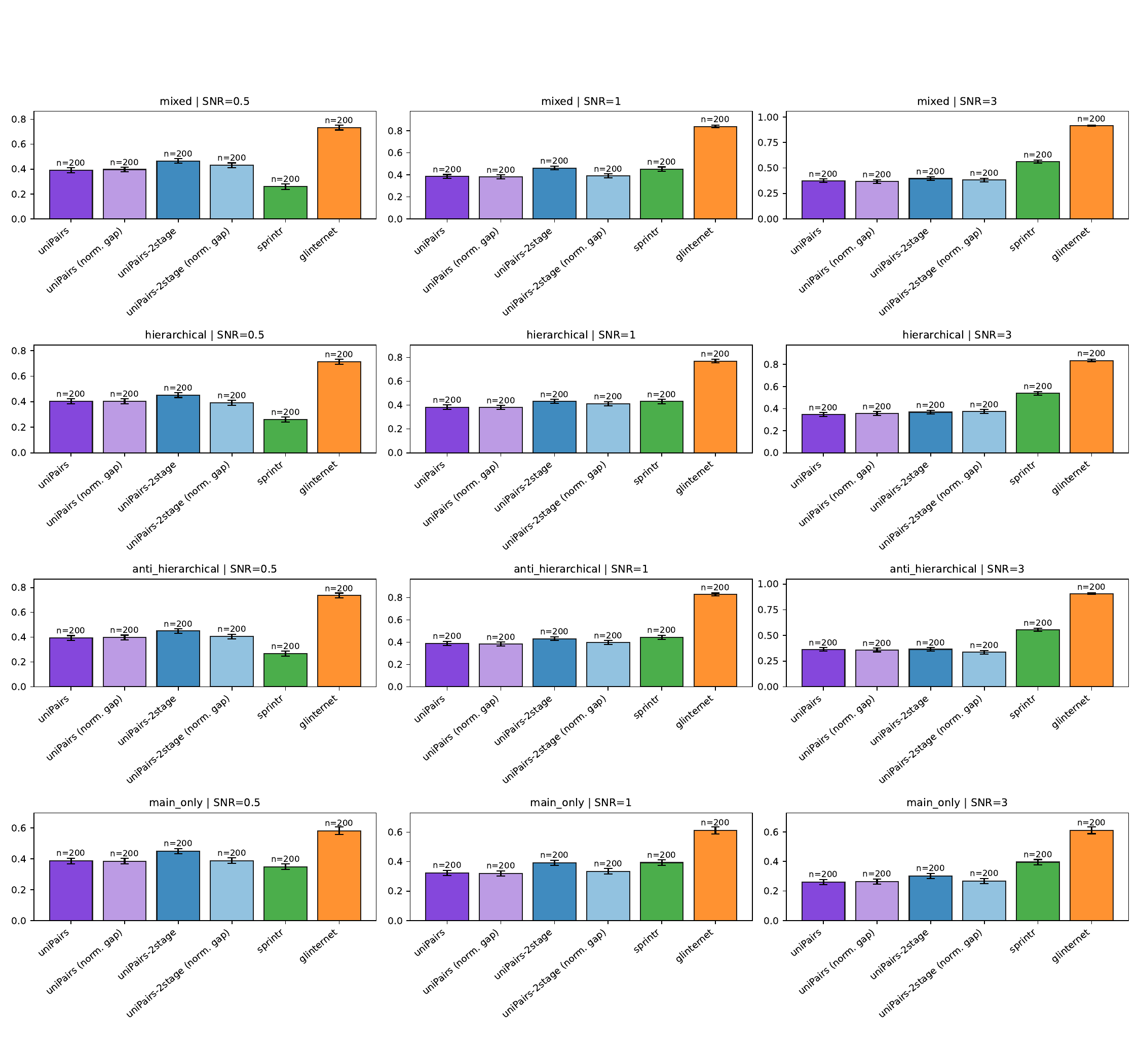}
    \caption{\em 
        Main-effects false discovery rate (FDR) for $(n,p)=(300,400)$ aggregated over $\rho$. Each bar shows mean $\pm$ one standard error across 200 replicates.
        Rows correspond to structures and columns to SNR levels $(0.5,1,3)$.
    }
    \label{fig:fdp_main_rho05}
\end{figure*}

In Figure~\ref{fig:fdp_main_rho05}, we see that \texttt{Glinternet} consistently exhibits the highest main-effects FDR, independent of structure and SNR level. This is consistent with its over-selection of main-effects as seen in Figure~\ref{fig:model_size_main_rho05} and shows that many of its selected main-effects are false positives. \texttt{uniPairs} and \texttt{uniPairs-2stage} have substantially lower FDR than \texttt{Glinternet}, but their comparison with \texttt{Sprinter} varies with SNR. At SNR $=0.5$, for example, \texttt{Sprinter}'s main-effect FDR is about $0.26$, compared with about $0.40$ for \texttt{uniPairs}.

\begin{figure*}[!htbp]
    \centering
    \includegraphics[width=\textwidth]{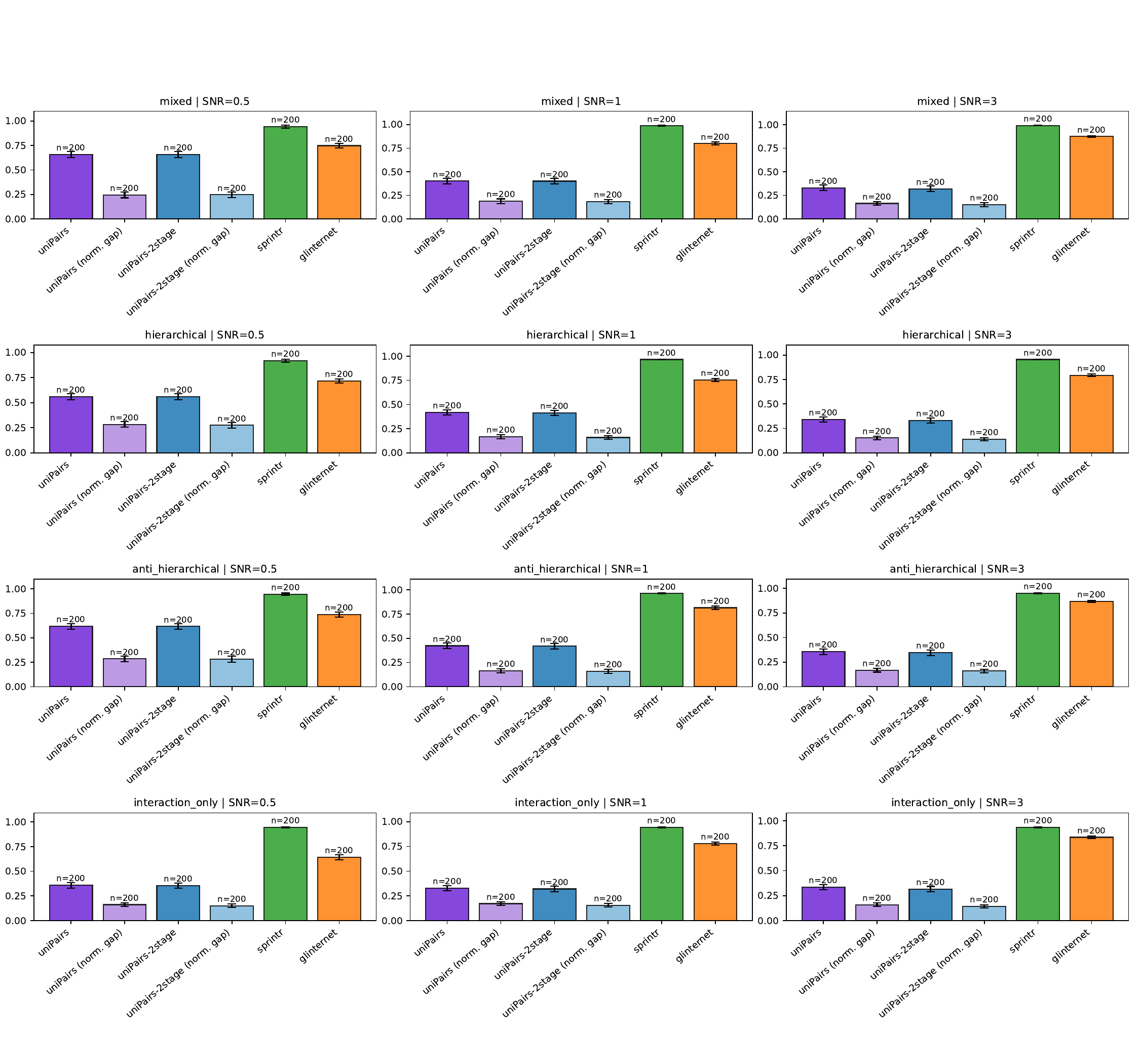}
    \caption{\em 
        Interaction false discovery rate (FDR) for $(n,p)=(300,400)$ aggregated over $\rho$. Each bar shows mean $\pm$ one standard error across 200 replicates.
        Rows correspond to structures and columns to SNR levels $(0.5,1,3)$.
    }
    \label{fig:fdp_interactions_rho05}
\end{figure*}

In Figure~\ref{fig:fdp_interactions_rho05}, we see that both \texttt{Sprinter} and \texttt{Glinternet} exhibit high interactions FDR, which is consistent at least for \texttt{Glinternet} with its over-selection of interactions as seen in Figure~\ref{fig:model_size_interactions_rho05}. In contrast, the \texttt{uniPairs} procedures maintain a lower interactions FDR, which tends to decrease slightly as SNR increases, and the normalized-gap variants reduce it further. Across all simulations, Table~\ref{glob-avg} shows a reduction of roughly two thirds for both procedures. This pattern holds across structures and SNR levels, showing that the \texttt{uniPairs} procedures identify interactions more conservatively, leading to lower FDR.

\begin{figure*}[!htbp]
    \centering
    \includegraphics[width=\textwidth]{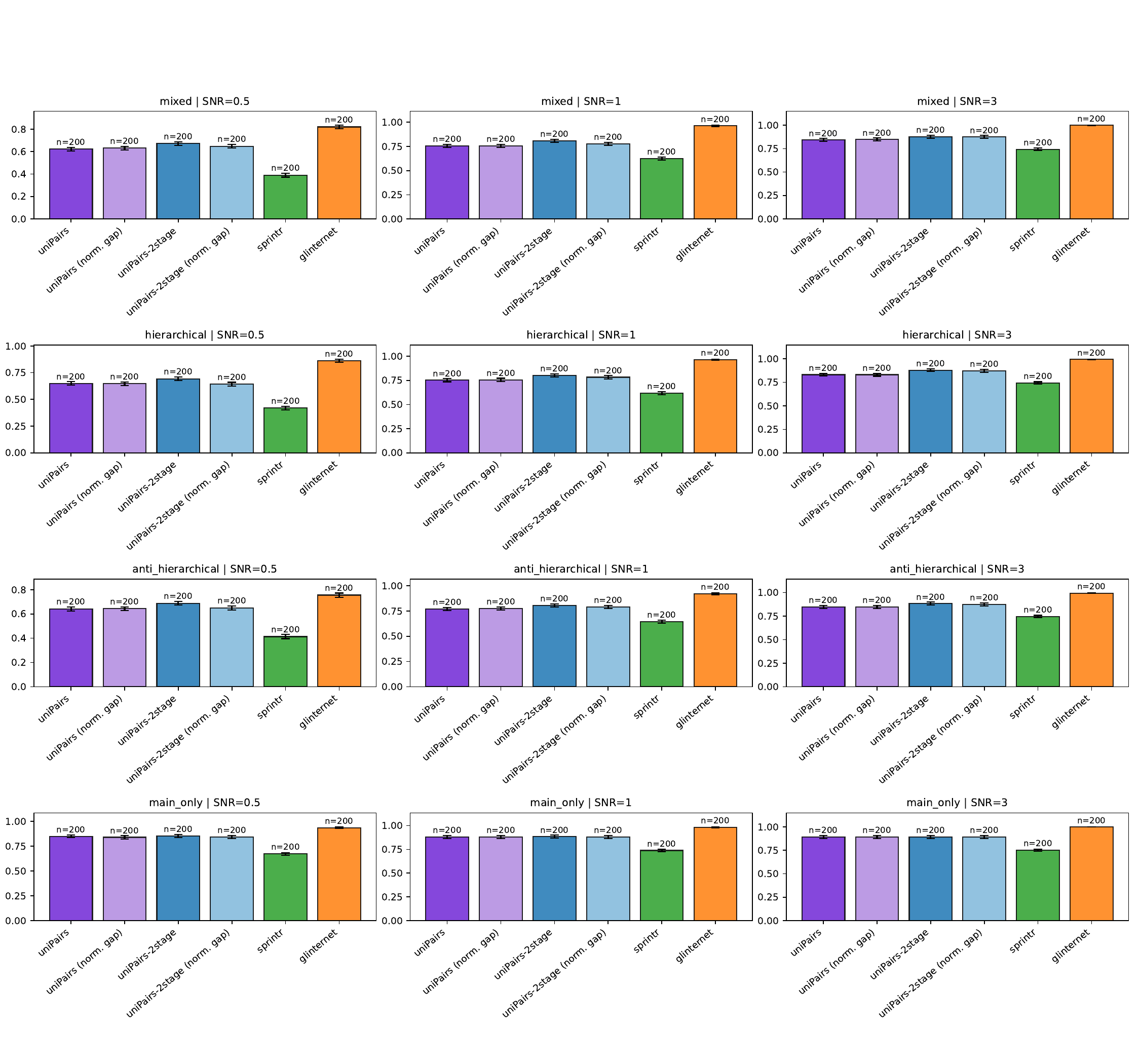}
    \caption{\em 
        Main-effects coverage for $(n,p)=(300,400)$ aggregated over $\rho$. Each bar shows mean $\pm$ one standard error across 200 replicates.
        Rows correspond to structures and columns to SNR levels $(0.5,1,3)$.
    }
    \label{fig:coverage_main_rho05}
\end{figure*}

In Figure~\ref{fig:coverage_main_rho05}, \texttt{Glinternet} consistently achieves the highest main-effects coverage, but at the expense of many false positives as seen in Figure~\ref{fig:model_size_main_rho05}. \texttt{uniPairs} and \texttt{uniPairs-2stage} maintain good coverage, typically around $0.6$ at low SNR, and around $0.8$ to $0.87$ at high SNR. \texttt{Sprinter} has considerably lower coverage than the other methods, particulary at low SNR.  

\begin{figure*}[!htbp]
    \centering
    \includegraphics[width=\textwidth]{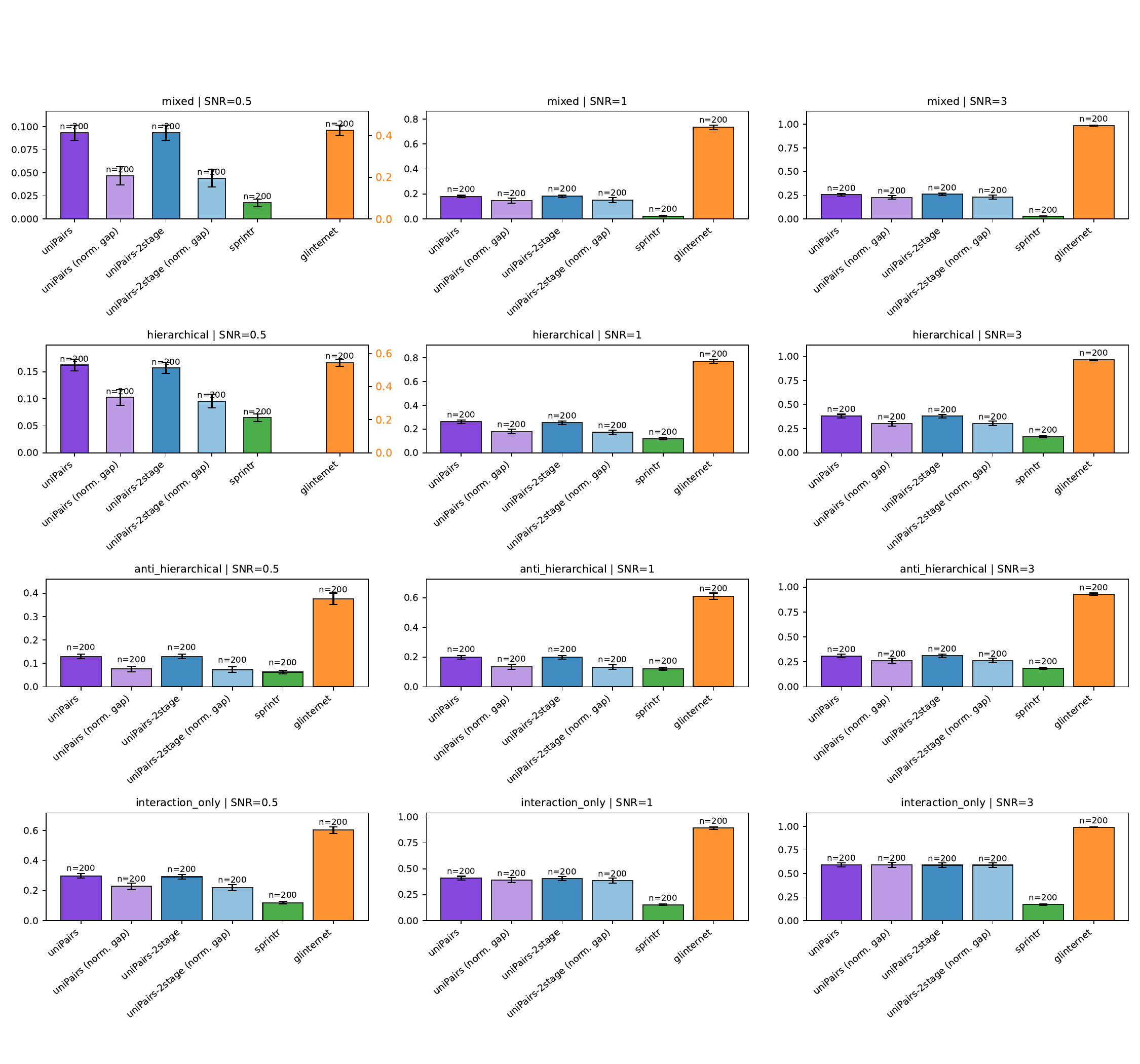}
    \caption{\em 
        Interaction coverage for $(n,p)=(300,400)$ aggregated over $\rho$. Each bar shows mean $\pm$ one standard error across 200 replicates.
        Rows correspond to structures and columns to SNR levels $(0.5,1,3)$.
    }
    \label{fig:coverage_interactions_rho05}
\end{figure*}

In Figure~\ref{fig:coverage_interactions_rho05}, we see that \texttt{Glinternet} achieves the highest interaction coverage in every setting, but at the cost of extremely high FDR and very large model size as seen in Figures~\ref{fig:fdp_interactions_rho05}~\ref{fig:model_size_interactions_rho05}. The \texttt{uniPairs} procedures achieve moderate interactions coverage, but with far fewer false positives as seen in Figure~\ref{fig:fdp_interactions_rho05}. The normalized-gap variants give up a little of that coverage for a further reduction in false positives. \texttt{Sprinter} shows  the lowest overall interactions coverage.

\begin{figure*}[!htbp]
    \centering
    \includegraphics[width=\linewidth]{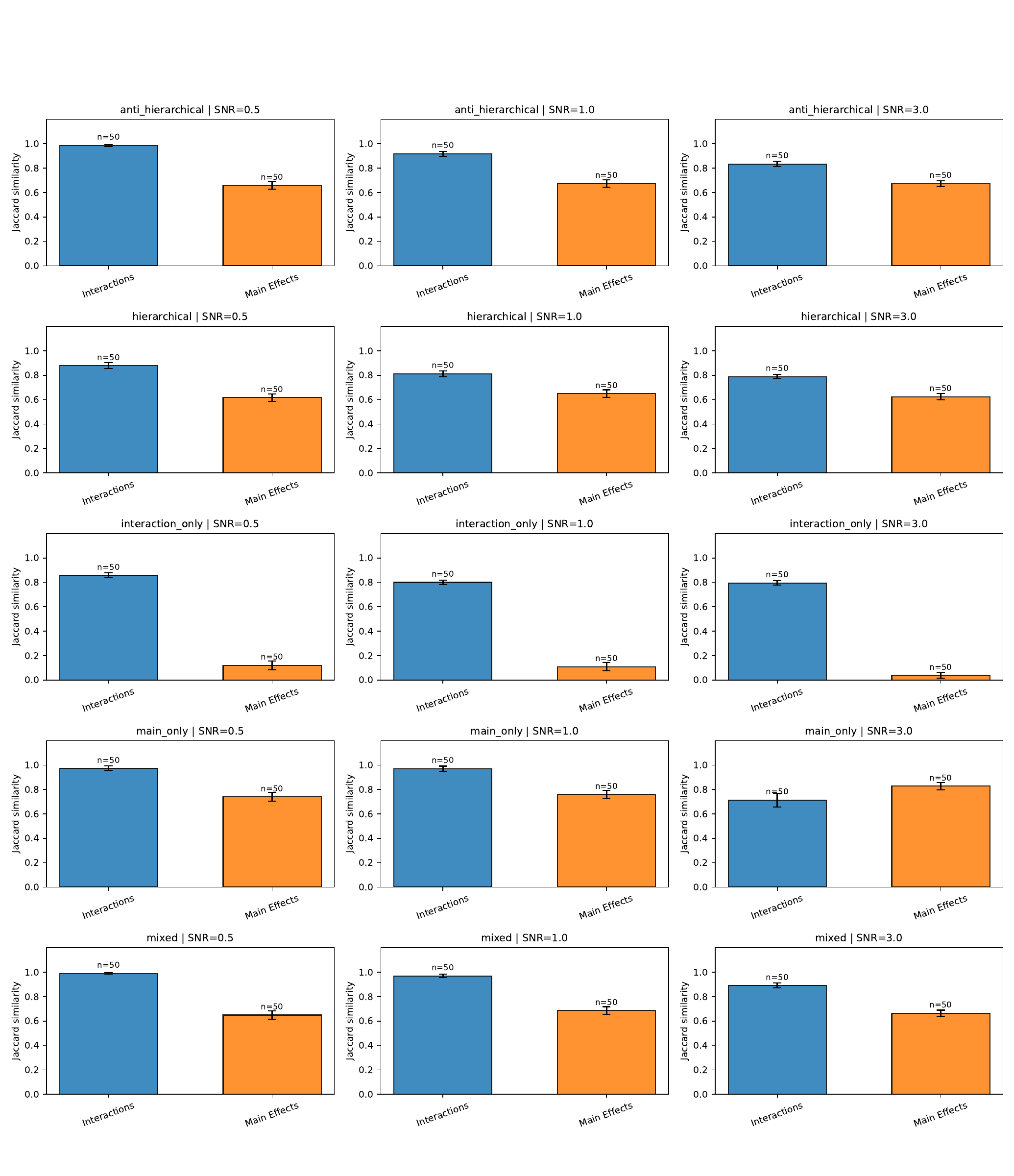}
    \caption{\em Jaccard index \(\frac{|A \cap B|}{|A \cup B|}\) between the predicted models of \texttt{uniPairs} and \texttt{uniPairs-2stage} for \(n=300, p=400\) and \(\rho=0.8\). When \(A\cup B=\varnothing\), we set the index to \(0\). Each bar shows mean $\pm$ one standard error across 50 replicates. Rows correspond to structures and columns to SNR levels $(0.5,1,3)$.}
    \label{fig:jacc}
\end{figure*}
In Figure~\ref{fig:jacc}, the interaction sets selected by \texttt{uniPairs} and \texttt{uniPairs-2stage} exhibit high similarity across structures and SNR levels. Main-effect agreement is lower, particularly in \texttt{interaction\_only}.

\begin{figure*}[!htbp]
    \centering
    \includegraphics[width=\linewidth]{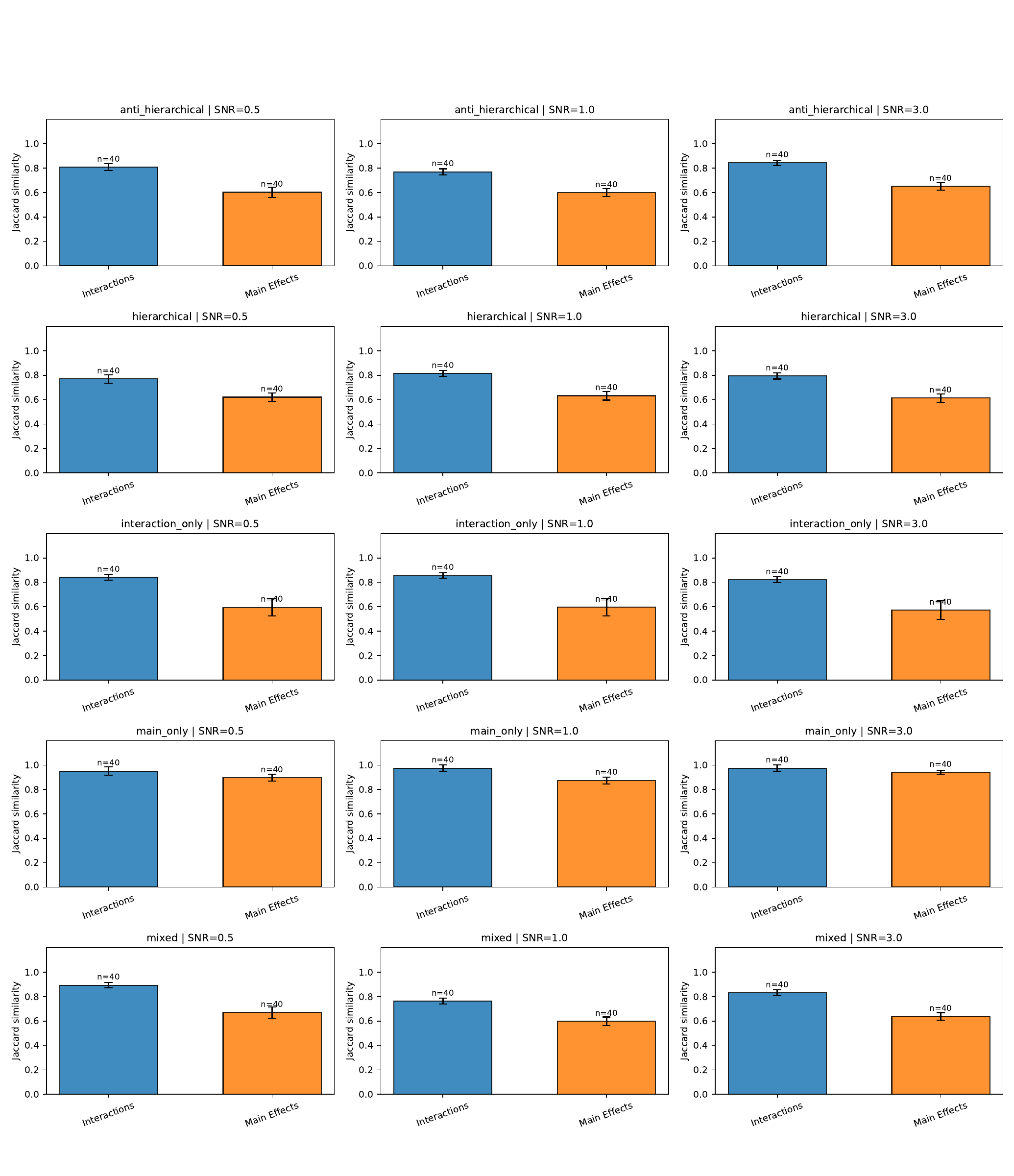}
    \caption{\em Jaccard index between the predicted models of the two normalized-gap variants, \texttt{uniPairs (norm.\ gap)} and \texttt{uniPairs-2stage (norm.\ gap)}, for \(n=300, p=400\) and \(\rho=0.8\). The empty-set convention is as in Figure~\ref{fig:jacc}. Each bar shows mean $\pm$ one standard error across 40 replicates. Rows and columns are as in Figure~\ref{fig:jacc}.}
    \label{fig:jacc-ng}
\end{figure*}
Figure~\ref{fig:jacc-ng} repeats the comparison for the two normalized-gap variants. Interaction-set agreement remains high in both figures. 
Averaged over all structures and SNR levels, main-effect agreement is about \(0.56\) for the raw-gap pair and \(0.67\) for the normalized-gap pair with the largest difference occuring in \texttt{interaction\_only}.

\end{document}